\documentclass[11pt]{article}

\usepackage[margin=1in]{geometry}
\usepackage{amsmath,amssymb,amsthm,mathtools}
\usepackage{booktabs,multirow,array,tabularx}
\usepackage{graphicx,subcaption}
\usepackage{placeins}
\usepackage{float}
\usepackage[
  backend=biber,
  style=authoryear,
  natbib=true,
  maxcitenames=2,
  maxbibnames=99,
  giveninits=true,
  uniquename=init,
  doi=true,
  url=true,
  isbn=false
]{biblatex}
\AtBeginBibliography{\small}
\usepackage{microtype}
\usepackage{enumitem}
\usepackage{xcolor}
\usepackage{hyperref}
\usepackage{bookmark}
\usepackage{cleveref}
\graphicspath{{figures/}}
\hypersetup{
  bookmarksopen=true,
  bookmarksopenlevel=2,
  bookmarksnumbered=true,
  colorlinks=true,
  linkcolor=blue,
  citecolor=blue,
  urlcolor=blue,
  pdftitle={From Value Bounds to Policy-Distance and Active-Face Certificates: Same-Grid Duality for Constrained Dynamic Portfolios},
  pdfauthor={Jeonggyu Huh},
  pdfsubject={Policy-distance and active-face certification for constrained dynamic portfolios},
  pdfkeywords={policy verification, policy-distance certification, active-face certification, same-grid duality, Doob compensation, multiplier sensitivity, constrained portfolio optimization}
}

\newtheorem{assumption}{Assumption}
\newtheorem{theorem}{Theorem}
\newtheorem{proposition}{Proposition}
\newtheorem{corollary}{Corollary}
\newtheorem{lemma}{Lemma}
\theoremstyle{remark}
\newtheorem{remark}{Remark}

\newcommand{\E}{\mathbb{E}}
\newcommand{\R}{\mathbb{R}}
\newcommand{\F}{\mathcal{F}}

\newcommand{\cC}{\mathcal{C}}
\newcommand{\cK}{\mathcal{K}}
\newcommand{\Curv}{\mathsf{C}}
\newcommand{\Proj}{\operatorname{Proj}}
\newcommand{\CE}{\operatorname{CE}}
\newcommand{\SE}{\operatorname{SE}}
\newcommand{\Rad}{\operatorname{Rad}}
\newcommand{\argmax}{\operatorname*{arg\,max}}
\newcommand{\argmin}{\operatorname*{arg\,min}}
\newcommand{\norm}[1]{\left\lVert #1\right\rVert}
\newcommand{\prooflocation}[1]{\par\noindent\emph{Proof.} #1\par}

\title{\textbf{From Value Bounds to Policy-Distance and Active-Face Certificates}\\[0.35em]
\large Same-Grid Duality for Constrained Dynamic Portfolios}
\author{Jeonggyu Huh\thanks{Corresponding author. Email: \href{mailto:jghuh@skku.edu}{jghuh@skku.edu}}\\
\small Department of Mathematics, Sungkyunkwan University, Suwon, Republic of Korea}
\date{August 4, 2026}

\begin{document}
\maketitle

\begin{abstract}
Neural and numerical policy solvers can produce feasible controls even when the optimal rule and its binding constraints are unavailable. A primal--dual bracket certifies value loss, but it does not locate the optimal policy or explain which constraints genuinely bind. We show that, on the same declared simulation grid, one bracket can support both conclusions. For polyhedral controls, an exact conditional budget identity rewrites the residual as a pathwise-nonnegative terminal Fenchel defect plus date-by-constraint complementary-slackness terms; a canonical Doob compensation removes the budget martingale that obscures small residuals. Bellman-primitive curvature conditions then yield an occupancy-weighted policy region with the sharp $O(\!\sqrt G)$ radius, while a paired constraint relaxation lower-bounds the optimal multiplier and certifies a binding face. A finite-sample resolution theorem quantifies the path budget needed to certify a target policy tolerance or face. Locked one- and two-asset audits cover every external policy error and make no false face declaration. An exact-wrapper stress test remains tight through $50$ assets, while a separate state-dependent pilot identifies learned-dual tightness as the high-dimensional bottleneck. Reference solutions enter only after certification.
\end{abstract}


\section{Introduction}

Neural, regression, and simulation-based control methods increasingly return feasible feedback policies in problems for which neither the optimal policy nor its binding constraints can be computed. Their objective values do not answer two natural validation questions: how far can the optimal decision rule be from the deployed rule, and is a boundary decision economically genuine or merely an approximation artifact? These questions are especially important for constrained portfolios, where a zero holding can represent either an optimal exclusion or numerical clipping, and a binding leverage cap changes the admissible comparative-statics directions.

Classical portfolio duality supplies a value bracket. An admissible policy gives a primal lower value, a dual-feasible artificial market gives an upper value, and their difference bounds utility loss \citep{HePearson1991,CvitanicKaratzas1992,HaughKoganWang2006}. This paper asks what the same bracket says about the decision rule:
\begin{equation}
\label{eq:intro-question}
J(\pi)\le V^\star\le D
\quad\Longrightarrow\quad
\begin{cases}
\text{where can }\pi^\star\text{ lie in policy space?}\\
\text{which tested constraints must bind at }\pi^\star\text{?}
\end{cases}
\end{equation}
\phantomsection\label{def:value-residuals}
Write
\[
G(\pi):=V^\star-J(\pi),
\qquad
\Delta_D(y,\mu):=D(y,\mu)-V^\star,
\]
and
\[
\mathcal R(y,\mu;\pi)
:=D(y,\mu)-J(\pi)
=G(\pi)+\Delta_D(y,\mu).
\]
Weak duality gives $G,\Delta_D,\mathcal R\ge0$. Our two inferential outputs are
\begin{equation}
\label{eq:intro-policy-pipeline}
\begin{aligned}
\text{residual upper bound}
&\Longrightarrow
\text{occupancy-weighted policy region},\\
\text{paired constraint relaxation}
&\Longrightarrow
\text{multiplier lower bound}
\Longrightarrow
\text{binding-face certificate}.
\end{aligned}
\end{equation}
The first output localizes the unknown optimizer around the deployed policy; the second distinguishes genuine optimal boundaries from numerical clipping. A positive face certificate is operationally meaningful: it proves that an exclusion or leverage boundary is a property of the optimal problem, reduces the locally admissible tangent space for sensitivity analysis, and supplies an auditable explanation of a deployed boundary action. Failure to certify remains unresolved rather than being classified as inactivity.

The statistical obstacle is precision. Under quadratic policy growth, scalar value information yields a directional radius of order $O(\sqrt G)$, and this exponent is sharp. Reducing a policy radius by ten therefore requires roughly a hundredfold reduction in the residual bound. Ordinary primal--dual differencing can be adequate for value evaluation and still be unusable at policy scale. We derive an exact identity on the declared log-Euler grid that writes the residual as a terminal Fenchel defect plus pathwise-nonnegative date-by-constraint complementary-slackness terms. The predictable part is the Doob compensator of the budget process. Removing the budget martingale allows the residual variance to collapse near an exact pair and attributes the remaining loss by date and constraint.

The policy transfer is made constructive in two steps. First, Bellman-primitive strong-concavity conditions yield computable lower bounds on action curvature. In the portfolio models, conditional dual restarts supply the continuation-value information needed for these bounds, while nonnegative Gaussian square or covariance terms sharpen them. Second, a finite-sample resolution theorem links residual variance, curvature-weight uncertainty, and multiplier margins to the paths required for a target policy tolerance or active-face declaration. Compensation therefore reduces the sampling burden for policy-level statements, not merely the variance of a value estimator.

\Cref{fig:flagship} summarizes these external audits. A solved five-asset Merton benchmark validates residual precision, coordinate radii, exact exclusions, and leakage bounds. In a one-asset predictable-return problem, all ten locked policy radii cover independently refined errors; a difficult seed is correctly withheld at a predeclared tolerance. A two-asset incomplete market contains two no-short faces and a shared borrowing cap, and all covariance and directional radii cover. The face procedure certifies subsets of externally active faces with no false declaration. A separate locked Merton stress test covers all $15$ dimension--market policy errors at $d=5,20,50$ with a maximum radius-to-error ratio of $1.023$ under an exact dual wrapper. A state-dependent $d=20,50$ pilot remains valid but becomes conservative, locating the practical high-dimensional bottleneck in learning a tight feasible dual rather than in curvature estimation or the residual-to-policy transfer.

\begin{figure}[!htbp]
\centering
\includegraphics[width=\textwidth]{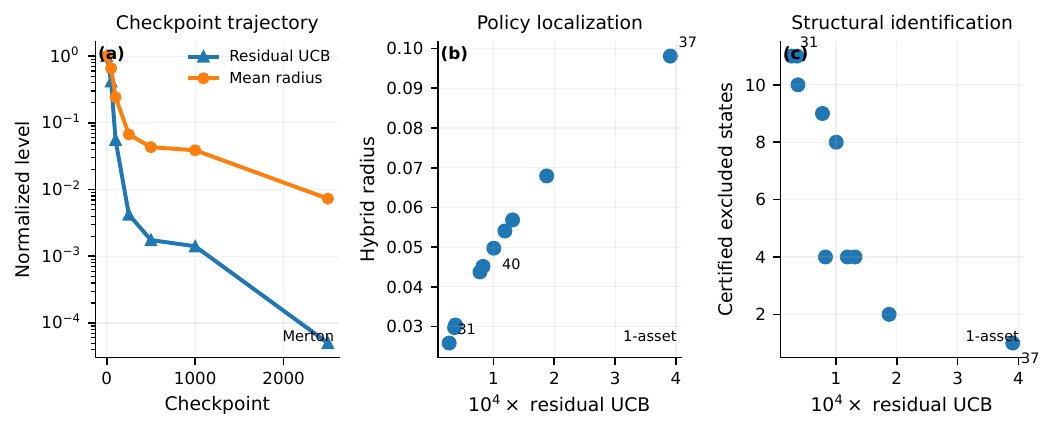}
\caption{\textbf{Residual precision translates into policy and structural certification.} Panel~(a) follows the solved Merton checkpoint trajectory. Panels~(b)--(c) use the locked one-asset ten-seed audit. Smaller residual bounds coincide with tighter policy localization and more certified excluded states; analytical and dynamic-program references enter only after certification.}
\label{fig:flagship}
\end{figure}

\paragraph{Relation to dual and pathwise bounds.}
Portfolio duality and artificial-market constructions provide upper values under trading constraints \citep{HePearson1991,CvitanicKaratzas1992,Cuoco1997,SchroderSkiadas2003,HaughKoganWang2006}. Martingale dual simulation, information relaxation, and pathwise optimization provide related value bounds and penalty constructions in stopping and stochastic control \citep{Rogers2002,HaughKogan2004,AndersenBroadie2004,Rogers2007,BrownSmithSun2010,DesaiFariasMoallemi2012Pathwise}. Prior portfolio work tightens brackets through policy-induced artificial markets, pathwise value-derivative estimation, tractable dual-control families, and simulation \citep{HaughJain2007,BickKraftMunk2013,MaLiZheng2020}. Doob penalties and variance collapse under ideal penalties are therefore not new. The contribution here is the same-grid Fenchel--complementarity representation and its conversion into finite-sample policy-location and binding-constraint statements.

\paragraph{Relation to approximate dynamic programming and policy evaluation.}
Approximate linear programming and dynamic-programming analyses translate Bellman or value-function approximation errors into performance guarantees for induced policies \citep{deFariasVanRoy2003,DesaiFariasMoallemi2012SALP,Munos2003,Munos2007,FarahmandSzepesvariMunos2010}. High-confidence off-policy evaluation instead gives one-sided value guarantees for a target policy from logged data \citep{ThomasTheocharousGhavamzadeh2015}. These literatures address value or policy-loss guarantees generated by approximate value functions or off-policy samples. Our input is a frozen feasible same-grid primal--dual pair; our outputs are a confidence region around the unknown optimal policy and a subset of polyhedral rows that must bind. Recent regime-dependent utility work also obtains an integrated squared control-error bound from a value gap under strong concavity \citep{Gandhi2026}. We add sharpness, a constructive curvature route, finite-sample resolution, and active-face inference.

\paragraph{Contributions.}
\begin{enumerate}[leftmargin=1.6em]
\item \textbf{An exact residual representation with operational precision.}
For polyhedral controls, the same-grid residual is the expectation of a pathwise-nonnegative Fenchel--complementarity quantity. Its canonical Doob compensation removes the budget martingale and attributes the residual by date and constraint.

\item \textbf{Policy and binding-constraint inference.}
Bellman-primitive curvature conditions imply quadratic policy growth. A residual upper bound then yields the sharp $O(\sqrt G)$ policy region, while a paired relaxation statistic lower-bounds an optimal multiplier and certifies a binding face. A finite-sample theorem quantifies the resolution margin and path complexity of both conclusions.

\item \textbf{Locked validation and scalability diagnosis.}
One- and two-asset state-dependent audits cover every reported external policy error and make no false face declaration. An exact-wrapper replication remains tight through $50$ assets, whereas a learned state-dependent pilot identifies dual approximation---not the policy transfer---as the principal high-dimensional bottleneck.
\end{enumerate}

\paragraph{Scope.}
The certificates concern the declared grid and the deployed occupancy law. Informative conclusions require a tight feasible dual, positive Bellman curvature, and sufficient relaxation sensitivity. The paper develops these objects for constrained terminal-wealth portfolios with polyhedral controls; transaction costs, consumption, and generic high-dimensional state-dependent dual learning require additional structure. All proofs, supporting derivations, implementation details, numerical audits, and robustness diagnostics are in the Electronic Companion.

\section{Same-grid market and polyhedral controls}
\label{sec:setup}

Fix a horizon $T>0$, a grid $t_k=kh$, $k=0,\ldots,K$, and $h=T/K$. Work on a filtered probability space
\[
(\Omega,\mathcal F,(\mathcal F_{t_k})_{k=0}^K,\mathbb P),
\qquad
\E_k[\,\cdot\,]:=\E[\,\cdot\mid\mathcal F_{t_k}].
\]
Let
\[
\Delta W_k\sim N(0,hI_{d_W})
\]
be independent Gaussian increments driving risky returns, with $\Delta W_k$ independent of $\mathcal F_{t_k}$. Let $Y_k\in\R^{d_Y}$ be an observable auxiliary state adapted to $(\mathcal F_{t_k})$; it may contain predictors, stochastic investment opportunities, regimes, or other state variables, and it is absent when $d_Y=0$. Its transition is part of the declared simulator and may depend on both return-driving and orthogonal innovations.

\phantomsection\label{def:full-state}
Write
\[
S_k:=(X_k,Y_k)
\]
for the full observable state, with the convention $S_k=X_k$ when no auxiliary state is present. Conditional on $\mathcal F_{t_k}$, let
\phantomsection\label{def:predictable-coefficients}
\[
r_k=r(t_k,S_k),
\qquad
b_k=b(t_k,S_k)\in\R^n,
\qquad
A_k=A(t_k,S_k)\in\R^{n\times d_W}
\]
be predictable one-step coefficients, with $b_k$ the risky excess-return vector, and set
\[
\Sigma_k:=A_kA_k^\top\succ0,
\qquad d_W\ge n.
\]
The generic duality results below require only that these coefficients and the chosen controls be $\mathcal F_{t_k}$-measurable; they do not otherwise restrict the auxiliary-state transition.

At $t_k$, the investor chooses a predictable risky-weight vector $\pi_k$ in the nonempty polyhedron
\begin{equation}
\label{eq:polyhedral-set}
\cK=\{\pi\in\R^n:\Gamma\pi\le g\},
\qquad
\Gamma\in\R^{m\times n},\quad g\in\R^m.
\end{equation}
Starting from $X_0=x_0>0$, wealth evolves under the declared state-frozen log-Euler step
\begin{equation}
\label{eq:log-euler}
X_{k+1}
=
X_k\exp\left(
\left[r_k+\pi_k^\top b_k-\frac12\pi_k^\top\Sigma_k\pi_k\right]h
+\pi_k^\top A_k\Delta W_k
\right).
\end{equation}
When the coefficients are constant and no auxiliary state is present, \eqref{eq:log-euler} is the exact constant-weight GBM step used in the Merton benchmark.

The primal objective is terminal CRRA utility
\begin{equation}
\label{eq:crra}
U(x)=\frac{x^{1-\gamma}-1}{1-\gamma},
\qquad \gamma>1,
\end{equation}
\begin{equation}
\label{eq:primal}
J(\pi)=\E[U(X_K^\pi)],
\qquad
V^\star=\sup_{\pi_k\in\cK}J(\pi).
\end{equation}
\phantomsection\label{def:optimal-policy}
Whenever an optimizer exists, $\pi^\star$ denotes an optimal policy attaining $V^\star$. We write this optimum more explicitly as $V_h^\star$ when comparing grids. Throughout the theory and reported certificates, $h$ is fixed and we suppress the subscript. The target is the optimal value over grid-predictable, piecewise-constant controls in the declared simulator; no continuous-time policy-loss guarantee is implied. Here and below, \emph{same-grid} means that the full-state simulator, wealth update, policy evaluation, dual deflator, budget identity, and reported certificates all refer to this same declared discrete-time dynamics and constraint set.
For any admissible policy, its certainty equivalent is $\CE(\pi):=U^{-1}(J(\pi))$. The convex conjugate under the convention used below is
\begin{equation}
\label{eq:conjugate}
U^\star(z)
=
\sup_{x>0}\{U(x)-zx\}
=
\frac1{\gamma-1}
-\frac\gamma{\gamma-1}z^{(\gamma-1)/\gamma}.
\end{equation}

A direct policy optimization (DPO)\footnote{Here DPO denotes direct policy optimization and is unrelated to direct preference optimization.} reference policy is parameterized by a feasible neural feedback $\pi_\vartheta(t_k,S_k)\in\cK$ and trained by ordinary backpropagation through time (BPTT) on the same declared simulator. All dual optimization, calibration, and certification operations are performed after freezing a saved checkpoint.

\paragraph{Notation guide.}
The table lists recurring global notation; each linked symbol points to its first definition.
\begingroup
\fontsize{7.2}{8.0}\selectfont
\setlength{\tabcolsep}{2.6pt}
\renewcommand{\arraystretch}{0.96}
\begin{center}
\begin{tabularx}{\textwidth}{@{}>{\raggedright\arraybackslash}p{0.19\textwidth}>{\raggedright\arraybackslash}p{0.29\textwidth}>{\raggedright\arraybackslash}p{0.19\textwidth}>{\raggedright\arraybackslash}X@{}}
\toprule
\textbf{Symbol} & \textbf{Meaning} & \textbf{Symbol} & \textbf{Meaning}\\
\midrule
\hyperref[def:full-state]{$S_k,\pi,\pi^\star,\cK,V_h^\star$\,{\tiny[\S2]}}
& State, policies, and fixed-grid optimum
& \hyperref[eq:distortion-general]{$\mu_k,d_k,\theta_k,H_k$\,{\tiny[\S3.1]}}, \hyperref[eq:orthogonal-mean-one]{$\Lambda_{k+1}^{\perp}$\,{\tiny[\S3.4]}}
& Multiplier, distortion, market price of risk, deflator, and orthogonal tilt\\
\hyperref[def:value-residuals]{$G,\Delta_D,\mathcal R$\,{\tiny[\S1]}}, \hyperref[eq:fenchel-res]{$Z^{\rm gap}$\,{\tiny[\S4.1]}}, \hyperref[eq:residual-chain]{$\overline{\mathcal R}_{1-\alpha}$\,{\tiny[\S5]}}
& Policy loss, dual excess, residuals, and residual UCB
& \hyperref[eq:policy-norm]{$\rho^\pi,\Curv,\mathcal W_\rho$\,{\tiny[\S5.1]}}, \hyperref[def:weight-observation]{$W_i$\,{\tiny[\S5.2]}}
& Curvature-weighted occupancy geometry and positive weight observation\\
\hyperref[sec:value-to-policy]{$Q_k^\star,\mathcal Q_k$\,{\tiny[\S5.1]}}, \hyperref[eq:predictable-value-reduction]{$q_k^\star$\,{\tiny[\S7.1]}}
& Bellman action value, reduced CRRA objective, and continuation coefficient
& \hyperref[eq:directional-curvature]{$\|\cdot\|_{\bar\rho,v},\chi_v$\,{\tiny[\S5.1]}}
& Directional seminorm and inverse-curvature factor\\
\hyperref[eq:relaxed-state-value]{$\varepsilon,\lambda_{j,s}^\star,Z^{\rm face}_{j,s,\varepsilon}$\,{\tiny[\S5.3]}}
& Relaxation design, optimal face multiplier, and paired face statistic
& \hyperref[eq:conjugate]{$U^\star$\,{\tiny[\S2]}}, \hyperref[eq:general-dual-bound]{$D(y,\mu)$\,{\tiny[\S3.1]}}
& Utility conjugate and same-grid dual upper bound\\
\bottomrule
\end{tabularx}
\end{center}
\endgroup

\section{Same-grid polyhedral duality and feasibility restoration}
\label{sec:polyhedral}

Sections~\ref{sec:polyhedral}--\ref{sec:value-policy-face} form the verification pipeline. This section constructs the exactly feasible same-grid dual objects that make the later certificates valid. The theory does not require a neural dual: constant-coefficient benchmarks use a finite-dimensional convex program, whereas the state-dependent experiments use a neural or regression field only as a candidate generator. In both cases exact feasibility is imposed algebraically before certification. Section~\ref{subsec:budget-identity} derives the conditional budget identity and dual upper bound; Section~\ref{subsec:distortion-recovery} restores a candidate distortion to the feasible cone and selects a canonical multiplier representation. Sections~\ref{subsec:benchmark-geometries}--\ref{sec:orthogonal-enrichment} give canonical constraint geometries and the incomplete-market enrichment. The output is a frozen exactly feasible dual $(y,\mu)$; Section~\ref{sec:compensated-gap} then identifies the pathwise observable whose expectation is its residual.

\subsection{Predictable multipliers and the exact budget identity}
\label{subsec:budget-identity}

For a predictable multiplier $\mu_k\in\R_+^m$, define the constraint-generated drift distortion
\begin{equation}
\label{eq:distortion-general}
d_k(\mu_k)=-\Gamma^\top\mu_k
\end{equation}
and the minimum-norm market price of risk
\begin{equation}
\label{eq:theta-general}
\theta_k(\mu_k)
=
A_k^\top\Sigma_k^{-1}(b_k+d_k(\mu_k)),
\qquad
A_k\theta_k=b_k-\Gamma^\top\mu_k.
\end{equation}
The associated deflator is
\begin{equation}
\label{eq:deflator-general}
H_0=1,
\qquad
H_{k+1}
=
H_k\exp\left(
-(r_k+g^\top\mu_k)h
-\theta_k^\top\Delta W_k
-\frac12\norm{\theta_k}^2h
\right).
\end{equation}
Define the one-step complementary-slackness defect
\begin{equation}
\label{eq:defect-general}
\kappa_k(\pi_k,\mu_k)
=
\mu_k^\top(g-\Gamma\pi_k)
\ge0.
\end{equation}

\begin{theorem}[Exact same-grid identity]
\label{thm:polyhedral-budget}
For every predictable policy $\pi_k\in\cK$ and every predictable multiplier $\mu_k\ge0$,
\begin{equation}
\label{eq:general-conditional-budget}
\E_k[H_{k+1}X_{k+1}]
=
H_kX_k e^{-h\kappa_k}.
\end{equation}
Consequently,
\begin{equation}
\label{eq:general-budget-telescope}
x_0-\E[H_KX_K]
=
\E\sum_{k=0}^{K-1}H_kX_k(1-e^{-h\kappa_k})
\ge0.
\end{equation}
For every $y>0$,
\begin{equation}
\label{eq:general-dual-bound}
J(\pi)
\le
D(y,\mu)
:=
\E[U^\star(yH_K)]+yx_0.
\end{equation}
\end{theorem}

The result is conditional and therefore permits state-dependent multiplier fields, neural or otherwise, provided they are predictable and nonnegative. Exact validity is a property of the multiplier parameterization or recovery map, not of acquisition or optimization accuracy.

For $p=(\gamma-1)/\gamma$ and
\[
\mathfrak m_p(\mu):=\E[H_K^p],
\]
the scalar minimizer is
\begin{equation}
\label{eq:y-general}
y^\star(\mu)
=
\left(\frac{\mathfrak m_p(\mu)}{x_0}\right)^\gamma,
\end{equation}
and the dual certainty equivalent is
\begin{equation}
\label{eq:ce-general}
\CE_{\mathrm{dual}}(\mu)
=
x_0\mathfrak m_p(\mu)^{-\gamma/(\gamma-1)}.
\end{equation}
Thus a state-dependent dual can be trained on one split, its moment and $y$ calibrated on a second split, and its final certificate evaluated on a third.

\begin{remark}[Independent-test upper confidence bound for the dual value]
\label{rem:dual-ucb}
Fix a trained feasible dual and a scalar $y$ chosen without the final test sample. Let $N_{\rm test}$ be the number of independent final-test paths. If $U^\star(yH_K)$ has finite variance, an asymptotic one-sided upper confidence bound is
\[
\widehat D_{1-\alpha}^{\mathrm{UCB}}
=
\frac1{N_{\rm test}}\sum_{\ell=1}^{N_{\rm test}} U^\star(yH_K^{(\ell)})
+yx_0
+z_{1-\alpha}\frac{\widehat\sigma_D}{\sqrt{N_{\rm test}}}.
\]
Because the CRRA certainty-equivalent transform is increasing, applying it to this utility bound yields a dual CE UCB. Raw dual CE point estimates are treated only as diagnostics unless accompanied by this independent-test adjustment.
\end{remark}

\subsection{Support functions and distortion-cone recovery}
\label{subsec:distortion-recovery}

The recovery construction is applied pointwise at each time--state input. To keep the notation readable throughout Sections~\ref{subsec:distortion-recovery}--\ref{subsec:benchmark-geometries}, we suppress the index $k$ and the dependence on $S_k$, and write $r,b,A,\Sigma,d$ for the corresponding realized predictable quantities.

Define the feasible distortion cone
\begin{equation}
\label{eq:distortion-cone}
\cC_{\cK}=-\Gamma^\top\R_+^m
\end{equation}
and the support function in the sign convention induced by the drift distortion,
\begin{equation}
\label{eq:support-function}
\sigma_{\cK}(d)
:=
\sup_{\pi\in\cK}(-d^\top\pi).
\end{equation}
For $d\in\cC_{\cK}$, linear-programming duality gives the exact representation
\begin{equation}
\label{eq:support-lp}
\sigma_{\cK}(d)
=
\min_{\mu\ge0:\,-\Gamma^\top\mu=d}g^\top\mu.
\end{equation}
This is the finite-grid polyhedral counterpart of the constraint support-function penalty in \citet{CvitanicKaratzas1992}.

A candidate generator may output either a raw distortion $\widehat d^{\mathrm{raw}}\in\R^n$ or a raw market price of risk $\widehat\theta^{\mathrm{raw}}\in\operatorname{Range}(A^\top)$. In the latter case set
\begin{equation}
\label{eq:raw-d-general}
\widehat d^{\mathrm{raw}}
=A\widehat\theta^{\mathrm{raw}}-b.
\end{equation}
For a direct distortion output, use the associated minimum-norm representative $\widehat\theta^{\mathrm{raw}}:=A^\top\Sigma^{-1}(b+\widehat d^{\mathrm{raw}})$ when interpreting distances in market-price space. The recovery itself depends only on $\widehat d^{\mathrm{raw}}$.

\begin{lemma}[Support-function duality and attainment]
\label{lem:lp-attainment}
Assume $\cK\neq\varnothing$. For every $d\in\cC_{\cK}$, the minimum in \eqref{eq:support-lp} is finite and attained. Its value is $\sigma_{\cK}(d)$.
\end{lemma}

Let
\begin{equation}
\label{eq:least-cost-set}
\mathcal M(d)
=
\argmin_{\mu\ge0:\,-\Gamma^\top\mu=d}g^\top\mu.
\end{equation}
The total support-function defect is unique even if $\mathcal M(d)$ has multiple elements. For reproducible constraint-wise attribution, select the canonical minimum-norm representative
\begin{equation}
\label{eq:canonical-multiplier}
\mu^{\mathrm{can}}(d)
=
\argmin_{\mu\in\mathcal M(d)}\frac12\norm{\mu}^2.
\end{equation}
The strict convexity of the tie-break makes $\mu^{\mathrm{can}}(d)$ unique.

\begin{remark}[Canonical, not uniquely causal, attribution]
The minimum-norm rule makes the row-level KKT ledger reproducible when multiplier representations are nonunique. It is a canonical attribution convention, not a claim that the resulting row decomposition is the only economically causal explanation of the loss. The total support-function defect is invariant to the tie-break.
\end{remark}

\begin{theorem}[Polyhedral feasibility restoration]
\label{thm:polyhedral-recovery}
Define
\begin{equation}
\label{eq:general-distortion-projection}
\widehat d
=
\argmin_{d\in\cC_{\cK}}
\frac12(d-\widehat d^{\mathrm{raw}})^\top
\Sigma^{-1}(d-\widehat d^{\mathrm{raw}})
\end{equation}
and set $\widehat\mu=\mu^{\mathrm{can}}(\widehat d)$. Then the recovered pair is exactly dual feasible. Moreover,
\begin{equation}
\label{eq:general-nearest-kernel}
\widehat d
\in
\argmin_{d\in\cC_{\cK}}
\frac12\norm{
A^\top\Sigma^{-1}(b+d)-\widehat\theta^{\mathrm{raw}}
}^2.
\end{equation}
For every feasible action $\pi\in\cK$, the recovered total KKT defect is
\begin{equation}
\label{eq:support-defect}
\kappa(\pi,\widehat d)
=
\sigma_{\cK}(\widehat d)+\widehat d^\top\pi
=
\widehat\mu^\top(g-\Gamma\pi)
\ge0.
\end{equation}
Thus the total defect depends only on $(\widehat d,\pi)$; the canonical multiplier additionally makes its constraint-level allocation reproducible.
\end{theorem}

The recovery map restores feasibility and canonical attribution; global dual optimization remains a separate step. In a deterministic constant-coefficient benchmark, the full dual itself is the separable convex program
\begin{equation}
\label{eq:direct-dual-qp}
\min_{\mu_k\ge0}
\sum_{k=0}^{K-1}h\left[
 g^\top\mu_k
 +\frac1{2\gamma}
 \norm{A^\top\Sigma^{-1}(b-\Gamma^\top\mu_k)}^2
\right].
\end{equation}
This program is solved directly in the Merton benchmarks. In state-dependent experiments, a learned field amortizes candidate generation and the same exact recovery map is applied before certification.

\subsection{Two canonical polyhedral constraint geometries}
\label{subsec:benchmark-geometries}

The following constructions are pointwise in time and state and therefore apply to both constant-coefficient and state-dependent models.

\begin{corollary}[No-short orthant]
\label{cor:no-short-recovery}
For
\[
\cK=\R_+^n,
\qquad
\Gamma=-I_n,
\qquad g=0,
\]
the cone is $\cC_{\cK}=\R_+^n$. The first stage is the genuine covariance-metric projection
\begin{equation}
\label{eq:metric-proj}
\widehat\nu
=
\argmin_{\nu\ge0}
\frac12(\nu-\widehat d^{\mathrm{raw}})^\top
\Sigma^{-1}(\nu-\widehat d^{\mathrm{raw}}),
\end{equation}
and the multiplier representation is unique. The deflator satisfies
\begin{equation}
\label{eq:no-short-budget}
\E_k[H_{k+1}X_{k+1}]
=
H_kX_k e^{-h\nu_k^\top\pi_k}.
\end{equation}
\end{corollary}

\begin{corollary}[No-borrowing simplex]
\label{cor:simplex-recovery}
For
\[
\cK=\{\pi\ge0:\mathbf1^\top\pi\le1\},
\quad
\Gamma=\begin{pmatrix}-I_n\\ \mathbf1^\top\end{pmatrix},
\quad
g=(0,\ldots,0,1)^\top,
\]
write $\mu=(\nu,\xi)$ with $\nu\ge0$ and $\xi\ge0$. Then
\[
d=\nu-\xi\mathbf1,
\qquad
\cC_{\cK}
=
\operatorname{cone}\{e_1,\ldots,e_n,-\mathbf1\}
=
\R^n.
\]
Hence the projection stage is the identity. The least-cost exact decomposition is
\begin{equation}
\label{eq:canonical-xi}
\xi^{\mathrm{can}}(d)
=
\max\{0,-\min_i d_i\},
\qquad
\nu^{\mathrm{can}}(d)
=
d+\xi^{\mathrm{can}}(d)\mathbf1.
\end{equation}
The one-step defect is
\begin{equation}
\label{eq:simplex-defect}
\kappa_k
=
\nu_k^\top\pi_k
+\xi_k(1-\mathbf1^\top\pi_k).
\end{equation}
\end{corollary}

\paragraph{Constant-coefficient specialization.}
For constant market coefficients and deterministic multiplier sequences, the dual CE is
\begin{equation}
\label{eq:deterministic-dual-ce}
\CE_{\mathrm{dual}}(\mu)
=
x_0\exp\left(
 rT
 +\sum_{k=0}^{K-1}hg^\top\mu_k
 +\frac1{2\gamma}\sum_{k=0}^{K-1}h\norm{\theta_k(\mu_k)}^2
\right).
\end{equation}
This formula specializes to the no-short and simplex expressions used in the experiments.

\subsection{Orthogonal density enrichments in incomplete markets}
\label{sec:orthogonal-enrichment}

The minimum-norm market price of risk in Section~\ref{subsec:budget-identity} already secures dual feasibility and places no density distortion on innovations orthogonal to risky returns. When such innovations affect the auxiliary state and hence future investment opportunities, an additional positive conditionally mean-one factor can tighten the dual without changing the budget identity. Suppose the one-step Gaussian innovation decomposes as $(Z_{k+1}^{\parallel},Z_{k+1}^{\perp})$, where risky returns depend only on $Z_{k+1}^{\parallel}$ and the two blocks are conditionally independent given $\F_{t_k}$. Let $\Lambda_{k+1}^{\perp}>0$ be a measurable function of $Z_{k+1}^{\perp}$ with $\F_{t_k}$-measurable parameters, normalized by
\begin{equation}
\label{eq:orthogonal-mean-one}
\E_k[\Lambda_{k+1}^{\perp}]=1.
\end{equation}

\begin{proposition}[Incomplete-market density enrichment]
\label{prop:orthogonal-enrichment}
Multiplying each one-step deflator increment in \eqref{eq:deflator-general} by $\Lambda_{k+1}^{\perp}$ preserves the conditional budget identity \eqref{eq:general-conditional-budget}, the telescoping budget inequality, and the Fenchel upper bound. Hence an unspanned density field can be optimized freely subject to positivity, the stated conditional measurability, integrability, and \eqref{eq:orthogonal-mean-one}.
\end{proposition}

Sections~\ref{sec:predictable} and~\ref{sec:two-asset} learn the traded candidate jointly with the two bounded coefficients $(\theta_{\perp,k},c_{\perp,k})$ of a normalized positive linear--quadratic $\Lambda_{k+1}^{\perp}$. The wealth-only Merton benchmark in Section~\ref{sec:policy-geometry} sets $\Lambda_{k+1}^{\perp}\equiv1$. The analytic normalizer and moment condition are given in the Electronic Companion.

\section{Doob-compensated KKT-ledger estimation}
\label{sec:compensated-gap}

Section~\ref{sec:polyhedral} supplies a frozen feasible primal--dual pair. This section rewrites its population residual as the expectation of a directly observable, pathwise-nonnegative quantity. Section~\ref{subsec:positive-residual-ledger} separates the terminal Fenchel defect from a time-by-constraint KKT ledger, Section~\ref{subsec:doob-compensator} identifies the budget martingale removed by the canonical Doob compensator, and Section~\ref{subsec:residual-estimators} compares the resulting estimator with separate and uncompensated value differences.

Doob penalties and zero-variance behavior under ideal information-relaxation penalties are established ideas \citep{Rogers2007,BrownSmithSun2010}. The contribution here is not that generic variance-reduction principle, but the exact same-grid Fenchel--KKT representation and its finite-sample conversion into policy-distance and active-face statements.

Residual precision is structural because every policy-distance radius below scales as the square root of a residual upper bound. We report the \emph{effective path gain} as $(\SE_{\rm sep}/\SE_{\rm comp})^2$, the ratio of Monte Carlo path counts required to attain the same standard error under the usual $N^{-1/2}$ scaling. In the analytically calibrated solved benchmark, this gain rises from $1.21\times10^4$ at an early checkpoint to $2.06\times10^8$ near optimality. The increase is predicted by \Cref{cor:zero-var}: the compensated summand collapses pathwise near an exact pair, whereas ordinary value differencing retains the budget martingale. Earlier portfolio-duality methods can make the bracket itself tight; the distinct question here is how to resolve a fixed same-grid bracket precisely enough for policy-level inference \citep{HaughKoganWang2006,HaughJain2007,MaLiZheng2020}.

\subsection{Positive residuals and time-by-constraint attribution}
\label{subsec:positive-residual-ledger}

For any exactly feasible dual $(y,\mu)$, define the terminal Fenchel residual
\begin{equation}
\label{eq:fenchel-res}
F_T(y)
=
U^\star(yH_K)+yH_KX_K-U(X_K)
\ge0
\end{equation}
and the predictable complementary-slackness compensator
\begin{equation}
\label{eq:budget-comp}
B_T
=
y\sum_{k=0}^{K-1}H_kX_k(1-e^{-h\kappa_k})
\ge0.
\end{equation}
For later use, name their pathwise sum
\begin{equation}
\label{eq:pathwise-gap-residual}
Z^{\rm gap}(y,\mu;\pi)
:=F_T(y)+B_T.
\end{equation}

\begin{theorem}[Positive same-grid residual decomposition]
\label{thm:positive-gap}
For every admissible policy and every predictable exactly feasible polyhedral dual,
\begin{equation}
\label{eq:positive-gap}
\mathcal R(y,\mu;\pi)
:=
D(y,\mu)-J(\pi)
=
\E[Z^{\rm gap}(y,\mu;\pi)].
\end{equation}
At an exact primal--dual pair satisfying
\[
yH_K=U'(X_K),
\qquad
\kappa_k=0\quad\text{for every }k,
\]
both residuals vanish pathwise.
\end{theorem}

For constraint row $j$ at time $k$, define the nonnegative slack and row-level KKT contribution
\begin{equation}
\label{eq:kkt-atom}
s_{j,k}=g_j-\Gamma_j\pi_k,
\qquad
c_{j,k}=\mu_{j,k}s_{j,k},
\qquad
\kappa_k=\sum_{j=1}^m c_{j,k}.
\end{equation}
Let
\begin{equation}
\label{eq:kkt-share}
\omega_{j,k}
=
\begin{cases}
c_{j,k}/\kappa_k,&\kappa_k>0,\\
0,&\kappa_k=0,
\end{cases}
\end{equation}
and define the exact constraint-time attribution
\begin{equation}
\label{eq:kkt-ledger-entry}
B_{j,k}
=
yH_kX_k\bigl(1-e^{-h\kappa_k}\bigr)\omega_{j,k}.
\end{equation}

\begin{proposition}[Pathwise KKT ledger]
\label{prop:kkt-ledger}
Every ledger entry is nonnegative and
\begin{equation}
\label{eq:kkt-ledger-sum}
\sum_{j=1}^m B_{j,k}
=
yH_kX_k(1-e^{-h\kappa_k}),
\qquad
\mathcal R(y,\mu;\pi)
=
\E\left[F_T+\sum_{k,j}B_{j,k}\right].
\end{equation}
The total time-$k$ defect is uniquely determined by $(d_k,\pi_k)$ through \eqref{eq:support-defect}. If the multiplier representation is nonunique, the canonical tie-break \eqref{eq:canonical-multiplier} makes the constraint-level split unique as well.
\end{proposition}

Thus the compensated summand is not only a low-variance estimator. It is an audit ledger that attributes the certified residual across dates and constraints. The no-borrowing experiments aggregate these entries into no-short and borrowing-cap components; finer row-level decompositions are available without changing the estimator.

\subsection{The canonical Doob compensator}
\label{subsec:doob-compensator}

Define
\begin{equation}
\label{eq:martingale-difference}
\Delta M_{k+1}
=
H_{k+1}X_{k+1}
-\E_k[H_{k+1}X_{k+1}].
\end{equation}

\begin{theorem}[Doob-compensated residual identity]
\label{thm:doob-comp}
Pathwise,
\begin{equation}
\label{eq:doob-budget}
x_0-H_KX_K
=
\sum_{k=0}^{K-1}H_kX_k(1-e^{-h\kappa_k})
-\sum_{k=0}^{K-1}\Delta M_{k+1}.
\end{equation}
Consequently, the single-path common-random-number (CRN) residual
\begin{equation}
\label{eq:crn-path}
\mathcal R_{\mathrm{crn}}^{\mathrm{path}}
=
U^\star(yH_K)+yx_0-U(X_K)
\end{equation}
satisfies
\begin{equation}
\label{eq:crn-doob}
\mathcal R_{\mathrm{crn}}^{\mathrm{path}}
=
F_T(y)+B_T
-y\sum_{k=0}^{K-1}\Delta M_{k+1}.
\end{equation}
Thus the CRN and compensated estimators have the same expectation, but the former retains a mean-zero martingale component. Only $Z^{\rm gap}=F_T+B_T$ is pathwise nonnegative and pathwise zero at an exact pair.
\end{theorem}

The finite-variation term in \eqref{eq:budget-comp} is canonical: it is the predictable compensator in the Doob decomposition of the budget process, and its economic content is exactly dynamic complementary slackness.

\subsection{Three residual estimators and calibration effects}
\label{subsec:residual-estimators}

For a final bank of $N_{\rm MC}$ paths, we compare
\begin{align}
\widehat{\mathcal R}_{\mathrm{sep}}
&=\widehat D_{\mathrm{sep}}-\widehat J_{\mathrm{sep}},
\\
\widehat{\mathcal R}_{\mathrm{crn}}
&=\frac1{N_{\rm MC}}\sum_{\ell=1}^{N_{\rm MC}} \mathcal R_{\mathrm{crn}}^{\mathrm{path},(\ell)},
\\
\widehat{\mathcal R}_{\mathrm{comp}}
&=\frac1{N_{\rm MC}}\sum_{\ell=1}^{N_{\rm MC}}Z_{\ell}^{\rm gap}
=\frac1{N_{\rm MC}}\sum_{\ell=1}^{N_{\rm MC}}\left(F_T^{(\ell)}+B_T^{(\ell)}\right).
\end{align}
The first subtracts separately evaluated primal and dual means. The second shares paths but retains the martingale term. The third is the empirical mean of the pathwise residual in \eqref{eq:pathwise-gap-residual}, so
\[
\E[\widehat{\mathcal R}_{\mathrm{comp}}]
=\mathcal R(y,\mu;\pi).
\]
A negative CRN realization is therefore a variance event, not a duality violation, and must not be clipped into a false certificate.

\begin{corollary}[Pathwise-zero and vanishing-variance limit]
\label{cor:zero-var}
If $Z^{\rm gap}=F_T+B_T\to0$ in $L^2$, then the single-path variance of the compensated estimator converges to zero. At an exact pair its summand is identically zero. The uncompensated CRN summand generally retains the martingale term in \eqref{eq:crn-doob}.
\end{corollary}

\begin{proposition}[Second-order calibration residual]
\label{prop:y-floor}
Fix $x>0$ and $H>0$, and define
\[
\Phi_{x,H}(y)
=U^\star(yH)+yHx-U(x).
\]
Its unique minimizer is $y^\star=U'(x)/H$, and
\[
\Phi_{x,H}(y^\star)=0,
\qquad
\Phi_{x,H}'(y^\star)=0,
\qquad
\Phi_{x,H}''(y^\star)
=\frac{H^2}{-U''(x)}>0.
\]
Hence an independent calibration estimate satisfying $\widehat y-y^\star=O_p(M_{\mathrm{cal}}^{-1/2})$ induces a terminal Fenchel residual of second order, with expectation $O(M_{\mathrm{cal}}^{-1})$ under the corresponding moment conditions.
\end{proposition}

The analytically calibrated benchmark can therefore exhibit more extreme variance collapse than a learned state-dependent dual. The learned-setting decomposition is reported with the state-dependent results, where terminal calibration and density mismatch are present in addition to dynamic complementarity.

\section{From value bounds to policy-distance and active-face certificates}
\label{sec:value-policy-face}

The previous section supplies the population residual and its directly observable estimator:
\begin{equation}
\label{eq:residual-chain}
\begin{gathered}
G(\pi)
\le
\mathcal R(y,\mu;\pi)
=
\E[Z^{\rm gap}],
\qquad
\widehat{\mathcal R}_{\rm comp}
=
\frac1N\sum_{\ell=1}^N Z_{\ell}^{\rm gap},
\\[-1mm]
\Pr\!\left(
\mathcal R(y,\mu;\pi)
\le
\overline{\mathcal R}_{1-\alpha}
\right)
\ge 1-\alpha.
\end{gathered}
\end{equation}
Thus, on the confidence event in \eqref{eq:residual-chain},
\[
G(\pi)
\le
\mathcal R(y,\mu;\pi)
=
\E[Z^{\rm gap}]
\le
\overline{\mathcal R}_{1-\alpha}.
\]
Here $\overline{\mathcal R}_{1-\alpha}$ denotes any valid one-sided $(1-\alpha)$ upper confidence bound for the population residual; when the confidence level is understood, we write simply $\overline{\mathcal R}$. This section first gives the deterministic value-to-policy geometry, then estimates the residual and curvature-weight means and variances to form an operational policy-distance confidence bound, and finally treats active-face certification as a separate paired-relaxation problem.

\subsection{Value-only policy geometry}
\label{sec:value-to-policy}

Let $\mathsf S$ denote the time--full-state space associated with $S_k=(X_k,Y_k)$ in \Cref{sec:setup}, and let $\rho^\pi$ be a finite on-policy measure on $\mathsf S$; in applications it combines the deployed visitation law, the grid weight $h$, and a positive continuation weight. For a positive-definite curvature field $\Curv(s)$, define
\begin{equation}
\label{eq:policy-norm}
\|u\|_{\rho,\Curv}^2
:=
\int u(s)^\top \Curv(s)u(s)\,\rho^\pi(ds),
\qquad
\|u\|_{\rho}^2
:=
\int \|u(s)\|^2\,\rho^\pi(ds),
\qquad
\mathcal W_\rho:=\rho^\pi(\mathsf S).
\end{equation}
The measure can depend on the deployed policy because the target is policy evaluation under its own visitation law. This is not an off-policy coverage norm.

\begin{assumption}[Quadratic value-to-policy growth]
\label{ass:quadratic-growth}
There is an optimal policy $\pi^\star$ and a positive-definite curvature field $\Curv$ such that the deployed policy $\pi$ satisfies
\begin{equation}
\label{eq:abstract-quadratic-growth}
G(\pi):=V^\star-J(\pi)
\ge
\frac12\|\pi-\pi^\star\|_{\rho,\Curv}^2.
\end{equation}
\end{assumption}

\begin{proposition}[Bellman-primitive sufficient condition for quadratic growth]
\label{prop:bellman-curvature}
For $k=0,\ldots,K-1$, define the optimal action-value (Bellman $Q$) objective
\[
Q_k^\star(s,a)
:=
\E\!\left[V_{k+1}^\star(S_{k+1})\mid S_k=s,\ \pi_k=a\right],
\qquad
V_k^\star(s)=\sup_{a\in\cK}Q_k^\star(s,a).
\]
Suppose $\cK$ is convex and, for $\rho^\pi$-almost every visited state, $Q_k^\star(s,\cdot)$ is twice differentiable on the segment joining $\pi_k(s)$ and an optimizer $\pi_k^\star(s)$. If a measurable field $\underline{\Curv}_k(s)\succ0$ satisfies
\begin{equation}
\label{eq:bellman-hessian-lower}
-\nabla_{aa}^2Q_k^\star(s,a)
\succeq
\underline{\Curv}_k(s)
\end{equation}
throughout that segment, then
\begin{equation}
\label{eq:bellman-primitive-growth}
V_0^\star-J(\pi)
\ge
\frac12\E_\pi\sum_{k=0}^{K-1}
\bigl(\pi_k-\pi_k^\star(S_k)\bigr)^\top
\underline{\Curv}_k(S_k)
\bigl(\pi_k-\pi_k^\star(S_k)\bigr).
\end{equation}
Hence \Cref{ass:quadratic-growth} holds with the deployed time--state occupancy measure and any computable lower field dominated by the Bellman curvature in \eqref{eq:bellman-hessian-lower}.
\end{proposition}

\paragraph{How the lower curvature is computed.}
The proposition does not require recovery of the oracle Hessian. Let $\ell_k(s,a,Z)$ denote the one-step log-wealth increment induced by \eqref{eq:log-euler}, and consider the reduced Gaussian CRRA objective
\[
\mathcal Q_k(s,a;q)
=
\E\!\left[q(S_{k+1})e^{(1-\gamma)\ell_k(s,a,Z)}\mid S_k=s\right],
\]
where $q$ is the continuation coefficient. A covariance lower bound $\Sigma_k(s)\succeq\underline\sigma I$, a positive continuation bound $q_{k+1}^\star\ge\underline q_{k+1}$, and a primitive bound $\underline\eta_k(s)>0$---available from compactness or Gaussian coercivity---give
\begin{equation}
\label{eq:crra-primitive-curvature}
\nabla_{aa}^2\mathcal Q_k
\succeq
(\gamma-1)h\,\underline q_{k+1}\,
\underline\eta_k(s)\,\Sigma_k(s),
\qquad
\underline\eta_k(s)
:=
\inf_{a\in\cK}
\E\!\left[e^{(1-\gamma)\ell_k(s,a,Z)}\mid S_k=s\right].
\end{equation}
Conditional dual restarts provide $\underline q_{k+1}$, and nonnegative Gaussian square or covariance terms may sharpen this bound.

\paragraph{Operational scalar--geometry factorization.}
In \Cref{sec:policy-geometry,sec:predictable,sec:two-asset}, the certified lower curvature factors into a positive scalar weight $\omega_k(S_k)$ and an analytic covariance geometry $\Curv_{0,k}$. We absorb $h\omega_k$ into $\rho^\pi$ and set $\Curv=\Curv_0$; consequently, $\chi_v$ is computed from $\Curv_0^{-1}$, whereas $\mathcal W_\rho=\E_\pi\sum_k h\omega_k(S_k)$ is estimated from on-policy paths. The three geometries are $\gamma\Sigma$, $\sigma_S^2$, and $\Sigma$, respectively.

\begin{lemma}[Deterministic residual-to-policy transfer]
\label{lem:raw-policy-certificate}
Under \Cref{ass:quadratic-growth}, every valid scalar bound satisfying
\[
G(\pi)\le\mathcal R(y,\mu;\pi)\le\overline{\mathcal R}
\]
yields
\begin{equation}
\label{eq:raw-ellipsoid}
\pi^\star
\in
\mathfrak C_{\rho,\Curv}(\pi,\overline{\mathcal R})
:=
\left\{u:\|u-\pi\|_{\rho,\Curv}^2\le2\overline{\mathcal R}\right\}.
\end{equation}
Assume $\mathcal W_\rho>0$, write $\bar\rho^\pi:=\rho^\pi/\mathcal W_\rho$, and for a deterministic direction $v$ set
\begin{equation}
\label{eq:directional-curvature}
\begin{aligned}
\langle u,w\rangle_{\bar\rho,v}
&:=\int\bigl(v^\top u(s)\bigr)\bigl(v^\top w(s)\bigr)\,\bar\rho^\pi(ds),
&
\|u\|_{\bar\rho,v}
&:=\langle u,u\rangle_{\bar\rho,v}^{1/2}
=\|v^\top u\|_{L^2(\bar\rho^\pi)},\\
\|v\|_{\infty,\Curv^{-1}}
&:=\operatorname*{ess\,sup}_{s\sim\rho^\pi}
\sqrt{v^\top\Curv(s)^{-1}v},
&
\chi_v&:=\|v\|_{\infty,\Curv^{-1}}^2.
\end{aligned}
\end{equation}
The first pairing is a semi-inner product on vector-valued policies and the ordinary $L^2(\bar\rho^\pi)$ inner product on their projections; hence $d_{\bar\rho,v}(u,w):=\|u-w\|_{\bar\rho,v}$ is a pseudometric on policies and a metric on projected policies. Then
\begin{equation}
\label{eq:directional-radius}
\|\pi-\pi^\star\|_{\bar\rho,v}
\le
\frac{\|v\|_{\infty,\Curv^{-1}}}{\sqrt{\mathcal W_\rho}}
\|\pi-\pi^\star\|_{\rho,\Curv}
\le
\left(
\frac{2\chi_v\overline{\mathcal R}}{\mathcal W_\rho}
\right)^{1/2}.
\end{equation}
On the confidence event defining $\overline{\mathcal R}_{1-\alpha}$ in \eqref{eq:residual-chain}, the ellipsoid and every projected-policy radius hold simultaneously; hence each has confidence at least $1-\alpha$.
\end{lemma}

The lemma is constructive and limiting. If $\Curv(s)$ has a small eigenvalue in direction $v$, then the dual directional norm $\|v\|_{\infty,\Curv^{-1}}$---and hence $\chi_v$---is large even when the residual is tiny. Duality certifies the policy only to the extent that the objective identifies it.

\begin{proposition}[Sharpness of the square-root law]
\label{prop:sqrt-sharp}
For the quadratic local model
\[
\mathcal Q_{\Curv}(z)=\frac12 z^\top \Curv z,
\qquad \Curv\succ0,
\]
and any scalar value budget $\delta\ge0$, the maximum of $|v^\top z|$ over $\mathcal Q_{\Curv}(z)\le\delta$ is
\begin{equation}
\label{eq:sharp-radius}
\sqrt{2\delta\,v^\top \Curv^{-1}v}.
\end{equation}
Hence neither the exponent $1/2$ nor the inverse-curvature dependence in \eqref{eq:directional-radius} can be improved using scalar value information alone. As curvature tends to zero, every value-only directional policy interval must diverge.
\end{proposition}

This sharpness explains the statistical demands of policy certification. A tenfold reduction in a policy-distance radius requires an approximately hundredfold reduction in the residual UCB. The Doob-compensated estimator in \Cref{sec:compensated-gap} is introduced precisely to estimate the small positive residual rather than subtract two noisy values.

\begin{remark}[Scope of the square-root sharpness result]
\label{rem:beyond-value}
\Cref{prop:sqrt-sharp} limits identification from a scalar residual alone. Continuation-informed policy improvement uses additional information and requires a separate implementation and audit.
\end{remark}

\begin{corollary}[Coordinate intervals and support asymmetry]
\label{cor:support-asymmetry}
For a constant policy vector and $v=e_i$, one has
$\|\pi-\pi^\star\|_{\bar\rho,e_i}=|\pi_i-\pi_i^\star|$, so \eqref{eq:directional-radius} gives the coordinate interval
\[
\pi_i^\star\in[\pi_i-r_i,\pi_i+r_i].
\]
If $\pi_i-r_i>0$, inclusion of coordinate $i$ is certified. A finite symmetric interval does not certify exact exclusion; it gives only an upper tolerance on the magnitude of an excluded coordinate unless additional multiplier or margin information is available.

For a state-dependent nonnegative policy, let $\mathcal A\subseteq\mathsf S$ be a measurable set on which $\pi_i(s)\ge\eta>0$. If
\[
\int_{\mathcal A}(\pi_i-\pi_i^\star)^2d\rho^\pi\le r_{i,\mathcal A}^2,
\]
then
\begin{equation}
\label{eq:support-failure-mass}
\rho^\pi\bigl(\mathcal A\cap\{\pi_i^\star=0\}\bigr)
\le
\frac{r_{i,\mathcal A}^2}{\eta^2}.
\end{equation}
Thus an occupancy-weighted certificate can bound the visitation mass on which a claimed inclusion fails, but it is not a pointwise support guarantee.
\end{corollary}

\subsection{Statistical policy-distance certification}
\label{sec:statistical-layers}

The deterministic geometry in \Cref{sec:value-to-policy} becomes operational once the primal--dual residual is bounded from above and a model-specific lower curvature-weighted mass is bounded from below. This subsection makes that step explicit by separating their sample means and variances from the subsequent confidence construction.

Let $Z_1^{\rm gap},\ldots,Z_{N_{\mathcal R}}^{\rm gap}$ be independent copies of the compensated observation in \eqref{eq:pathwise-gap-residual}, with
\[
\E[Z_i^{\rm gap}]=\mathcal R(y,\mu;\pi).
\]
\phantomsection\label{def:weight-observation}
For the factorization above, $\mathcal W_\rho=\E_\pi\sum_k h\omega_k(S_k)$. Once the lower-weight field is fixed, a positive observation may be the full-path sum $W_i=\sum_k h\omega_k(S_k^{(i)})$ or the random-time form $W_i=h\omega_{K_i}(S_{K_i}^{(i)})/p_{K_i}$, where $\Pr(K_i=k)=p_k>0$; both have mean $\mathcal W_\rho$, and uniform dates give $W_i=T\omega_{K_i}(S_{K_i}^{(i)})$. Sections~\ref{sec:predictable}--\ref{sec:two-asset} use stratified random-time states whose local lower weights come from independent conditional-dual restarts. Let $W_1,\ldots,W_{N_W}$ denote these independent positive observations. Define the sample means and variances
\begin{align}
\label{eq:sample-gap-weight-means}
\widehat{\mathcal R}_{\rm comp}
&:=\frac1{N_{\mathcal R}}\sum_{i=1}^{N_{\mathcal R}}Z_i^{\rm gap},
&
\widehat{\mathcal W}
&:=\frac1{N_W}\sum_{i=1}^{N_W}W_i,\\
\label{eq:sample-gap-weight-variances}
\widehat\sigma_{\rm gap}^2
&:=\frac1{N_{\mathcal R}-1}
\sum_{i=1}^{N_{\mathcal R}}
\left(Z_i^{\rm gap}-\widehat{\mathcal R}_{\rm comp}\right)^2,
&
\widehat\sigma_W^2
&:=\frac1{N_W-1}
\sum_{i=1}^{N_W}
\left(W_i-\widehat{\mathcal W}\right)^2.
\end{align}

Under the usual finite-variance central-limit conditions, normal one-sided bounds are
\begin{align}
\label{eq:normal-gap-ucb}
\overline{\mathcal R}_{\rm N}
&:=
\widehat{\mathcal R}_{\rm comp}
+z_{1-\alpha_{\mathcal R}}
\frac{\widehat\sigma_{\rm gap}}{\sqrt{N_{\mathcal R}}},\\
\label{eq:normal-weight-lcb}
\underline{\mathcal W}_{\rm N}
&:=
\widehat{\mathcal W}
-z_{1-\alpha_W}
\frac{\widehat\sigma_W}{\sqrt{N_W}},
\end{align}
where $z_q$ is the $q$-quantile of the standard normal law. By one-sided coverage and a union bound,
\begin{equation}
\label{eq:normal-joint-policy-event}
\Pr\!\left(
\mathcal R(y,\mu;\pi)\le\overline{\mathcal R}_{\rm N},
\ \mathcal W_\rho\ge\underline{\mathcal W}_{\rm N}
\right)
\ge
1-\alpha_{\mathcal R}-\alpha_W+o(1).
\end{equation}
Whenever $\underline{\mathcal W}_{\rm N}>0$, the event in \eqref{eq:normal-joint-policy-event} gives
\begin{equation}
\label{eq:normal-policy-distance-certificate}
\pi^\star
\in
\mathfrak C_{\rho,\Curv}(\pi,\overline{\mathcal R}_{\rm N}),
\qquad
\|\pi-\pi^\star\|_{\bar\rho,v}
\le
\left(
\frac{2\chi_v\overline{\mathcal R}_{\rm N}}
{\underline{\mathcal W}_{\rm N}}
\right)^{1/2}.
\end{equation}
Thus \eqref{eq:sample-gap-weight-means}--\eqref{eq:normal-policy-distance-certificate} are the direct statistical implementation of the deterministic transfer in \eqref{eq:raw-ellipsoid}--\eqref{eq:directional-radius}.

Distribution-free Cantelli bounds and their resolution power are analyzed jointly with face certification in \Cref{thm:resolution-complexity}.

\subsection{Active-face certification by paired relaxation}
\label{sec:active-face}

A symmetric policy radius cannot certify exact exclusion. Polyhedral sensitivity can. Fix a restart time--state point $s=(t_k,x,z)\in\mathsf S$, where $z$ collects the auxiliary observable state coordinates and is absent in the wealth-only case, and fix a constraint row $j$. Write $\E_s[\cdot]$ for expectation under the declared grid restarted from $s$, with the restarted deflator normalized by $H_k=1$. Let $J_s(\pi)$ and $D_s(y,\mu)$ denote the corresponding conditional primal and dual values. Let $\mathcal A_s(0)$ be the original set of continuation policies, and let $\mathcal A_s^j(\varepsilon)$ relax only the current action constraint from
\[
\Gamma_j\pi_k\le g_j
\qquad\text{to}\qquad
\Gamma_j\pi_k\le g_j+\varepsilon,
\]
while leaving every later constraint unchanged. Write
\begin{equation}
\label{eq:relaxed-state-value}
V_s^j(\varepsilon)
=
\sup_{\pi\in\mathcal A_s^j(\varepsilon)}J_s(\pi),
\qquad \varepsilon\ge0.
\end{equation}
The relaxation size $\varepsilon$ is a predeclared design choice, not an estimated parameter; a finite grid may be tested under a simultaneous confidence correction. In the concave portfolio problems considered here, $V_s^j$ is concave and nondecreasing. \phantomsection\label{def:active-multiplier}Let $\lambda_{j,s}^\star$ denote an optimal current-action multiplier, equivalently a supergradient of $V_s^j$ at zero under the usual constraint qualification \citep{BonnansShapiro2000}.

\begin{proposition}[Finite-relaxation multiplier bound]
\label{prop:multiplier-secant}
For every $\varepsilon>0$,
\begin{equation}
\label{eq:multiplier-secant}
0\le
\frac{V_s^j(\varepsilon)-V_s^j(0)}{\varepsilon}
\le
\lambda_{j,s}^\star.
\end{equation}
\end{proposition}

Fix before testing an exactly feasible conditional dual $(y,\mu)$ for the original problem, a deployed original-feasible policy $\pi$, and any relaxed-feasible policy $\pi^{j,\varepsilon}$. Let
\begin{equation}
\label{eq:conditional-gap-observation}
Z_s^{\rm gap}
:=F_{T\mid s}(y)+B_{T\mid s}
\end{equation}
be the conditional counterpart of \eqref{eq:pathwise-gap-residual}, and define the paired utility difference
\[
\Delta_{j,s,\varepsilon}
=
U(X_K^{j,\varepsilon})-U(X_K^\pi).
\]
Couple the relaxed and base paths by common innovations.

\begin{proposition}[Paired compensated sensitivity identity]
\label{prop:paired-face-identity}
The directly paired face variable
\begin{equation}
\label{eq:paired-sensitivity-variable}
Z_{j,s,\varepsilon}^{\rm face}
:=
\Delta_{j,s,\varepsilon}-Z_s^{\rm gap}
\end{equation}
satisfies
\begin{equation}
\label{eq:paired-sensitivity-expectation}
m_{j,s,\varepsilon}
:=
\E_s[Z_{j,s,\varepsilon}^{\rm face}]
=
J_s(\pi^{j,\varepsilon})-D_s(y,\mu).
\end{equation}
Equivalently,
\begin{equation}
\label{eq:resolving-power-identity}
m_{j,s,\varepsilon}
=
\underbrace{J_s(\pi^{j,\varepsilon})-J_s(\pi)}_{\text{relaxation gain}}
-
\underbrace{\bigl(D_s(y,\mu)-J_s(\pi)\bigr)}_{\text{base primal--dual residual}}.
\end{equation}
\end{proposition}

For $N_{\rm face}$ independent coupled continuations, define the empirical paired mean and variance
\begin{align}
\label{eq:face-sample-mean}
\widehat m_{j,s,\varepsilon}
&:=
\frac1{N_{\rm face}}
\sum_{\ell=1}^{N_{\rm face}}
Z_{j,s,\varepsilon}^{\rm face,(\ell)},\\
\label{eq:face-sample-variance}
\widehat\sigma_{j,s,\varepsilon}^2
&:=
\frac1{N_{\rm face}-1}
\sum_{\ell=1}^{N_{\rm face}}
\left(
Z_{j,s,\varepsilon}^{\rm face,(\ell)}
-\widehat m_{j,s,\varepsilon}
\right)^2.
\end{align}
Let $\mathcal I$ be the predeclared finite family of tested state--row--relaxation triples, and let $c_{1-\alpha}$ be a simultaneous one-sided critical value. The reported Bonferroni construction takes $c_{1-\alpha}=z_{1-\alpha/|\mathcal I|}$. The corresponding lower confidence bound is
\begin{equation}
\label{eq:face-normal-lcb}
L_{j,s}(\varepsilon)
:=
\widehat m_{j,s,\varepsilon}
-c_{1-\alpha}
\frac{\widehat\sigma_{j,s,\varepsilon}}
{\sqrt{N_{\rm face}}}.
\end{equation}
Under the same finite-variance central-limit conditions, these bounds are simultaneously lower-valid over $\mathcal I$ with probability at least $1-\alpha+o(1)$. Any finite-sample or robust simultaneous lower bound may be used in place of \eqref{eq:face-normal-lcb}.

\begin{corollary}[Operational active-face certificate]
\label{cor:operational-active-face}
Let $L_{j,s}(\varepsilon)$ be any simultaneous one-sided lower confidence bound for $m_{j,s,\varepsilon}$. Then
\begin{equation}
\label{eq:operational-multiplier-lb}
\underline\lambda_{j,s}(\varepsilon)
:=
\frac{[L_{j,s}(\varepsilon)]_+}{\varepsilon}
\le
\lambda_{j,s}^\star.
\end{equation}
If $\underline\lambda_{j,s}(\varepsilon)>0$, complementary slackness certifies
\begin{equation}
\label{eq:certified-active-face}
\Gamma_j\pi_k^\star(s)=g_j.
\end{equation}
For the no-short row $-\pi_i\le0$, this is the exact exclusion statement $\pi_{k,i}^\star(s)=0$.
\end{corollary}

\begin{theorem}[Finite-sample validity and resolution power]
\label{thm:resolution-complexity}
Assume
\[
\operatorname{Var}(Z^{\rm gap})\le v_{\rm gap},
\qquad
\operatorname{Var}(W)\le v_W,
\]
and fix $\alpha_{\mathcal R},\alpha_W\in(0,1/2)$. For a direction $v$ with $\chi_v>0$ and tolerance $\tau>0$, define
\begin{equation}
\label{eq:policy-resolution-margin}
a_{\tau,v}:=\frac{\tau^2}{2\chi_v},
\qquad
\Delta_{\tau,v}:=a_{\tau,v}\mathcal W_\rho-\mathcal R(y,\mu;\pi),
\end{equation}
and the Cantelli radii and reported bounds
\[
c_{\mathcal R}:=\sqrt{\frac{v_{\rm gap}(1-\alpha_{\mathcal R})}{N_{\mathcal R}\alpha_{\mathcal R}}},
\quad
c_W:=\sqrt{\frac{v_W(1-\alpha_W)}{N_W\alpha_W}},
\quad
U_{\mathcal R}:=\widehat{\mathcal R}_{\rm comp}+c_{\mathcal R},
\quad
L_W:=\widehat{\mathcal W}-c_W.
\]
Then:
\begin{enumerate}[label=(\roman*),leftmargin=1.6em]
\item \emph{Validity.} With probability at least $1-\alpha_{\mathcal R}-\alpha_W$, the implication
\begin{equation}
\label{eq:policy-valid-trigger}
L_W>0,\quad U_{\mathcal R}\le a_{\tau,v}L_W
\quad\Longrightarrow\quad
\|\pi-\pi^\star\|_{\bar\rho,v}\le\tau
\end{equation}
holds.
\item \emph{Resolution power.} If $\Delta_{\tau,v}>0$ and
\begin{equation}
\label{eq:policy-resolution-condition}
2c_{\mathcal R}+2a_{\tau,v}c_W<\Delta_{\tau,v},
\end{equation}
then the trigger in \eqref{eq:policy-valid-trigger} occurs, and its conclusion is valid, with probability at least $1-2\alpha_{\mathcal R}-2\alpha_W$. If $N_{\mathcal R}=N_W=N$, it is sufficient that
\begin{equation}
\label{eq:policy-resolution-sample-size}
N>
\frac{4}{\Delta_{\tau,v}^2}
\left(
\sqrt{\frac{v_{\rm gap}(1-\alpha_{\mathcal R})}{\alpha_{\mathcal R}}}
+a_{\tau,v}\sqrt{\frac{v_W(1-\alpha_W)}{\alpha_W}}
\right)^2.
\end{equation}
\end{enumerate}

For active faces, let $\mathcal I$ contain $M$ predeclared state--row--relaxation triples. If triple $i$ has paired mean $m_i>0$ and variance at most $v_i$, set $\beta=\alpha/(2M)$,
\[
c_i:=\sqrt{\frac{v_i(1-\beta)}{N_i\beta}},
\qquad
L_i:=\widehat m_i-c_i.
\]
If $m_i>2c_i$ for every target triple, equivalently
\begin{equation}
\label{eq:face-resolution-sample-size}
N_i>\frac{4v_i(1-\beta)}{\beta m_i^2},
\end{equation}
then, with probability at least $1-\alpha$, all corresponding lower bounds are simultaneously valid and positive, and all target faces are certified.
\end{theorem}

The theorem separates correctness from detectability. One-sided tails make every reported policy certificate valid; the reverse tails determine whether the positive population margin $\Delta_{\tau,v}$ is detected. A face is detectable only when its paired mean is positive. Doob compensation reduces $v_{\rm gap}$ in \eqref{eq:policy-resolution-sample-size}, while common-path pairing reduces $v_i$ in \eqref{eq:face-resolution-sample-size}; both reduce the required path count.

\begin{remark}[Cantelli versus reported normal bounds]
The experiments use less conservative split-sample normal one-sided bounds, whereas \Cref{thm:resolution-complexity} records a distribution-free Cantelli route. The validity--power distinction is identical; only the critical radii and sample-size constants differ.
\end{remark}

\paragraph{Operational relevance of a certified face.}
For a no-short row, a positive certificate proves that the zero holding is optimal rather than an output clipping artifact. For a leverage row, it proves that the risk budget genuinely restricts the optimum. Such declarations support model validation, boundary-sensitive comparative statics, and auditable explanations of deployed decisions. A row that is not certified remains unresolved.

The resolving-power identity \eqref{eq:resolving-power-identity} shows exactly what must be overcome: the relaxation gain must exceed the base residual and the simultaneous one-sided sampling margin. A tighter primal--dual pair therefore resolves more active faces, all else equal. The directly paired variable is essential when the sensitivity is much smaller than the conditional primal and dual values; it avoids a separate confidence interval for a large conditional dual moment. A positive lower bound certifies the face, whereas a nonpositive bound leaves it unresolved.

The result is valid for any relaxed-feasible $\pi^{j,\varepsilon}$ fixed before the held-out test. This choice affects statistical power, not certificate validity. The model sections below state the boundary-exposing or relaxed-QP choices used in each experiment.

\paragraph{From certification theory to computation.}
Sections~\ref{sec:polyhedral}--\ref{sec:value-policy-face} require frozen feasible primal and dual objects but do not prescribe how those objects are produced. Analytical formulas, dynamic programming, regression, or neural approximation can therefore serve as candidate generators. Section~\ref{sec:policy-geometry} isolates the certificate geometry in a tractable Merton benchmark. Sections~\ref{sec:predictable} and~\ref{sec:two-asset} instantiate the same audit with state-dependent neural primal and dual approximations. In every study, the reference solution is queried only after the certificate has been computed.

\section{Operational policy-distance certificates in the Merton class}
\label{sec:policy-geometry}

The Merton class is the wealth-only, constant-coefficient special case of \Cref{sec:setup}, with no auxiliary state and $r_k=r$, $b_k=b$, and $A_k=A$. Unlike the state-dependent studies that follow, it is an analytical benchmark: its dual can be optimized directly and its optimum is available for external validation. It therefore isolates the geometry and statistical tightness of the residual-to-policy and multiplier transfers before introducing approximation error from state-dependent neural objects. This section specializes \Cref{lem:raw-policy-certificate} to the classical Merton portfolio structure \citep{Merton1969,Merton1971}, evaluated on the declared log-Euler grid, and constructs occupancy-weighted coordinate radii without using the analytical optimum in the certificate.

Define
\begin{equation}
\label{eq:f-def}
f(\pi)=b^\top\pi-\frac\gamma2\pi^\top\Sigma\pi
\end{equation}
and let
\begin{equation}
\label{eq:pi-star-general}
\pi^\star=\argmax_{\pi\in\cK}f(\pi).
\end{equation}
Let $\mu^\star\ge0$ satisfy
\begin{equation}
\label{eq:kkt-general}
b-\gamma\Sigma\pi^\star-\Gamma^\top\mu^\star=0,
\qquad
(\mu^\star)^\top(g-\Gamma\pi^\star)=0.
\end{equation}
The optimum is unique under $\Sigma\succ0$, and
\begin{equation}
\label{eq:ce-star-general}
\CE^\star
=x_0\exp([r+f(\pi^\star)]T).
\end{equation}
Set
\begin{equation}
\label{eq:value-coefficient}
a_k^\star
=
\exp\left((1-\gamma)[r+f(\pi^\star)](T-t_k)\right),
\qquad
w_k=a_k^\star(X_k^\pi)^{1-\gamma}.
\end{equation}

\begin{theorem}[Exact Merton performance difference]
\label{thm:performance}
For every admissible policy,
\begin{equation}
\label{eq:perf-exact}
V^\star-J(\pi)
=
\E\sum_{k=0}^{K-1}
\frac{w_k}{\gamma-1}
\left[
\exp((\gamma-1)h\Delta f_k)-1
\right],
\end{equation}
where
\begin{equation}
\label{eq:delta-decomp}
\Delta f_k
:=f(\pi^\star)-f(\pi_k)
=
\frac\gamma2(\pi_k-\pi^\star)^\top\Sigma(\pi_k-\pi^\star)
+(\mu^\star)^\top(g-\Gamma\pi_k).
\end{equation}
Consequently,
\begin{equation}
\label{eq:perf-lower}
V^\star-J(\pi)
\ge
\E\sum_{k=0}^{K-1}hw_k
\left[
\frac\gamma2(\pi_k-\pi^\star)^\top\Sigma(\pi_k-\pi^\star)
+(\mu^\star)^\top(g-\Gamma\pi_k)
\right].
\end{equation}
\end{theorem}

\subsection{Dual-derived lower continuation weight}
\label{subsec:merton-continuation-weight}

Let $\overline C\ge\CE^\star$ be any valid dual CE upper bound and set
\begin{equation}
\label{eq:lower-coefficient}
\overline r_{\rm CE}=\frac1T\log\frac{\overline C}{x_0},
\qquad
\underline a_k(\overline C)=\exp\left((1-\gamma)\overline r_{\rm CE}(T-t_k)\right).
\end{equation}
Since $\overline r_{\rm CE}\ge r+f(\pi^\star)$ and $1-\gamma<0$, we have $\underline a_k\le a_k^\star$. Define
\begin{equation}
\label{eq:lower-W}
\underline{\mathcal W}(\overline C)
=
\E\sum_{k=0}^{K-1}h\underline a_k(\overline C)(X_k^\pi)^{1-\gamma}.
\end{equation}
Then $\underline{\mathcal W}\le\mathcal W:=\E\sum_khw_k$, and $\underline{\mathcal W}$ is estimable from the deployed paths and dual CE. In the notation of Section~\ref{sec:value-policy-face}, \eqref{eq:perf-lower} sets $C_0=\gamma\Sigma$ and absorbs $h w_k$ into the on-policy measure, so the abstract integral is exactly the displayed expectation--sum, $W_\rho=\mathcal W$, and $\chi_v=\gamma^{-1}v^\top\Sigma^{-1}v$. The bar in $\bar\rho_{\rm M}^\pi$ denotes normalization; the operational denominator is a one-sided held-out LCB for $\underline{\mathcal W}(\overline C)$.

\begin{corollary}[Operational Merton coordinate radius]
\label{cor:operational-radius}
Let $\overline{\mathcal R}$ be any upper bound on $V^\star-J(\pi)$, and let $\underline{\mathcal W}>0$ be any lower bound of the form \eqref{eq:lower-W}. Then
\begin{equation}
\label{eq:coord-radius}
\|\pi-\pi^\star\|_{\bar\rho_{\rm M},e_i}
\le
\Rad_{{\rm M},i}
:=
\left[
\frac{2\overline{\mathcal R}}{\gamma\underline{\mathcal W}}
(\Sigma^{-1})_{ii}
\right]^{1/2}.
\end{equation}
If a one-sided held-out LCB for \eqref{eq:lower-W} and a one-sided UCB for the compensated residual hold jointly, the resulting radius is a high-probability operational certificate within the Merton class.
\end{corollary}

\begin{remark}[On-policy scope]
\label{rem:occupancy}
The projected-policy seminorm on the left side of \eqref{eq:coord-radius} is induced by the deployed policy's own visitation law and continuation weight. It is therefore an on-policy, occupancy-weighted certificate. It does not imply a uniform statewise bound in regions that the deployed policy rarely or never visits. \Cref{thm:restarted-operational-radius} constructs a restarted conditional curvature bound on that same on-policy support; a genuinely uniform statewise certificate still requires an explicit coverage distribution or a concentrability argument.
\end{remark}

At the final no-short checkpoint, the ratio of the dual-derived coefficient lower bound to the analytical coefficient is at least
\[
(\CE_{\mathrm{dual}}/\CE^\star)^{1-\gamma}
=0.9999954,
\]
so replacing the analytical coefficient by the dual-derived one inflates the deterministic part of the radius by at most a factor $1.0000023$ before the held-out weight-estimation margin.

\begin{figure}[!htbp]
\centering
\includegraphics[width=\textwidth]{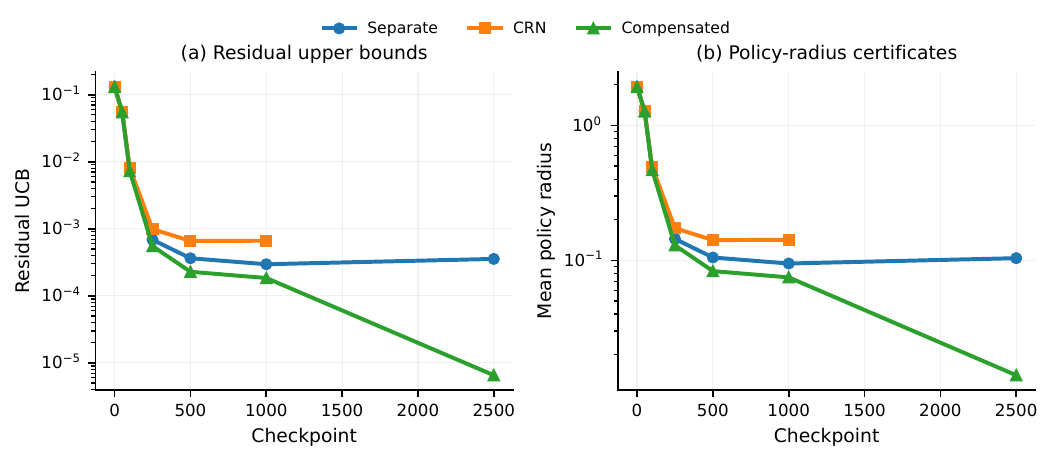}
\caption{Doob compensation makes the Merton residual and the induced policy radius operational across frozen-policy checkpoints. Nonpositive uncompensated CRN upper bounds are omitted on the log scale rather than clipped.}
\label{fig:merton-gap-radius}
\end{figure}

\Cref{fig:merton-gap-radius} shows final-checkpoint coverage of all five occupancy-weighted coordinate RMS errors. The mean radius is $0.01409$ under compensation versus $0.10385$ under separate evaluation; the actual coordinate RMS errors range from $8.4\times10^{-5}$ to $5.52\times10^{-3}$. Repeated audits at checkpoints $500$, $1000$, and $2500$ preserve coordinate coverage throughout and attain nominal $95\%$ residual-UCB coverage in $228/240$ cells. Full checkpoint, path-budget, and replication tables are in the Electronic Companion.

\subsection{Operational exclusion and leakage bounds}
\label{sec:merton-exclusion}

For the no-short Merton problem, relax the $i$th constraint to $\pi_i\ge-\varepsilon$. Let $\underline C_{i,\varepsilon}$ be a one-sided lower confidence bound for the relaxed policy's certainty equivalent, and let $\overline C_0$ be an original-problem dual CE upper bound. Concavity of the optimal growth rate gives the operational lower bound
\begin{equation}
\label{eq:merton-multiplier-lb}
\underline\nu_i(\varepsilon)
=
\frac{1}{T\varepsilon}
\left[\log\frac{\underline C_{i,\varepsilon}}{\overline C_0}\right]_+
\le \nu_i^\star.
\end{equation}
For a preselected relaxation level $\varepsilon_i^\star$, write $\underline\nu_i:=\underline\nu_i(\varepsilon_i^\star)$. Whenever $\underline\nu_i>0$, the optimal portfolio excludes asset $i$. Moreover, combining \eqref{eq:perf-lower} with a residual UCB $\overline{\mathcal R}$ and a continuation-weight LCB $\underline{\mathcal W}$ yields the operational leakage bound
\begin{equation}
\label{eq:merton-leakage-bound}
\frac{\E\sum_k h w_k\pi_{k,i}}{\mathcal W}
\le
\frac{\overline{\mathcal R}}{\underline\nu_i\,\underline{\mathcal W}}.
\end{equation}
As shown in \Cref{fig:merton-exclusion}, the solved benchmark validates both statements without the analytical multiplier. Among five assets, the certificate identifies exactly assets $2$ and $4$, the two analytically excluded assets, and produces no false exclusion. The certified multiplier lower bounds are $9.16\times10^{-4}$ and $1.04\times10^{-2}$, compared with analytical multipliers $2.73\times10^{-3}$ and $1.35\times10^{-2}$. The corresponding weighted-allocation bounds are $7.63\times10^{-3}$ and $6.70\times10^{-4}$, covering actual weighted allocations $4.50\times10^{-4}$ and $8.34\times10^{-5}$.

\begin{figure}[!htbp]
\centering
\includegraphics[width=\textwidth]{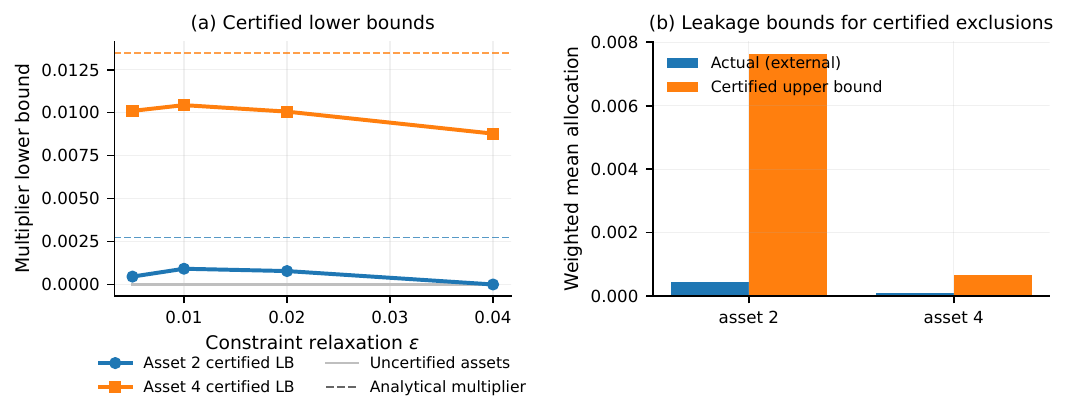}
\caption{Operational active-face certification in the solved five-asset benchmark. Certified multiplier lower bounds are shown against analytical multipliers, and the certified leakage bounds cover the weighted allocations of the two excluded assets. The analytical solution is used only for external validation.}
\label{fig:merton-exclusion}
\end{figure}

\subsection{High-dimensional stress test under an exact dual wrapper}
\label{sec:highdim-merton}
To separate dimensional scalability of the verification layer from the harder problem of learning an expressive state-dependent dual, we consider no-short Merton markets with a shared risky-weight cap at dimensions $d\in\{5,20,50\}$. A short, deliberately imperfect feasible neural policy is frozen; the exactly feasible dual is then obtained from the deterministic convex direct-dual program. Three pilot markets fixed the protocol, after which five new market seeds were run without retuning. The external primal QP is solved only after the residual and radius are fixed.

\begin{table}[!htbp]
\centering
\caption{Locked high-dimensional Merton replication under an exact dual wrapper.}
\label{tab:highdim-merton-main}
\small
\begin{tabular}{rcccc}
\toprule
$d$ & Radius coverage & Mean radius/error & Maximum radius/error & Mean inactive assets\\
\midrule
5  & $5/5$ & 1.005 & 1.008 & 1.0\\
20 & $5/5$ & 1.007 & 1.023 & 7.8\\
50 & $5/5$ & 1.005 & 1.008 & 32.2\\
\bottomrule
\end{tabular}
\end{table}

\Cref{tab:highdim-merton-main} shows $15/15$ radius coverage; the direct dual and external QP also match the no-short active set and cap status in every cell. The prespecified cellwise $97.5\%$ residual UCB covers $14/15$ external true gaps; the sole miss is $4.39\times10^{-9}$. A familywise Bonferroni recalibration, reported as a sensitivity analysis, covers all $15$ while raising the overall maximum radius-to-error ratio only from $1.023$ to $1.025$. Thus action dimension alone does not degrade the certificate when a tight feasible wrapper is available; the result does not establish tightness of a generic learned state-dependent dual. Seed-level results, the direct-dual program, and timing are in the Electronic Companion. A separate predictor-driven pilot at $d=20,50$ preserves coverage, conditional-dual coverage, and stable curvature-weight effective sample sizes, but its radius-to-error ratios rise to $28.5$ and $44.8$. Both the terminal Fenchel and dynamic complementarity components increase with dimension. This diagnostic locates the practical bottleneck in learning a tight state-dependent feasible dual, rather than in the restart curvature or residual-to-policy transfer itself.

\section{State-dependent operational study}
\label{sec:predictable}

The predictable-return model is the first state-dependent case. It supports the locked ten-seed replication, detailed restart-budget audits, and pointwise active-face certification at visited states. \Cref{sec:two-asset} then repeats the full operational chain in the smallest state-dependent multivariate market with multiple interacting faces.

\paragraph{Neural instantiation and freezing protocol.}
The state-dependent experiments implement the solver-neutral theory through the common pipeline
\[
\text{candidate generation}
\;\longrightarrow\;
\text{exact feasibility wrapper}
\;\longrightarrow\;
\text{held-out same-grid certification}.
\]
The primal network is trained first and frozen. Each state-dependent study then uses one vector-output dual network followed by the algebraic feasibility wrapper of Section~\ref{sec:polyhedral}. The one-asset output is $(\nu_\phi,\theta_{\perp,\phi},c_{\perp,\phi})$ with $d_\phi=\nu_\phi\ge0$; the two-asset output is $(d_{1,\phi},d_{2,\phi},\theta_{\perp,\phi},c_{\perp,\phi})$. In both cases, the remaining multipliers and traded market price of risk are reconstructed algebraically. After snapshot selection and scalar calibration, the primal and dual objects are frozen and evaluated on disjoint held-out paths. Exact parameterizations are summarized in the Electronic Companion.

\subsection{Same-grid model and external dynamic program}
\label{subsec:predictable-model-dp}

There is one risky asset and one Ornstein--Uhlenbeck predictor, following the predictable-return portfolio tradition of \citet{KimOmberg1996,Liu2007}. With independent standard normals $(Z_{1,k},Z_{2,k})$, let
\begin{align}
\label{eq:predictable-wealth}
\log X_{k+1}
&=
\log X_k
+
\left[r+\pi_k b(Y_k)-\frac12\pi_k^2\sigma_S^2\right]h
+
\pi_k\sigma_S\sqrt h\,Z_{1,k},\\
\label{eq:predictable-factor}
Y_{k+1}
&=
(1-\kappa_Yh)Y_k
+
\sigma_Y\sqrt h\left(\rho Z_{1,k}+\rho_{\perp}Z_{2,k}\right),
\qquad
\rho_{\perp}=\sqrt{1-\rho^2},
\end{align}
with
\begin{equation}
\label{eq:predictable-premium}
b(y)=b_0+b_1\tanh(y/b_{\mathrm{sc}}).
\end{equation}
The numerical specification is
\[
T=1,
\quad K=30,
\quad r=0.02,
\quad \gamma=3,
\quad \sigma_S=0.2,
\quad \kappa_Y=1.5,
\quad \sigma_Y=0.35,
\quad \rho=-0.35,
\]
\[
b_0=0.025,
\qquad b_1=0.10,
\qquad b_{\mathrm{sc}}=0.30,
\qquad (X_0,Y_0)=(1,0).
\]
The policy is constrained by $\pi_k\ge0$. A loose numerical output cap $\pi_{\max}=1.5$ stabilizes the neural and grid implementations; the finest-grid optimum stays strictly below the cap, with $\pi_0(Y_0)=0.22344$ and a maximum grid action of approximately $1.05$.

CRRA homogeneity reduces the same-grid Bellman value to
\begin{equation}
\label{eq:predictable-value-reduction}
V_k(x,y)
=
\frac{x^{1-\gamma}q_k^\star(y)-1}{1-\gamma},
\qquad q_K^\star(y)=1.
\end{equation}
The external dynamic program computes
\begin{equation}
\label{eq:predictable-dp-recursion}
q_k^\star(y)
=
\inf_{0\le\pi\le\pi_{\max}}
\E\left[
 e^{(1-\gamma)\ell_k(\pi,y,Z_{1,k})}
 q_{k+1}^\star(Y_{k+1})
\mid Y_k=y
\right],
\end{equation}
where $\ell_k$ is the log-wealth increment in \eqref{eq:predictable-wealth}. This solver is not used in primal training, anchor fitting, dual refinement, snapshot selection, scalar-$y$ calibration, or final certification. It is only an external same-grid numerical cross-check.

\subsection{Hybrid conditional curvature and an operational radius}
\label{sec:restarted-radius-theory}

Write $\beta=1-\gamma<0$ and define
\begin{equation}
\label{eq:reduced-one-step-objective}
\mathcal Q_k(y,\pi;q)
=
\E\left[q(Y_{k+1})\exp\{\beta\ell_k(\pi,y,Z_{1,k})\}\mid Y_k=y\right].
\end{equation}
For $q_{k+1}^\star$, the optimal coefficient is $q_k^\star(y)=\min_{\pi\in[0,\pi_{\max}]}\mathcal Q_k(y,\pi;q_{k+1}^\star)$.

\begin{proposition}[Exact and hybrid predictable-return curvature]
\label{prop:predictable-curvature}
Let
\[
u_k(\pi,y,z)
=(b(y)-\pi\sigma_S^2)h+\sigma_S\sqrt h\,z,
\qquad
\mathcal L_{k+1}=q_{k+1}^\star(Y_{k+1})e^{\beta\ell_k}.
\]
Then
\begin{equation}
\label{eq:predictable-curvature-exact}
\partial_{\pi\pi}\mathcal Q_k(y,\pi;q_{k+1}^\star)
=
\E\left[\mathcal L_{k+1}\left\{(\gamma-1)\sigma_S^2h+(\gamma-1)^2u_k(\pi,y,Z_{1,k})^2\right\}\mid Y_k=y\right].
\end{equation}
The first term gives the basic bound
\begin{equation}
\label{eq:basic-curvature-lower}
\partial_{\pi\pi}\mathcal Q_k
\ge
(\gamma-1)\sigma_S^2h\,q_k^\star(y).
\end{equation}
Suppose additionally that $\underline q_{k+1}^{\rm unif}\le q_{k+1}^\star(y')$ for every declared predictor state. Define
\begin{align}
\label{eq:gaussian-square-factor}
\Psi_k(y,\pi)
&=
\E\left[e^{\beta\ell_k(\pi,y,Z)}u_k(\pi,y,Z)^2\right]\\
&=
\exp\left(\beta\left[r+\pi b(y)-\frac\gamma2\pi^2\sigma_S^2\right]h\right)
\left\{\sigma_S^2h+[b(y)-\gamma\pi\sigma_S^2]^2h^2\right\},
\end{align}
\begin{equation}
\label{eq:psi-min}
\underline\Psi_k(y)=\min_{0\le\pi\le\pi_{\max}}\Psi_k(y,\pi),
\qquad
\argmin\Psi_k
=
\Proj_{[0,\pi_{\max}]}
\left(\frac{b(y)}{\gamma\sigma_S^2}\right).
\end{equation}
For any lower bound $\underline q_k^{\rm rst}(y)\le q_k^\star(y)$, the effective coefficient
\begin{equation}
\label{eq:hybrid-q-lower}
\underline q_k^{\rm hyb}(y)
=
\underline q_k^{\rm rst}(y)
+
\frac{\gamma-1}{\sigma_S^2h}
\underline q_{k+1}^{\rm unif}\underline\Psi_k(y)
\end{equation}
obeys the performance transfer
\begin{equation}
\label{eq:predictable-performance-hybrid}
V_0^\star-J(\pi)
\ge
\frac{\sigma_S^2}{2}
\E\sum_{k=0}^{K-1}h(X_k^\pi)^{1-\gamma}
\underline q_k^{\rm hyb}(Y_k)
(\pi_k-\pi_k^\star(Y_k))^2.
\end{equation}
\end{proposition}

The restart component in \eqref{eq:hybrid-q-lower} is obtained from the same conditional dual used for value certification. Start a feasible density at $(t_k,Y_k=y)$ with $H_k=1$, write
\[
H_{k,K}:=H_K/H_k
\]
for the restarted deflator increment, and set
\begin{equation}
\label{eq:restart-moment}
m_k(y)=\E[H_{k,K}^{p}\mid Y_k=y],
\qquad p=\frac{\gamma-1}{\gamma}.
\end{equation}
Since the conditional dual CE is $m_k(y)^{-\gamma/(\gamma-1)}$, weak duality gives
\begin{equation}
\label{eq:restart-q-lower}
q_k^\star(y)\ge m_k(y)^\gamma.
\end{equation}
A one-sided moment LCB $\underline m_k(y)$ therefore gives $\underline q_k^{\rm rst}(y)=\underline m_k(y)^\gamma$.

\begin{theorem}[Operational restarted policy radius]
\label{thm:restarted-operational-radius}
Let $\overline{\mathcal R}$ be any valid upper bound for $V_0^\star-J(\pi)$ and define
\begin{equation}
\label{eq:hybrid-restart-weight}
\mathcal W_{\rm hyb}
=
\E\sum_{k=0}^{K-1}h(X_k^\pi)^{1-\gamma}\underline q_k^{\rm hyb}(Y_k).
\end{equation}
Let $\bar\rho_{\rm hyb}^\pi$ normalize the weight in \eqref{eq:hybrid-restart-weight} by $\mathcal W_{\rm hyb}$.
On any event on which the gap, restart moments, uniform next-stage bounds, and a positive lower bound $\underline{\mathcal W}_{\rm hyb}\le\mathcal W_{\rm hyb}$ all hold,
\begin{equation}
\label{eq:hybrid-restart-radius}
\|\pi-\pi^\star\|_{\bar\rho_{\rm hyb},1}
\le
\Rad_{\rm hyb}
:=
\left(\frac{2\overline{\mathcal R}}{\sigma_S^2\underline{\mathcal W}_{\rm hyb}}\right)^{1/2}.
\end{equation}
Thus any joint $(1-\alpha)$ confidence construction for those one-sided events yields a $(1-\alpha)$ on-policy policy certificate. No optimal policy or dynamic-programming value enters its construction.
\end{theorem}

\begin{remark}[Reconciliation with the homothetic benchmark]
The basic bound \eqref{eq:basic-curvature-lower} discards the entire positive square term. In the homothetic Merton limit, retaining it recovers the exact factor $\gamma$ in the action curvature from \Cref{thm:performance}. In the state-dependent problem, \eqref{eq:hybrid-q-lower} recovers only the portion certified by the uniform next-stage bound; it is therefore a finite-grid lower bound rather than a claim of an exact factor-$\gamma$ improvement.
\end{remark}

\subsection{Locked ten-seed policy-radius replication}
\label{sec:operational-ten-seed}

The reported protocol uses $96$ independently generated on-policy restart states per seed, stratified over the $30$ dates, and $16{,}384$ conditional continuations per state. Their outer weights are the stratified random-time observations of \Cref{sec:statistical-layers}. The first three seeds were used during development; all hyperparameters, checkpoint selection rules, confidence allocation, restart budget, and relaxation grid were then frozen before seven new seeds were run. The gap and curvature construction uses no external dynamic program.

In the three development seeds, adding the closed-form term in \eqref{eq:hybrid-q-lower} reduces the mean radius from $0.06246$ to $0.04167$, a $33.29\%$ reduction. It changes the number of certified policies from one to two at tolerance $0.05$ and from two to three at tolerance $0.075$. The improvement is therefore operational rather than cosmetic.

Seed-level development/replication results are reported in the Electronic Companion.

\begin{figure}[!htbp]
\centering
\includegraphics[width=0.66\textwidth]{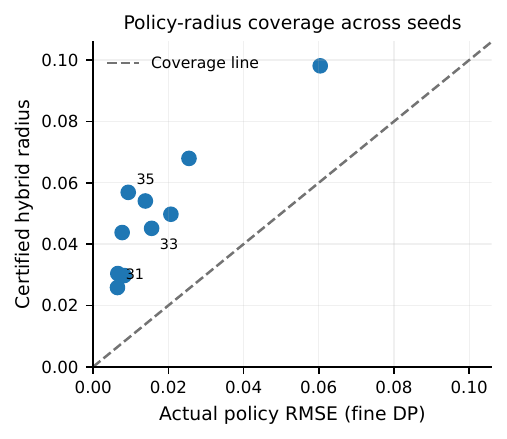}
\caption{Locked one-asset policy-radius validation. Every fine-grid weighted policy error lies below the operational hybrid radius.}
\label{fig:ten-seed-radius}
\end{figure}

\Cref{fig:ten-seed-radius} shows all ten radii covering their fine-grid weighted errors. The mean radius is $0.05016$, the mean fine-grid error is $0.01741$, and the mean radius-to-error ratio is $3.76$ (range $1.63$--$6.11$). The difficult Seed 20260737 has a fine-grid error of $0.06037$ and is still covered by radius $0.09812$, showing that coverage is not restricted to nearly converged policies. Across seeds, the restarted continuation lower bound retains $99.03\%$ of the fine-grid continuation coefficient on average; the mean of the seedwise minimum ratios is $97.37\%$.

In the short-budget learned-dual audit, the mean terminal Fenchel and dynamic complementarity residuals are $3.42\times10^{-5}$ and $4.10\times10^{-5}$, respectively. Their comparable scale explains why the extreme path gains of the analytically calibrated Merton benchmark should not be transferred unchanged to the learned setting: the terminal component contains scalar-calibration and density mismatch in addition to the dynamic KKT component.

\subsection{State-dependent active-face certification}
\label{sec:state-exclusion}

At each of the same $96$ restart states, the current no-short constraint is relaxed to $\pi_k\ge-\varepsilon$ for
\[
\varepsilon\in\{0.005,0.01,0.02,0.04,0.08,0.12,0.16\}.
\]
For each candidate, the relaxed first action is set to $\pi_k^{\varepsilon}=-\varepsilon$; from the next date onward the frozen deployed feedback is used. The scalar dual is calibrated on $16{,}384$ independent paths, fixed before testing, and the paired face variable \eqref{eq:paired-sensitivity-variable} is evaluated on $32{,}768$ common held-out paths. A single Bonferroni family covers all state--relaxation candidates in each seed.

\begin{figure}[!htbp]
\centering
\includegraphics[width=0.66\textwidth]{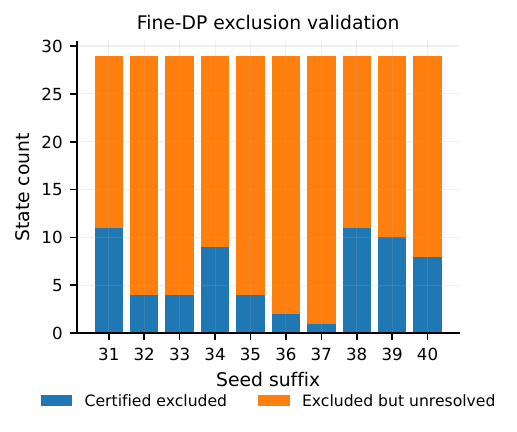}
\caption{One-asset active-face validation against the independently refined dynamic program. Across seeds, certified exclusions are a subset of the fine-grid excluded states, with no false declaration.}
\label{fig:ten-seed-exclusion}
\end{figure}

\Cref{fig:ten-seed-exclusion} shows $64$ of $290$ fine-grid exclusions certified, with zero false exclusions among the other $670$ states. Every reported multiplier lower bound is below the corresponding fine-grid oracle secant. The seven new seeds alone certify $45$ of $203$ externally excluded states, again with zero false exclusions.

The certified fraction varies with the balance in \eqref{eq:resolving-power-identity}. The five lower-gap seeds certify $45$ of $145$ externally excluded audited states ($31.0\%$), whereas the five higher-gap seeds certify $19$ of $145$ ($13.1\%$). The corresponding ten-seed Spearman diagnostic is $-0.960$. We report it as a finite-sample check of the exact resolving-power identity, not as a population scaling law: a face becomes certifiable only after the relaxation gain clears the base residual and the simultaneous sampling margin.

The relaxation-grid audit also supports the chosen range. Expanding the development grid from five values ending at $0.08$ to seven values ending at $0.16$ preserves all $18$ original certificates, adds one, and causes no multiplicity-induced loss. Only two of the resulting $19$ development-seed certificates select the upper endpoint. Failure to certify remains inconclusive; the method is designed to avoid false structural statements rather than maximize recall.

\subsection{External refinement and accept/retrain decisions}
\label{sec:external-refinement-decision}
A separate restart-budget audit finds coverage in every tested cell; at the reported $(N,M)=(96,16{,}384)$, the mean weight ESS is $92.0/96$ and the radius is already stable to further path increases. Full budget surfaces are reported in the Electronic Companion.

For final external validation, the predictor/action/quadrature grid is refined from $401\times321\times7$ to $1201\times961\times11$. The fine-grid CE is $1.030130895586$, and the two finest CE levels differ by $4.54\times10^{-7}$. Refinement changes the audited optimal action by $5.99\times10^{-4}$ on average and by at most $1.56\times10^{-3}$. It causes no exclusion-label flip in any seed, preserves all ten radius coverages and all $64$ active-face certificates, and leaves every multiplier lower bound below its oracle secant.

\paragraph{Accept/retrain decision rule.}
\label{sec:accept-retrain}
For a predeclared occupancy-weighted tolerance $\tau$, the operational rule is
\[
\Rad_{\rm hyb}\le\tau
\quad\Longrightarrow\quad
\text{certify the current policy at tolerance }\tau.
\]
Otherwise the rule recommends further primal training, tighter dual optimization, or a larger certification budget. At $\tau=0.05$, $6$ of $10$ policies are certified; at $0.075$, $9$ are certified; and at $0.10$, all ten are certified. Seed 20260737 is the informative rejection: its external error is $0.06037$, its valid radius is $0.09812$, and the rule declines certification at $\tau=0.075$. The same seed certifies only one of $29$ externally excluded audited states. The framework therefore preserves distance coverage for a materially inaccurate policy while withholding unsupported structural conclusions. These are single-policy decisions under each policy's own occupancy law; value comparison remains the role of the CE brackets.

\section{Two-asset state-dependent validation}
\label{sec:two-asset}

The two-asset experiment repeats the frozen-checkpoint and held-out-certification pipeline of \Cref{sec:predictable} with a genuinely multivariate action and three interacting faces. It rules out a scalar-action explanation while retaining an independently refinable external reference; \Cref{sec:highdim-merton} separately isolates action-dimension scaling.

\subsection{Model and exact dual feasibility}
\label{subsec:two-asset-model-dual}

The feasible set is the simplex
\begin{equation}
\label{eq:two-asset-simplex}
\pi_1\ge0,
\qquad
\pi_2\ge0,
\qquad
\pi_1+\pi_2\le1.
\end{equation}
Let $Z_{1,k},Z_{2,k},Z_{3,k}$ be independent standard normals. Wealth and the predictor follow
\begin{align}
\label{eq:two-asset-wealth}
\log X_{k+1}
&=
\log X_k
+
\left[r+\pi_k^\top b(Y_k)-\frac12\pi_k^\top\Sigma\pi_k\right]h
+\sqrt h\,\pi_k^\top A(Z_{1,k},Z_{2,k})^\top,\\
\label{eq:two-asset-factor}
Y_{k+1}
&=
(1-\kappa_Yh)Y_k
+\sigma_Y\sqrt h
\left(\rho_1Z_{1,k}+\rho_2Z_{2,k}+\rho_3Z_{3,k}\right),
\qquad
\rho_3=\sqrt{1-\rho_1^2-\rho_2^2},
\end{align}
where $AA^\top=\Sigma$ and
\begin{equation}
\label{eq:two-asset-premia}
b_1(y)=0.045+0.105\tanh(y/0.25),
\qquad
b_2(y)=0.060-0.120\tanh(y/0.25).
\end{equation}
The remaining parameters are $T=1$, $K=30$, $r=0.02$, $\gamma=3$, $\sigma_1=0.18$, $\sigma_2=0.22$, return correlation $0.25$, $\kappa_Y=1.5$, $\sigma_Y=0.30$, $(\rho_1,\rho_2)=(-0.35,0.15)$, and $(X_0,Y_0)=(1,0)$. \Cref{fig:two-asset-structure} shows the three resulting faces: asset~1 is excluded at low predictor values, asset~2 at high values, and the borrowing cap binds in both tails.

The dual network outputs a bounded traded distortion $d_\phi(t,y)=(d_{1,\phi},d_{2,\phi})\in\R^2$ together with the orthogonal-tilt coefficients $(\theta_{\perp,\phi},c_{\perp,\phi})$ for $Z_{3,k}$. Only $d_\phi$ enters the simplex recovery:
\begin{equation}
\label{eq:two-asset-canonical-dual}
\xi_\phi=\max\{0,-\min_i d_{i,\phi}\},
\qquad
\nu_\phi=d_\phi+\xi_\phi\mathbf 1.
\end{equation}
The coordinates of $\nu_\phi$ are the no-short multipliers, $\xi_\phi$ is the borrowing-cap multiplier, and $\theta_{\parallel,\phi}=A^\top\Sigma^{-1}(b+d_\phi)$ is the reconstructed traded market price of risk. Hence $(d_\phi,\theta_{\perp,\phi},c_{\perp,\phi})$ are neural outputs, whereas $(\nu_\phi,\xi_\phi,\theta_{\parallel,\phi})$ are algebraic outputs; exact feasibility and row-level attribution follow from the canonical decomposition.

\begin{figure}[!htbp]
\centering
\includegraphics[width=\textwidth]{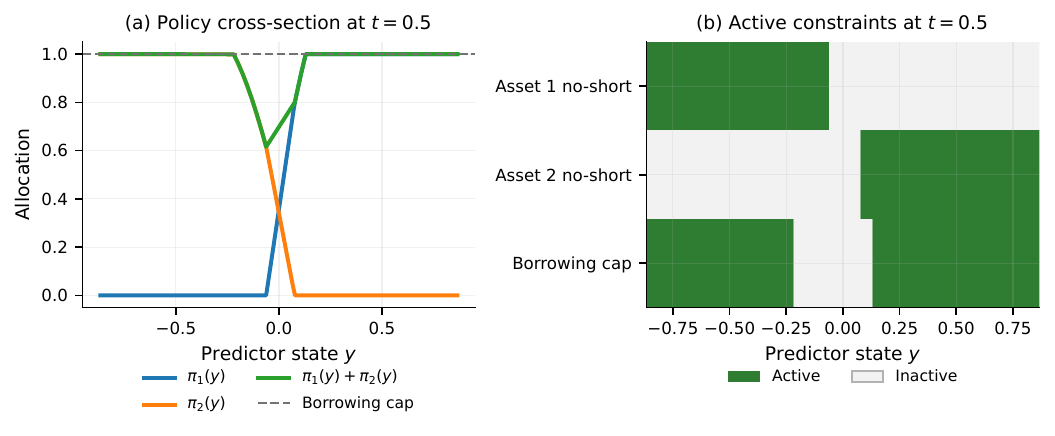}
\caption{Two-asset policy structure at mid-horizon. The policy cross-section exhibits switching between the two risky assets and binding total-risk exposure; the constraint strips display simultaneous activity of a no-short face and the shared borrowing cap in the tails.}
\label{fig:two-asset-structure}
\end{figure}

\subsection{Multivariate policy-distance and face constructions}
\label{subsec:two-asset-constructions}

Conditional restarts and the multivariate Gaussian square term give a lower curvature coefficient $\underline q_k^{\rm hyb}(y)$ in the covariance metric; the Electronic Companion provides the derivation. Define the associated total weight
\begin{equation}
\label{eq:two-asset-outer-weight}
\mathcal W_2
:=
\E\sum_{k=0}^{K-1}h(X_k^\pi)^{1-\gamma}\underline q_k^{\rm hyb}(Y_k),
\end{equation}
and let $\underline{\mathcal W}_2>0$ be a joint one-sided lower confidence bound obtained by the same stratified random-time construction as in \Cref{sec:statistical-layers}. Normalize the weight in \eqref{eq:two-asset-outer-weight} to $\bar\rho_2^\pi$ and set $\|u\|_{\bar\rho_2,\Sigma}^2:=\int u^\top\Sigma u\,d\bar\rho_2^\pi$. With $\|u\|_{\bar\rho_2,v}$ from \eqref{eq:directional-curvature}, the covariance-metric and projected-policy radii are
\begin{equation}
\label{eq:two-asset-radii}
\|\pi-\pi^\star\|_{\bar\rho_2,\Sigma}
\le
\Rad_\Sigma
:=
\left(\frac{2\overline{\mathcal R}}{\underline{\mathcal W}_2}\right)^{1/2},
\qquad
\|\pi-\pi^\star\|_{\bar\rho_2,v}
\le
\Rad_v
:=
\left(\frac{2\overline{\mathcal R}\,v^\top\Sigma^{-1}v}{\underline{\mathcal W}_2}\right)^{1/2}.
\end{equation}
For active-face tests, one current row of \eqref{eq:two-asset-simplex} is relaxed at a time. The current candidate action is the exact relaxed myopic quadratic-program solution, and the frozen deployed feedback resumes at the next date. This choice affects power but not validity.

\subsection{Locked protocol and policy-distance results}
\label{subsec:two-asset-policy-results}

The five-seed protocol was fixed before inspection and inherits the restart and confidence architecture from the one-asset study. Exact path budgets, checkpoint rules, and the declared five-point relaxation grid are documented in the Electronic Companion; the external dynamic program is reserved for post hoc validation.

The mean covariance-norm radius is $0.01327$, compared with mean external error $0.00340$; the seedwise radius-to-error ratio ranges from $3.60$ to $4.83$. Its mean ratio, $3.95$, is close to the one-asset value $3.76$, so moving to two-dimensional controls with three interacting faces does not materially worsen the observed conservatism. All four reported directions in \eqref{eq:two-asset-radii} also cover in every seed. The restarted continuation lower bound retains $98.1\%$ of the external coefficient on average and at least $95.5\%$ over all audited states. The mean outer-weight ESS fraction is $89.6\%$.

\subsection{Active-face and KKT-ledger results}
\label{subsec:two-asset-face-results}

\Cref{fig:two-asset-summary} summarizes the five-seed results: the procedure certifies $59/140$ asset-1 no-short, $46/120$ asset-2 no-short, and $6/114$ borrowing-cap face--state pairs ($111/374$ overall), with no false declaration. Every multiplier lower bound lies below its fine-grid oracle secant. The KKT ledger is genuinely multiconstraint: averaged over seeds, the terminal Fenchel residual contributes $40.4\%$ of the positive residual, asset-1 no-short complementarity $13.3\%$, asset-2 no-short complementarity $21.6\%$, and the borrowing cap $24.6\%$. The low cap-face resolving rate despite a material ledger share illustrates that loss attribution and positive multiplier certification are different tasks.

The covariance-norm radius covers in all five seeds, and the four reported directions---the two coordinates, aggregate exposure $(1,1)$, and spread $(1,-1)$---also cover in every seed. Mean radius-to-error ratios are $3.95$, $7.03$, $4.86$, $6.26$, and $5.53$, respectively. The complete directional table is in the Electronic Companion.

\begin{figure}[!htbp]
\centering
\includegraphics[width=0.88\textwidth]{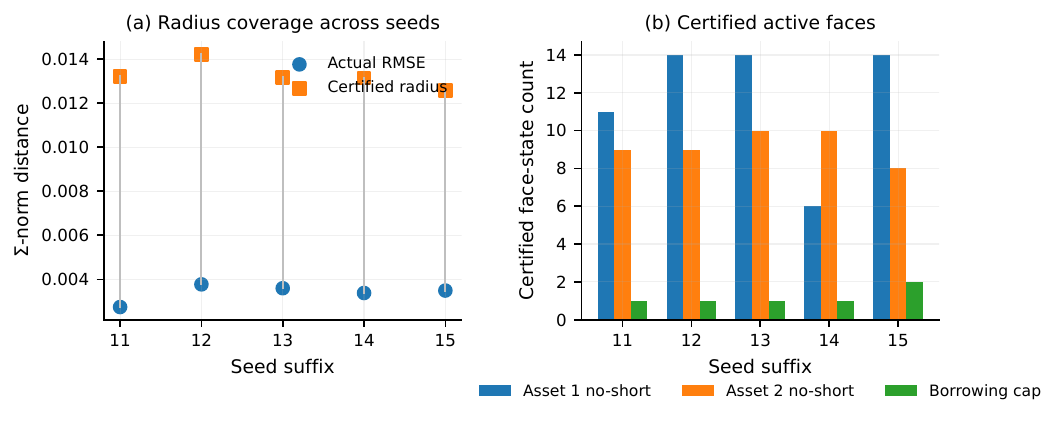}
\caption{Locked five-seed two-asset validation. Every covariance-norm radius covers its external error, and all three face types are certified without a false declaration.}
\label{fig:two-asset-summary}
\end{figure}

\subsection{External dynamic-program refinement}
\label{subsec:two-asset-dp-refinement}

The external program uses active-set Newton steps and is refined through $401\times5^3$, $801\times7^3$, and $1201\times9^3$ predictor/quadrature levels. The last two levels differ in CE by $2.03\times10^{-6}$ and in policy RMSE by $5.04\times10^{-5}$. Detailed computational and recall accounting is in the Electronic Companion.

\section{Discussion and conclusion}
\label{sec:discussion-conclusion}

A primal--dual bracket can support more than a value statement. On the declared grid, the Fenchel--complementarity representation and its Doob compensation make the residual auditable. Bellman curvature converts a residual upper bound into a policy region, paired relaxation lower-bounds a binding-constraint multiplier, and the resolution theorem shows how variance and margin determine the required path budget.

The audits cover every reported external policy error and make no false face declaration. Exact-wrapper replication through $50$ assets shows that action dimension alone need not degrade the certificate; the conservative state-dependent pilot instead identifies tight learned-dual construction as the practical bottleneck. The guarantees remain local to the grid and deployed occupancy law, and uncertified faces are unresolved. Tighter learned duals, sharper simultaneous inference, and transaction-cost extensions remain open. Within this scope, the method answers two questions unavailable from a value bracket alone: where the optimal policy can lie and which tested constraints must bind.

\phantomsection\label{last:main-body}

\clearpage
\begin{center}
{\LARGE\bfseries Electronic Companion}\\[0.35em]
{\large to \emph{From Value Bounds to Policy-Distance and Active-Face Certificates:\\
Same-Grid Duality for Constrained Dynamic Portfolios}}
\end{center}
\vspace{0.75em}

This Electronic Companion provides the full proofs, numerical audits, and implementation details supporting the main paper's same-grid policy certificates. We use CE, RMSE, UCB, and CRN for certainty equivalent, root-mean-square error, upper confidence bound, and common random numbers, respectively. The sections below document the locked state-dependent protocols, derive the multivariate curvature bound, report dynamic-program convergence and Merton calibration, present the locked high-dimensional exact-wrapper stress test and the learned state-dependent dual diagnostic, and give computational-budget and wall-clock accounting. Reference solutions are used only for external validation after each certificate has been computed.

\setcounter{section}{0}
\setcounter{subsection}{0}
\setcounter{figure}{0}
\setcounter{table}{0}
\setcounter{equation}{0}
\setcounter{assumption}{0}
\setcounter{theorem}{0}
\setcounter{proposition}{0}
\setcounter{corollary}{0}
\setcounter{lemma}{0}
\setcounter{remark}{0}
\renewcommand{\thesection}{S\arabic{section}}
\renewcommand{\thesubsection}{S\arabic{section}.\arabic{subsection}}
\renewcommand{\thefigure}{S\arabic{figure}}
\renewcommand{\thetable}{S\arabic{table}}
\renewcommand{\theequation}{S\arabic{section}.\arabic{equation}}
\renewcommand{\theassumption}{S\arabic{assumption}}
\renewcommand{\thetheorem}{S\arabic{theorem}}
\renewcommand{\theproposition}{S\arabic{proposition}}
\renewcommand{\thecorollary}{S\arabic{corollary}}
\renewcommand{\thelemma}{S\arabic{lemma}}
\renewcommand{\theremark}{S\arabic{remark}}
\renewcommand{\theHsection}{supp.S\arabic{section}}
\renewcommand{\theHsubsection}{supp.S\arabic{section}.\arabic{subsection}}
\renewcommand{\theHfigure}{supp.S\arabic{figure}}
\renewcommand{\theHtable}{supp.S\arabic{table}}
\renewcommand{\theHequation}{supp.S\arabic{section}.\arabic{equation}}
\providecommand{\theHassumption}{}
\providecommand{\theHtheorem}{}
\providecommand{\theHproposition}{}
\providecommand{\theHcorollary}{}
\providecommand{\theHlemma}{}
\providecommand{\theHremark}{}
\renewcommand{\theHassumption}{supp.S\arabic{assumption}}
\renewcommand{\theHtheorem}{supp.S\arabic{theorem}}
\renewcommand{\theHproposition}{supp.S\arabic{proposition}}
\renewcommand{\theHcorollary}{supp.S\arabic{corollary}}
\renewcommand{\theHlemma}{supp.S\arabic{lemma}}
\renewcommand{\theHremark}{supp.S\arabic{remark}}

\section{Additional theoretical and implementation details}
\label{ec:additional-theory}

\subsection{Normalized orthogonal density tilt}
\label{ec:normalized-orthogonal-tilt}
For a scalar orthogonal standard normal $z$, the state-dependent experiments use
\begin{equation}
\label{eq:quadratic-orthogonal-tilt}
\Lambda^{\perp}(z;\alpha,c_\perp)
=
\exp\{\alpha z+c_\perp(z^2-1)-\psi(\alpha,c_\perp)\},
\qquad c_\perp<\frac12,
\end{equation}
where
\begin{equation}
\label{eq:quadratic-log-normalizer}
\psi(\alpha,c_\perp)
=
-c_\perp-\frac12\log(1-2c_\perp)
+\frac{\alpha^2}{2(1-2c_\perp)}.
\end{equation}
Then $\E[\Lambda^\perp]=1$ exactly. The linear parameter $\alpha=-\theta_\perp\sqrt h$ represents the unspanned market-price-of-risk anchor, while $c_\perp$ permits a normalized nonlinear correction without changing the traded-asset martingale restriction. The normalizer requires $c_\perp<1/2$. Since $p=(\gamma-1)/\gamma<1$, this same structural bound also implies $pc_\perp<1/2$, so the moment required by the dual CE is finite. We reserve $\chi_v$ in the main paper for the directional inverse-curvature factor.

\subsection{Finite-sample policy-distance alternative}
\label{ec:finite-sample-alternative}
This subsection provides an operational finite-sample fallback for the policy-distance ellipsoid. It complements, rather than duplicates, the main paper's general finite-sample validity-and-resolution theorem: the result below gives a direct Cantelli plug-in once deterministic variance bounds are supplied, whereas the main theorem also analyzes resolution power and sample-size requirements.
\begin{proposition}[Finite-sample one-sided plug-in under variance bounds]
\label{prop:cantelli-policy-radius}
Suppose the observations above are independent and identically distributed within each bank, with
\[
\operatorname{Var}(Z_i^{\rm gap})\le v_{\rm gap},
\qquad
\operatorname{Var}(W_i)\le v_W.
\]
For $\alpha_{\mathcal R},\alpha_W\in(0,1)$ set
\begin{align}
\label{eq:cantelli-gap}
\overline{\mathcal R}_{\mathrm{Can}}
&=
\widehat{\mathcal R}_{\rm comp}
+
\sqrt{\frac{v_{\rm gap}(1-\alpha_{\mathcal R})}
{N_{\mathcal R}\alpha_{\mathcal R}}},\\
\label{eq:cantelli-weight}
\underline{\mathcal W}_{\mathrm{Can}}
&=
\widehat{\mathcal W}
-
\sqrt{\frac{v_W(1-\alpha_W)}
{N_W\alpha_W}}.
\end{align}
Then, with probability at least $1-\alpha_{\mathcal R}-\alpha_W$,
\[
G(\pi)
\le
\mathcal R(y,\mu;\pi)
\le
\overline{\mathcal R}_{\mathrm{Can}},
\qquad
\mathcal W_\rho\ge\underline{\mathcal W}_{\mathrm{Can}}.
\]
Whenever $\underline{\mathcal W}_{\mathrm{Can}}>0$, substituting these quantities into \eqref{eq:raw-ellipsoid} and \eqref{eq:directional-radius} gives a finite-sample valid policy-distance certificate.
\end{proposition}
\prooflocation{See \Cref{proof:cantelli-policy-radius}.}

\begin{remark}[Reported implementation]
The empirical tables use sample splitting and the normal one-sided bounds in \eqref{eq:normal-gap-ucb}--\eqref{eq:normal-weight-lcb}, with Bonferroni allocation across the jointly required events. Their coverage statements are therefore asymptotic rather than distribution free. \Cref{prop:cantelli-policy-radius} records a finite-sample route when deterministic variance bounds are available; median-of-means or other robust one-sided estimators can be substituted without changing the deterministic policy geometry.
\end{remark}

\begin{remark}[Constructive variance bounds in the bounded implementation]
\label{rem:constructive-variance}
In the bounded model class used below, deterministic variance bounds can in principle be made explicit. The risky-weight set is compact, the multiplier and orthogonal-tilt outputs are clipped to declared bounds, and the log-wealth dynamics have bounded predictable coefficients with Gaussian increments. Consequently, the terminal utility, the dual payoff, the support-function defect, the accumulated compensator, and the positive curvature weights possess explicit finite moments of every order required here. Combining their elementary moment bounds can yield conservative deterministic choices of $v_{\rm gap}$ and $v_W$. We use empirical variances and asymptotic one-sided intervals in the reported experiments because those analytic bounds are substantially looser; once deterministic choices are instantiated, \Cref{prop:cantelli-policy-radius} supplies a fully finite-sample route to the same policy geometry.
\end{remark}

The policy-distance certificate is now complete: a residual UCB and a curvature-weight LCB are inserted into the deterministic ellipsoid. Exact support identification requires asymmetric information that a symmetric radius cannot provide, and is therefore treated separately next.

\FloatBarrier

\section{Locked replication, relaxation, and restart audits}
\label{app:locked-replication}

\paragraph{Protocol separation.}
The first three predictable-return seeds were available during method development. The additional seven seeds were run only after fixing the primal training budget, checkpoint-selection rule, feasible-dual architecture, common-test protocol, hybrid-curvature formula, $96$ restart states, $16{,}384$ conditional paths, and the seven-point relaxation grid. No seed-specific hyperparameter tuning was permitted. \Cref{tab:ec-ten-seed-certificates} reports both groups separately so that development evidence is not presented as fully held-out replication.\par
The full two-asset protocol was fixed before inspecting any result from the five reported seeds. It inherited the restart count, conditional-path budget, and confidence-allocation architecture from the one-asset development phase; the three-face relaxed-action rule and five-point relaxation grid were declared before the reported runs. Primal checkpoints and dual snapshots were selected by fixed held-out rules, with no seed-specific tuning.

\begin{table}[!htbp]
\centering
\caption{Locked-protocol one-asset policy and active-face certificates. Seeds 20260731--20260733 are development seeds and 20260734--20260740 are locked replications.}
\label{tab:ec-ten-seed-certificates}
\scriptsize
\begin{tabular}{lcrrrrr}
\toprule
Seed & Role & Residual UCB & Fine error & $\Rad_{\rm hyb}$ & Radius/error & Certified exclusions\\
\midrule
20260731 & Dev. & $2.71\times10^{-5}$ & 0.00645 & 0.02584 & 4.01 & 11/29\\
20260732 & Dev. & $1.19\times10^{-4}$ & 0.01386 & 0.05408 & 3.90 & 4/29\\
20260733 & Dev. & $8.28\times10^{-5}$ & 0.01553 & 0.04515 & 2.91 & 4/29\\
20260734 & Rep. & $7.78\times10^{-5}$ & 0.00769 & 0.04376 & 5.69 & 9/29\\
20260735 & Rep. & $1.31\times10^{-4}$ & 0.00931 & 0.05686 & 6.11 & 4/29\\
20260736 & Rep. & $1.87\times10^{-4}$ & 0.02547 & 0.06794 & 2.67 & 2/29\\
20260737 & Rep. & $3.91\times10^{-4}$ & 0.06037 & 0.09812 & 1.63 & 1/29\\
20260738 & Rep. & $3.59\times10^{-5}$ & 0.00825 & 0.02974 & 3.61 & 11/29\\
20260739 & Rep. & $3.75\times10^{-5}$ & 0.00654 & 0.03039 & 4.65 & 10/29\\
20260740 & Rep. & $1.00\times10^{-4}$ & 0.02066 & 0.04972 & 2.41 & 8/29\\
\midrule
Mean/total & -- & $1.19\times10^{-4}$ & 0.01741 & 0.05016 & 3.76 & 64/290\\
\bottomrule
\end{tabular}
\end{table}

\paragraph{Relaxation-grid robustness.}
In the development seeds, the candidate grids are
\[
\mathcal E_5=\{0.005,0.01,0.02,0.04,0.08\},
\qquad
\mathcal E_7=\mathcal E_5\cup\{0.12,0.16\}.
\]
Enlarging $\mathcal E_5$ to $\mathcal E_7$ changes the certified count from $18$ to $19$. Applying the more stringent seven-candidate critical value to the original five candidates loses no certificate. The selected relaxation is $0.16$ in only two of the $19$ states, so the result is not primarily an upper-endpoint artifact.

\paragraph{Residual UCB versus structural resolving power.}
Across the ten seeds, the Pearson and Spearman correlations between the positive residual UCB and the certified audited-state fraction are $-0.800$ and $-0.960$, respectively. The five lower-gap seeds certify $45/145=31.0\%$ of externally excluded audited states, while the five higher-gap seeds certify $19/145=13.1\%$. The sample is too small for these numbers to be treated as a general statistical law. They are consistent with the exact paired identity: the relaxation gain must exceed the original primal--dual residual before a positive multiplier can be certified.

\section{Two-asset state-dependent certification}
\label{app:two-asset-details}

\subsection{Model and multivariate hybrid curvature}
\label{ec:two-asset-curvature}

The market in \Cref{sec:two-asset} has two traded Gaussian shocks and one orthogonal predictor shock. For $\beta=1-\gamma<0$, write the reduced one-step objective as
\begin{equation}
\label{eq:two-asset-reduced-q}
\mathcal Q_k^{(2)}(y,\pi;q)
=
\E\left[q(Y_{k+1})e^{\beta\ell_k(\pi,y,Z)}\mid Y_k=y\right],
\end{equation}
where
\[
\ell_k(\pi,y,Z)
=
\left[r+\pi^\top b(y)-\frac12\pi^\top\Sigma\pi\right]h
+\sqrt h\,\pi^\top AZ_{1:2}.
\]
Define
\begin{equation}
\label{eq:two-asset-u}
u_k(\pi,y,Z)
=
(b(y)-\Sigma\pi)h+\sqrt h\,AZ_{1:2},
\qquad
\mathcal L_{k+1}=q_{k+1}^\star(Y_{k+1})e^{\beta\ell_k}.
\end{equation}

\begin{proposition}[Multivariate hybrid curvature]
\label{prop:two-asset-hybrid-curvature}
The action Hessian satisfies
\begin{equation}
\label{eq:two-asset-hessian}
\nabla_{\pi\pi}^2\mathcal Q_k^{(2)}
=
\E\left[
\mathcal L_{k+1}
\left
\{(\gamma-1)h\Sigma+(\gamma-1)^2u_ku_k^\top\right\}
\mid Y_k=y
\right].
\end{equation}
Let $\underline q_k^{\rm rst}(y)\le q_k^\star(y)$ and let $\underline q_{k+1}^{\rm unif}\le q_{k+1}^\star(y')$ uniformly over the declared predictor domain. Define
\begin{align}
\label{eq:two-asset-cmin}
\underline c_k(y)
&=
\min_{\pi\in\cK}
\exp\left((1-\gamma)h
\left[r+\pi^\top b(y)-\frac\gamma2\pi^\top\Sigma\pi\right]\right),\\
\label{eq:two-asset-qhyb}
\underline q_k^{\rm hyb}(y)
&=
\underline q_k^{\rm rst}(y)
+(\gamma-1)\underline q_{k+1}^{\rm unif}\underline c_k(y).
\end{align}
Then
\begin{equation}
\label{eq:two-asset-performance-transfer}
V_0^\star-J(\pi)
\ge
\frac12\E\sum_{k=0}^{K-1}
h(X_k^\pi)^{1-\gamma}
\underline q_k^{\rm hyb}(Y_k)
(\pi_k-\pi_k^\star)^\top\Sigma(\pi_k-\pi_k^\star).
\end{equation}
Consequently, the covariance-metric norm bound and the projected-policy seminorm bounds in \eqref{eq:two-asset-radii} follow from the same joint one-sided construction as \Cref{thm:restarted-operational-radius}.
\end{proposition}

\begin{proof}
Differentiating the exponential in \eqref{eq:two-asset-reduced-q} gives \eqref{eq:two-asset-hessian}. Moreover,
\[
\mathcal Q_k^{(2)}(y,\pi;q_{k+1}^\star)\ge q_k^\star(y),
\]
so the first positive term yields $(\gamma-1)h\underline q_k^{\rm rst}(y)\Sigma$. For the square term, Gaussian exponential tilting gives
\begin{align}
\label{eq:two-asset-gaussian-matrix}
\E\left[e^{\beta\ell_k}u_ku_k^\top\right]
&=
c_k(y,\pi)
\left[
h\Sigma+h^2(b(y)-\gamma\Sigma\pi)(b(y)-\gamma\Sigma\pi)^\top
\right]\\
&\succeq c_k(y,\pi)h\Sigma,
\end{align}
where
\[
c_k(y,\pi)
=
\exp\left((1-\gamma)h
\left[r+\pi^\top b(y)-\frac\gamma2\pi^\top\Sigma\pi\right]\right).
\]
Replacing $q_{k+1}^\star$ by its uniform lower bound and $c_k$ by \eqref{eq:two-asset-cmin} gives the second term in \eqref{eq:two-asset-qhyb}. The Bellman performance-difference telescoping cancels the factor $\gamma-1$ and proves \eqref{eq:two-asset-performance-transfer}. Because $1-\gamma<0$, the minimizer in \eqref{eq:two-asset-cmin} is the constrained static QP maximizer of $b(y)^\top\pi-\frac\gamma2\pi^\top\Sigma\pi$.
\end{proof}

The two-asset dual uses one four-output neural map
\[
(t,y)\longmapsto
\bigl(d_{1,\phi}(t,y),d_{2,\phi}(t,y),
\theta_{\perp,\phi}(t,y),c_{\perp,\phi}(t,y)\bigr).
\]
The first two outputs form the bounded traded distortion $d_\phi(t,y)\in\R^2$. The canonical decomposition \eqref{eq:two-asset-canonical-dual} then recovers $(\nu_\phi,\xi_\phi)$, and $\theta_{\parallel,\phi}=A^\top\Sigma^{-1}(b+d_\phi)$ is computed algebraically. The remaining two outputs parameterize the analytically normalized orthogonal density tilt for the unspanned third shock. Thus exact traded-part feasibility holds for every network parameter, while the orthogonal outputs improve tightness without altering the budget identity.

\subsection{Locked protocol and policy-radius results}
\label{ec:two-asset-locked-radius}

The five seeds $20260811$--$20260815$ use the fixed protocol described in \Cref{app:locked-replication}: $4{,}000$ finite-horizon primal iterations after a static-QP anchor, $2{,}500$ feasible dual-refinement iterations, $262{,}144$ final residual paths, $96$ restart states, and $16{,}384$ conditional paths per state. Primal checkpoints and dual snapshots are selected on held-out splits without the external DP; \Cref{tab:two-asset-seeds,tab:two-asset-directions} report the resulting seed-level and directional audits.

\begin{table}[t]
\centering
\caption{Locked five-seed two-asset results. Seed-specific checkpoint and dual-step entries are outcomes of the same predeclared held-out selection rule, not seed-specific tuning. ``Faces'' reports certified over externally active face--state pairs.}
\label{tab:two-asset-seeds}
\scriptsize
\begin{tabular}{lrrrrrrr}
\toprule
Seed & Primal ckpt. & Dual step & Residual UCB & $\Rad_\Sigma$ & External error & Radius/error & Faces\\
\midrule
20260811 & 3000 & 2500 & $1.356\times10^{-4}$ & 0.01322 & 0.00274 & 4.83 & 21/66\\
20260812 & 2000 & 2500 & $1.531\times10^{-4}$ & 0.01422 & 0.00377 & 3.77 & 24/76\\
20260813 & 4000 & 1500 & $1.348\times10^{-4}$ & 0.01318 & 0.00360 & 3.66 & 25/80\\
20260814 & 4000 & 2500 & $1.334\times10^{-4}$ & 0.01315 & 0.00338 & 3.89 & 17/84\\
20260815 & 4000 & 2500 & $1.216\times10^{-4}$ & 0.01258 & 0.00349 & 3.60 & 24/68\\
\midrule
Mean/total & -- & -- & $1.357\times10^{-4}$ & 0.01327 & 0.00340 & 3.95 & 111/374\\
\bottomrule
\end{tabular}
\end{table}

\begin{table}[t]
\centering
\caption{Two-asset projected-policy seminorm certificates, averaged over five seeds.}
\label{tab:two-asset-directions}
\small
\begin{tabular}{lrrrr}
\toprule
Direction/metric & Mean radius & Mean external error & Mean radius/error & Coverage\\
\midrule
$\Sigma$ norm & 0.01327 & 0.00340 & 3.95 & 5/5\\
Asset 1 coordinate & 0.07614 & 0.01091 & 7.03 & 5/5\\
Asset 2 coordinate & 0.06229 & 0.01328 & 4.86 & 5/5\\
Aggregate $(1,1)$ & 0.08548 & 0.01393 & 6.26 & 5/5\\
Spread $(1,-1)$ & 0.10977 & 0.02002 & 5.53 & 5/5\\
\bottomrule
\end{tabular}
\end{table}

The mean restart ratio $\underline q_k^{\rm rst}/q_k^\star$ is $0.981$, the minimum over all audited states is $0.955$, and the mean outer-weight ESS fraction is $0.896$. Conditional dual CE coverage is one in every seed.

\subsection{Active faces, relaxation scale, and the KKT ledger}
\label{ec:two-asset-active-faces}

At each restart state, the three current faces are relaxed separately for
\[
\varepsilon\in\{0.01,0.02,0.04,0.08,0.16\}.
\]
The current candidate action is the exact relaxed static-QP solution under the declared coefficients; from the next date onward, the frozen deployed feedback is used. Calibration uses $8{,}192$ paths per state and the paired final test uses $32{,}768$. One simultaneous Bonferroni family covers $96\times3\times5=1{,}440$ state--face--relaxation candidates per seed.

\begin{table}[t]
\centering
\caption{Two-asset active-face and KKT-ledger summaries. Residual means are componentwise averages across seeds; gap fractions are averages of seedwise fractions and therefore need not equal ratios of the displayed means.}
\label{tab:two-asset-face-ledger}
\small
\begin{tabular}{lrrrr}
\toprule
\multicolumn{5}{l}{\textit{Panel A: active-face certificates}}\\
Face & Certified & Externally active & Certified fraction & False\\
\midrule
Asset 1 no-short & 59 & 140 & 42.1\% & 0\\
Asset 2 no-short & 46 & 120 & 38.3\% & 0\\
Borrowing cap & 6 & 114 & 5.3\% & 0\\
Total & 111 & 374 & 29.7\% & 0\\
\midrule
\multicolumn{5}{l}{\textit{Panel B: positive-residual ledger}}\\
Component & Mean residual & Gap fraction & \multicolumn{2}{c}{Interpretation}\\
\midrule
Terminal Fenchel & $5.419\times10^{-5}$ & 40.4\% & \multicolumn{2}{c}{terminal marginal mismatch}\\
Asset 1 no-short & $1.805\times10^{-5}$ & 13.3\% & \multicolumn{2}{c}{row-1 complementarity}\\
Asset 2 no-short & $2.943\times10^{-5}$ & 21.6\% & \multicolumn{2}{c}{row-2 complementarity}\\
Borrowing cap & $3.367\times10^{-5}$ & 24.6\% & \multicolumn{2}{c}{shared-cap complementarity}\\
\bottomrule
\end{tabular}
\end{table}

The face-specific counts and ledger decomposition are reported in \Cref{tab:two-asset-face-ledger}, while \Cref{fig:two-asset-diagnostics} shows the external policy transitions and the seedwise gap--power relation. The selected relaxation is $0.16$ for $53$ of the $59$ asset-1 certificates, $28$ of the $46$ asset-2 certificates, and $4$ of the $6$ cap certificates. The remaining selections are $0.08$ except for two asset-2 certificates at $0.04$. The results are therefore finite-secant certificates over a declared relaxation scale, not estimates of infinitesimal derivatives. The simultaneous family includes every candidate, and all $111$ multiplier lower bounds remain below their fine-grid oracle secants.

\begin{figure}[t]
\centering
\begin{subfigure}{0.49\textwidth}
\centering
\includegraphics[width=\textwidth]{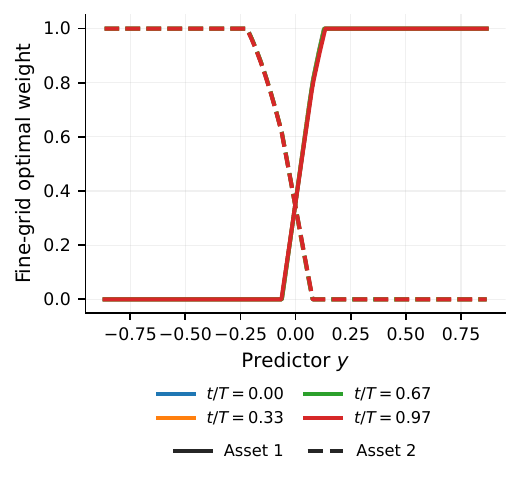}
\caption{Fine-grid optimal policies and face transitions.}
\end{subfigure}
\begin{subfigure}{0.49\textwidth}
\centering
\includegraphics[width=\textwidth]{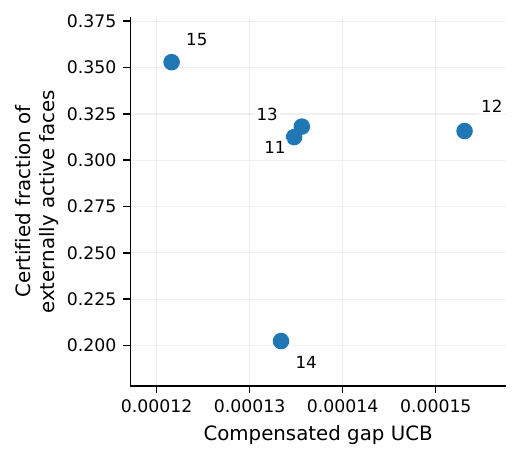}
\caption{Residual UCB and certified fraction by seed.}
\end{subfigure}
\caption{Two-asset external diagnostics. The dynamic program and face labels are queried only after the certificates are fixed. The five-seed scatter is a mechanism diagnostic, not a scaling law.}
\label{fig:two-asset-diagnostics}
\end{figure}

The borrowing cap accounts for roughly one quarter of the positive residual but is certified in only $6/114$ externally active state--face pairs. A ledger entry measures the deployed policy's complementary-slackness loss, whereas multiplier certification requires the gain from a specific relaxed candidate to exceed the entire base residual and the simultaneous sampling margin. The two quantities need not rank faces identically.

\subsection{Post hoc relaxation-endpoint robustness}
\label{ec:relaxation-endpoint-robustness}
\label{sec:two-asset-endpoint-robustness}

The locked five-point protocol above is retained as the primary result. Because the upper endpoint $\varepsilon=0.16$ was selected frequently, we ran a post hoc endpoint audit that extended the declared grid to
\[
\{0.01,0.02,0.04,0.08,0.16,0.24,0.32\}.
\]
The count and selected-scale effects are summarized in \Cref{tab:two-asset-endpoint-robustness} and \Cref{fig:two-asset-endpoint-robustness}. The audit separates the multiplicity cost from the benefit of the two new candidates. The original five-point family certifies $111$ face--state pairs. Applying the stricter seven-candidate threshold while retaining only the original five candidates yields $109$ certificates. Allowing all seven candidates yields $112$. Thus the stricter simultaneous threshold removes two certificates and the new relaxation scales add three, for a net change of one. Every extended-grid multiplier lower bound remains below its fine-grid oracle secant, and the procedure again produces no false declaration.

\begin{table}[!htbp]
\centering
\caption{Two-asset endpoint-robustness audit. The locked result remains the original five-grid count.}
\label{tab:two-asset-endpoint-robustness}
\small
\begin{tabular}{lrrr}
\toprule
Protocol & Certified & Externally active & False\\
\midrule
Locked five-grid threshold & 111 & 374 & 0\\
Original five candidates under seven-grid threshold & 109 & 374 & 0\\
Extended seven-grid protocol & 112 & 374 & 0\\
\bottomrule
\end{tabular}
\end{table}

\begin{figure}[!htbp]
\centering
\begin{subfigure}{0.49\textwidth}
\centering
\includegraphics[width=\textwidth]{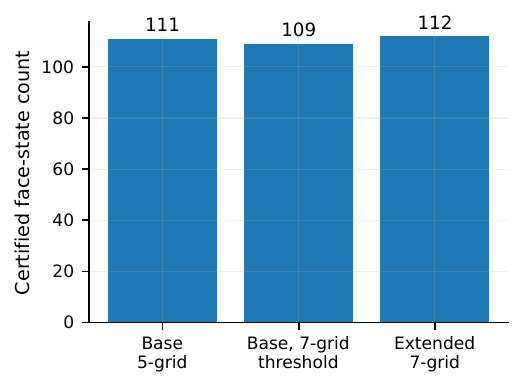}
\caption{Multiplicity cost and net certificate change.}
\end{subfigure}
\begin{subfigure}{0.49\textwidth}
\centering
\includegraphics[width=\textwidth]{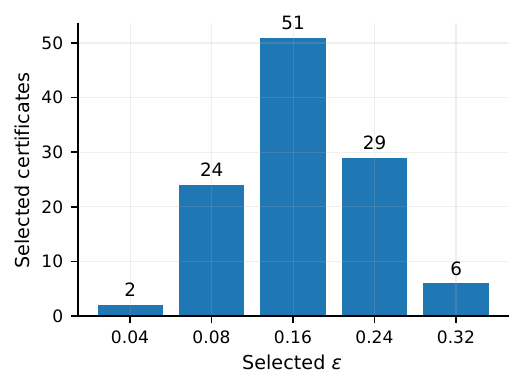}
\caption{Selected relaxation scales under the extended protocol.}
\end{subfigure}
\caption{Post hoc two-asset endpoint robustness. Extending the grid from $0.16$ to $0.32$ changes the total certified count only from $111$ to $112$, with no false declaration or oracle-secant violation.}
\label{fig:two-asset-endpoint-robustness}
\end{figure}

Among the $112$ extended-grid certificates, the selected relaxation counts are $2$, $24$, $51$, $29$, and $6$ at $\varepsilon=0.04,0.08,0.16,0.24,$ and $0.32$, respectively. Hence $35/112$ certificates select a scale above the original endpoint, but only three face--state pairs become newly certifiable after accounting for the stronger multiplicity correction. The larger scales therefore improve many finite-secant multiplier lower bounds without materially changing the structural classification. This supports the robustness of the locked five-point conclusion while confirming that the reported multiplier bounds should be interpreted as finite-scale secants rather than infinitesimal derivatives.

\subsection{External dynamic-program refinement}
\label{ec:two-asset-dp-refinement}

CRRA homogeneity leaves one continuous state variable despite the two-dimensional action. At every predictor node, the Bellman minimization is solved by active-set Newton iterations rather than a dense two-dimensional action grid. The external program uses three-dimensional Gauss--Hermite quadrature for the two traded and one unspanned shocks.

\begin{table}[t]
\centering
\caption{Two-asset external dynamic-program refinement. Policy RMSE is measured against the preceding level.}
\label{tab:two-asset-dp}
\small
\begin{tabular}{rrrrr}
\toprule
Predictor points & Gauss--Hermite (GH) order & CE & CE change & Policy RMSE\\
\midrule
401 & 5 & 1.063435394 & -- & --\\
801 & 7 & 1.063424777 & $-1.062\times10^{-5}$ & $1.194\times10^{-4}$\\
1201 & 9 & 1.063422746 & $-2.031\times10^{-6}$ & $5.037\times10^{-5}$\\
\bottomrule
\end{tabular}
\end{table}

The refinement in \Cref{tab:two-asset-dp} leaves a final policy change that is small relative to the external errors in \Cref{tab:two-asset-seeds}. No certificate uses the DP value, policy, or active labels in training, selection, residual construction, curvature estimation, or multiplier testing.

\section{External dynamic-programming convergence and validation}
\label{app:dp-validation}

As shown in \Cref{fig:predictable-dp} and \Cref{tab:predictable-dp}, the change between the two finest CE levels is $4.54\times10^{-7}$, and the coarse DP used during exploratory analysis differs from the finest value by only $3.02\times10^{-6}$. This convergence is much smaller than the common-test policy gaps and supports using the finest DP as an external numerical reference. No theorem or certificate depends on its exactness.

\begin{figure}[!htbp]
\centering
\begin{subfigure}{0.49\textwidth}
\centering
\includegraphics[width=\textwidth]{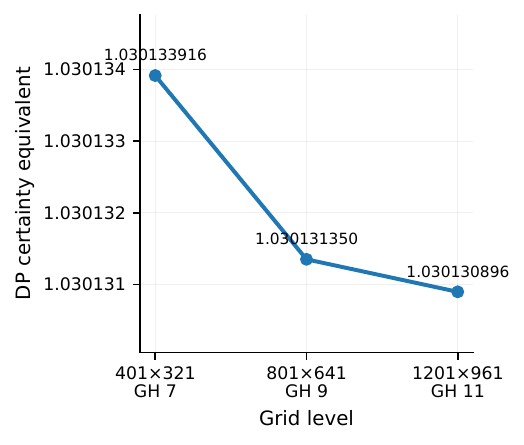}
\caption{Certainty-equivalent convergence.}
\end{subfigure}
\begin{subfigure}{0.49\textwidth}
\centering
\includegraphics[width=\textwidth]{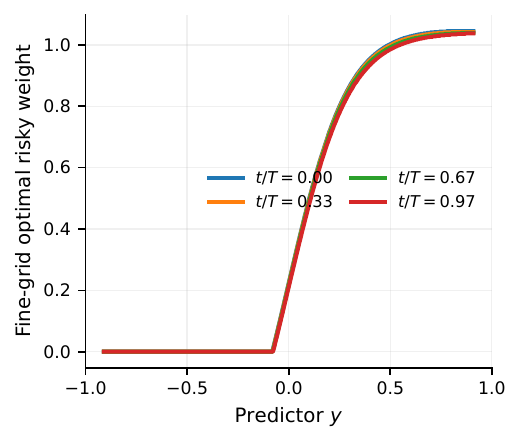}
\caption{State-dependent fine-grid policy.}
\end{subfigure}
\caption{External same-grid dynamic-programming diagnostics. The DP is never used by primal or dual learning.}
\label{fig:predictable-dp}
\end{figure}

\begin{table}[!htbp]
\centering
\caption{Fine-grid dynamic-programming convergence. Grid levels report predictor points $\times$ action points $\times$ two-dimensional Gauss--Hermite order.}
\label{tab:predictable-dp}
\small
\begin{tabular}{lrrrr}
\toprule
Grid & CE & $|\mathrm{CE}-\mathrm{CE}_{\mathrm{fine}}|$ & Policy RMSE & $\pi_0(Y_0)$\\
\midrule
$401\times321\times7$ & 1.030133916 & $3.02\times10^{-6}$ & $8.70\times10^{-4}$ & 0.225000\\
$801\times641\times9$ & 1.030131350 & $4.54\times10^{-7}$ & $5.52\times10^{-4}$ & 0.222656\\
$1201\times961\times11$ & 1.030130896 & 0 & 0 & 0.223438\\
\bottomrule
\end{tabular}
\end{table}

\section{Solved Merton calibration and statistical audit}

\subsection{Configuration}
\label{ec:merton-configuration}

This solved benchmark serves as a calibration laboratory for the certificate engine rather than as an application benchmark. It follows the classical continuous-time portfolio structure of \citet{Merton1969,Merton1971} and is deliberately used as a controlled validation case. Because the optimal constrained Merton weight is constant and \eqref{eq:log-euler} is the exact constant-weight GBM step, the grid and continuous-time optima coincide here. The experiment therefore validates feasibility restoration, KKT accounting, and statistical calibration; it does not demonstrate a grid--continuum mismatch. We use $n=d_W=5$, $T=1$, $K=20$, $r=0.03$, $\gamma=3$, and $x_0=1$. A fixed random loading matrix is row-normalized to volatilities between $0.16$ and $0.24$. The excess-return vector is generated from an alternating-sign unconstrained target and then held fixed. The resulting analytical solution is
\begin{equation}
\label{eq:numerical-opt}
\pi^\star
=
(0.034163,\ 0,\ 0.157889,\ 0,\ 0.141818)^\top,
\end{equation}
\begin{equation}
\nu^\star
=
(0,\ 0.002732,\ 0,\ 0.013454,\ 0)^\top,
\qquad
\CE^\star=1.0339299139.
\end{equation}

The direct-policy-optimization network has three hidden layers of width $128$, SiLU activations, and componentwise sigmoid outputs. It is trained for $2{,}500$ Adam iterations with batch size $8{,}192$. A candidate dual is recovered from each frozen checkpoint, projected into the exact feasible class, and then refined before evaluation. Unless otherwise stated, paired residual evaluation uses $262{,}144$ paths. Certificate validity depends on the final feasibility and held-out audit, not on the initialization used for the dual refinement.

\subsection{Checkpoint sweep and feasible-dual recovery}
\label{ec:merton-checkpoint-sweep}

\begin{figure}[t]
\centering
\includegraphics[width=0.72\textwidth]{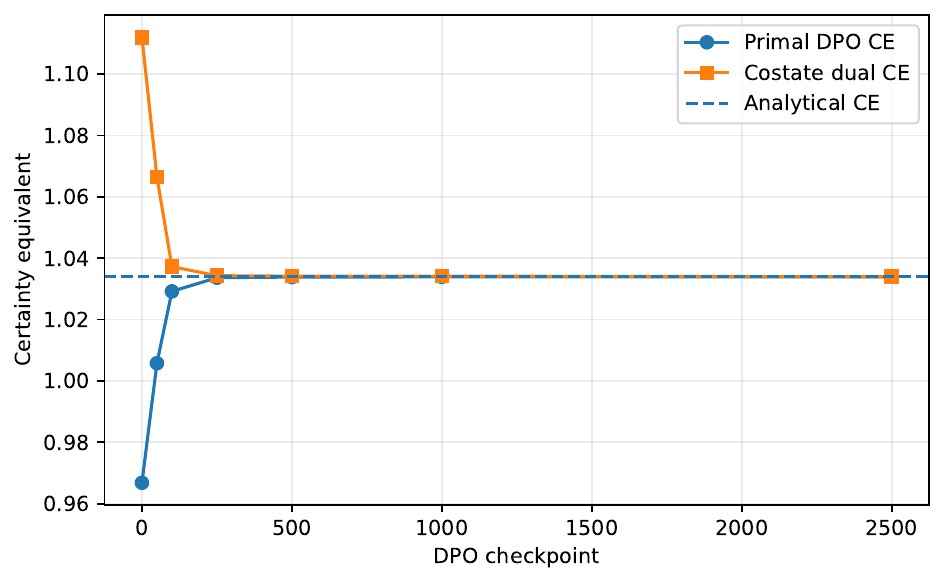}
\caption{Primal and dual certainty equivalents across the frozen-policy checkpoint sweep. As direct policy optimization progresses, the primal value rises and the recovered feasible dual falls toward the same analytical optimum.}
\label{fig:checkpoint}
\end{figure}

\Cref{fig:checkpoint} and \Cref{tab:checkpoint} show that better frozen policies are accompanied by tighter recovered duals. The candidate-kernel RMSE falls from $0.2643$ at initialization to $0.00166$ at checkpoint $2{,}500$, while the recovered dual CE falls from $1.11195$ to $1.0339323$, only $2.30\times10^{-4}\%$ above the analytical optimum. These diagnostics are external to certificate validity: every reported residual is computed only after the candidate dual has been transformed into the exact feasible class.

\begin{table}[t]
\centering
\caption{Selected checkpoint diagnostics in the solved Merton benchmark. Kernel and policy errors use the analytical solution only for external validation.}
\label{tab:checkpoint}
\small
\setlength{\tabcolsep}{5pt}
\begin{tabular}{rcccc}
\toprule
Checkpoint & Primal CE & Recovered dual CE & Kernel RMSE & Policy action RMSE\\
\midrule
0    & 0.966801 & 1.111947 & $2.643\times10^{-1}$ & $4.462\times10^{-1}$\\
50   & 1.005772 & 1.066414 & $1.607\times10^{-1}$ & $3.511\times10^{-1}$\\
100  & 1.029166 & 1.037233 & $5.692\times10^{-2}$ & $1.457\times10^{-1}$\\
250  & 1.033623 & 1.034272 & $8.006\times10^{-3}$ & $1.773\times10^{-2}$\\
500  & 1.033830 & 1.034115 & $4.008\times10^{-3}$ & $7.551\times10^{-3}$\\
1000 & 1.033883 & 1.034116 & $2.087\times10^{-3}$ & $4.177\times10^{-3}$\\
2500 & 1.033902 & 1.033932 & $1.657\times10^{-3}$ & $2.925\times10^{-3}$\\
\bottomrule
\end{tabular}
\end{table}

\subsection{Doob compensation versus common random numbers}
\label{ec:merton-compensation-audit}

\begin{figure}[t]
\centering
\begin{subfigure}{0.49\textwidth}
\centering
\includegraphics[width=\textwidth]{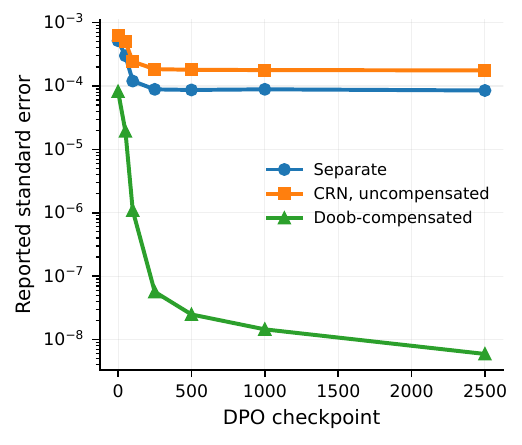}
\caption{Separate, uncompensated CRN, and compensated standard errors.}
\end{subfigure}
\begin{subfigure}{0.49\textwidth}
\centering
\includegraphics[width=\textwidth]{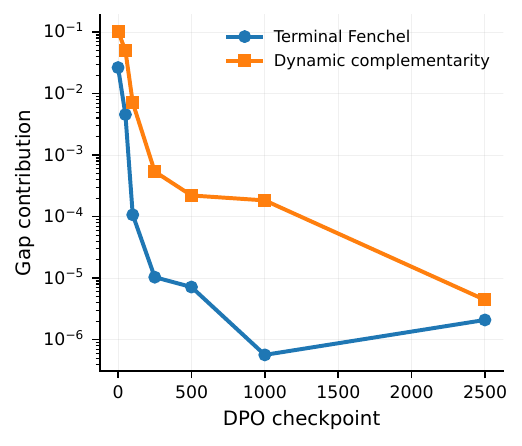}
\caption{Positive compensated-residual components.}
\end{subfigure}
\caption{Common random numbers alone do not remove the terminal budget martingale. Doob compensation isolates the positive Fenchel and complementarity residuals and produces near-zero-variance behavior close to optimality.}
\label{fig:gap}
\end{figure}

\Cref{fig:gap} visualizes the estimator and residual-component comparisons, and \Cref{tab:gap} reports their checkpoint values. At checkpoint $2{,}500$,
\[
\widehat{\mathcal R}_{\mathrm{comp}}
=
6.5260\times10^{-6},
\qquad
\SE(\widehat{\mathcal R}_{\mathrm{comp}})
=
5.92\times10^{-9}.
\]
The separate estimator has standard error $8.50\times10^{-5}$, while the uncompensated CRN estimator has standard error $1.76\times10^{-4}$ and sample mean $-3.70\times10^{-4}$. The latter is a direct empirical illustration of the martingale term in \eqref{eq:crn-doob}: path coupling alone does not create a certificate. In this benchmark, the squared ratio of separate to compensated standard errors is $2.06\times10^8$. The positive residual decomposes as
\[
2.0912\times10^{-6}
\quad\text{(terminal Fenchel)}
\]
plus
\[
4.4348\times10^{-6}
\quad\text{(budget/complementarity)}.
\]
The remaining certified defect is therefore driven primarily by dynamic complementary slackness rather than terminal marginal matching.

\begin{table}[t]
\centering
\caption{Estimator comparison. The CRN estimator is unbiased but uncompensated and can be negative. ``Path gain'' is the squared ratio of separate and compensated standard errors.}
\label{tab:gap}
\scriptsize
\setlength{\tabcolsep}{3.6pt}
\begin{tabular}{rcccccc}
\toprule
Checkpoint & Residual mean & Separate SE & CRN mean & CRN SE & Compensated SE & Path gain\\
\midrule
100  & $7.225\times10^{-3}$ & $1.197\times10^{-4}$ & $7.539\times10^{-3}$ & $2.423\times10^{-4}$ & $1.090\times10^{-6}$ & $1.21\times10^4$\\
250  & $5.507\times10^{-4}$ & $8.853\times10^{-5}$ & $6.255\times10^{-4}$ & $1.846\times10^{-4}$ & $5.658\times10^{-8}$ & $2.45\times10^6$\\
500  & $2.283\times10^{-4}$ & $8.669\times10^{-5}$ & $3.043\times10^{-4}$ & $1.804\times10^{-4}$ & $2.488\times10^{-8}$ & $1.21\times10^7$\\
1000 & $1.843\times10^{-4}$ & $8.867\times10^{-5}$ & $3.139\times10^{-4}$ & $1.783\times10^{-4}$ & $1.447\times10^{-8}$ & $3.75\times10^7$\\
2500 & $6.526\times10^{-6}$ & $8.500\times10^{-5}$ & $-3.704\times10^{-4}$ & $1.764\times10^{-4}$ & $5.920\times10^{-9}$ & $2.06\times10^8$\\
\bottomrule
\end{tabular}
\end{table}

\subsection{Merton policy-radius validation}
\label{ec:merton-radius-validation}

\begin{figure}[t]
\centering
\includegraphics[width=0.72\textwidth]{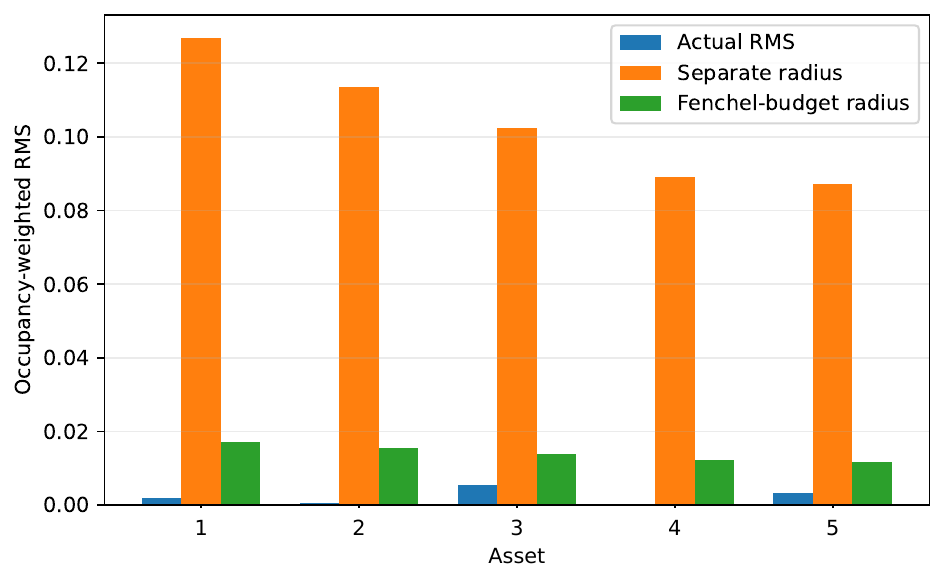}
\caption{Actual coordinate projected-policy seminorm errors (equivalently, occupancy-weighted coordinate RMS errors) and certified radii at checkpoint $2{,}500$. The positive Fenchel--budget gap reduces the mean radius from $0.10385$ under separate evaluation to $0.01409$, while retaining coverage in every coordinate.}
\label{fig:coordinate}
\end{figure}

\begin{table}[t]
\centering
\caption{Merton policy-radius validation at checkpoint $2{,}500$. Inactive-coordinate bounds use the analytical $\nu_i^\star$ and are oracle diagnostics.}
\label{tab:coordinate}
\small
\begin{tabular}{rccccc}
\toprule
Asset &
$\pi_i^\star$ &
Actual RMS &
Separate radius &
Fenchel--budget radius &
Inactive mean / bound\\
\midrule
1 & 0.034163 & 0.002033 & 0.126765 & 0.017202 & --\\
2 & 0        & 0.000453 & 0.113597 & 0.015415 & $0.000450/0.002558$\\
3 & 0.157889 & 0.005520 & 0.102480 & 0.013907 & --\\
4 & 0        & 0.000084 & 0.089239 & 0.012110 & $0.000083/0.000519$\\
5 & 0.141818 & 0.003304 & 0.087154 & 0.011827 & --\\
\bottomrule
\end{tabular}
\end{table}

The coordinate comparison in \Cref{fig:coordinate} and the detailed values in \Cref{tab:coordinate} show that all five projected-policy seminorm errors are covered. The two inactive-coordinate mean checks also cover, but they use analytical multipliers and are not operational certificates. The table uses the analytical continuation weight to validate the geometry directly. The operational construction in \Cref{cor:operational-radius} replaces that coefficient by the dual-derived lower coefficient. At this checkpoint the deterministic coefficient ratio is at least $0.9999954$, so the corresponding radius inflation is at most $1.0000023$ before the held-out lower-confidence adjustment for the weight. Thus the coordinate radius is operational within the Merton class; only the inactive-coordinate column remains oracle.

\subsection{Replication audit and benchmark interpretation}
\label{ec:merton-replication-audit}

For checkpoints $500$, $1000$, and $2500$, we construct a $2^{21}$-path reference and run 20 independent replications at
\[
M\in\{2^{12},2^{14},2^{16},2^{18}\}.
\]
Across the resulting $240$ cells, the nominal one-sided $95\%$ Fenchel--budget UCB covers the reference gap in $228$ cases:
\[
228/240=95.0\%.
\]
Coordinate and inactive-coordinate certificates cover in every tested cell.

\begin{figure}[t]
\centering
\includegraphics[width=0.68\textwidth]{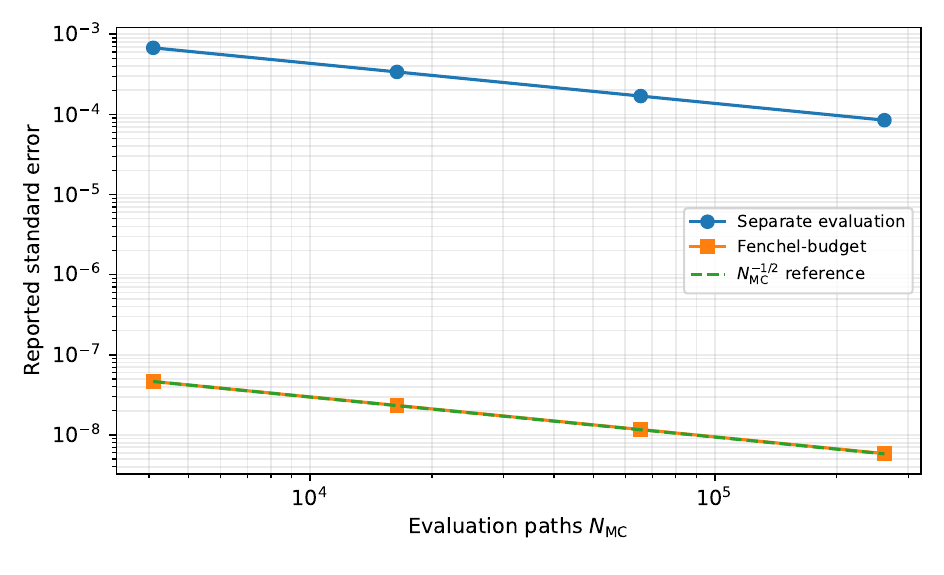}
\caption{Standard-error scaling at checkpoint $2{,}500$. The reported Fenchel--budget standard error follows $N_{\rm MC}^{-1/2}$, but at a level several orders of magnitude below separate evaluation.}
\label{fig:scaling}
\end{figure}

The scaling plot and numerical values in \Cref{fig:scaling} and \Cref{tab:scaling-appendix} give fitted slopes of reported Fenchel--budget SE against $N_{\rm MC}$ of
\[
-0.5008,\qquad -0.5000,\qquad -0.4980
\]
at checkpoints $500$, $1000$, and $2500$, respectively. At $N_{\rm MC}=2^{16}$, the empirical effective path gains are approximately
\[
1.18\times10^7,\qquad
1.63\times10^7,\qquad
2.39\times10^8.
\]
At the final checkpoint the mean coordinate radius is essentially unchanged as $N_{\rm MC}$ increases:
\[
0.014175,\ 0.014122,\ 0.014097,\ 0.014084.
\]
This indicates that statistical uncertainty has already become negligible by $N_{\rm MC}=2^{12}$; the remaining radius reflects the actual primal--dual policy discrepancy rather than Monte Carlo noise.

\begin{table}[ht]
\centering
\caption{Reported standard-error scaling in the no-short replication audit.}
\label{tab:scaling-appendix}
\small
\begin{tabular}{rcc}
\toprule
Paths & Separate reported SE & Compensated reported SE\\
\midrule
$2^{12}$ & $6.773\times10^{-4}$ & $4.639\times10^{-8}$\\
$2^{14}$ & $3.387\times10^{-4}$ & $2.333\times10^{-8}$\\
$2^{16}$ & $1.693\times10^{-4}$ & $1.168\times10^{-8}$\\
$2^{18}$ & $8.468\times10^{-5}$ & $5.850\times10^{-9}$\\
\bottomrule
\end{tabular}
\end{table}

The benchmark closes three downstream components: same-grid validity, polyhedral dual recovery, and compensated statistical certification. It is a controlled calibration study rather than a claim about the relative merits of alternative policy or dual-learning algorithms.

\section{High-dimensional Merton stress test}
\label{ec:highdim-merton}

The two experiments in this section separate action dimension from state-dependent dual approximation. The exact-wrapper Merton study asks whether action dimension and polyhedral geometry alone degrade the certificate; the predictor-driven pilot then restores learned state dependence to locate the remaining bottleneck. For each $d\in\{5,20,50\}$, the exact-wrapper feasible set is
\[
\mathcal K_d=\{\pi\in\R^d:\pi\ge0,\;\mathbf 1^\top\pi\le1\}.
\]
A short feasible neural policy is trained without access to the primal or dual optimum and then frozen. The exactly feasible dual wrapper is obtained from the deterministic convex program
\begin{equation}
\label{eq:highdim-direct-dual}
\min_{\nu\ge0,\,\lambda\ge0}
\left\{
\lambda+
\frac{1}{2\gamma}
(b+\nu-\lambda\mathbf 1)^\top
\Sigma^{-1}
(b+\nu-\lambda\mathbf 1)
\right\},
\end{equation}
where the cap equals one. The compensated residual and covariance-metric radius are fixed on $131{,}072$ locked paths before the external primal QP is solved. Three pilot markets were used to freeze the protocol; the five reported market seeds were new, and no hyperparameter was changed after the pilot.

\begin{table}[!htbp]
\centering
\caption{Locked high-dimensional Merton replication by dimension.}
\label{tab:ec-highdim-dimension}
\small
\begin{tabular}{rcccccc}
\toprule
$d$ & Cells & Radius coverage & Gap-UCB coverage & Mean ratio & Max ratio & Mean inactive\\
\midrule
5  & 5 & $5/5$ & $5/5$ & 1.005 & 1.008 & 1.0\\
20 & 5 & $5/5$ & $4/5$ & 1.007 & 1.023 & 7.8\\
50 & 5 & $5/5$ & $5/5$ & 1.005 & 1.008 & 32.2\\
\bottomrule
\end{tabular}
\end{table}

As reported in \Cref{tab:ec-highdim-dimension,tab:ec-highdim-cells}, all policy-radius certificates cover. With the locked $10^{-7}$ activity tolerance, the direct dual and external primal QP recover identical coordinatewise no-short active sets in all $15$ cells (Jaccard index $1.000$, with no false positive or false negative), and cap activity also agrees in every cell. The single nominal gap-UCB miss occurs at seed 20260831 and $d=20$:
\[
\overline{\mathcal R}=5.8201054\times10^{-5},
\qquad
G_{\rm ext}=5.8205448\times10^{-5},
\]
a difference of $4.39\times10^{-9}$. Applying the familywise sensitivity level $z_{1-0.025/15}=2.9352$ covers all $15$ external gaps and changes the maximum radius-to-error ratio from $1.0227$ to $1.0246$.

\begin{table}[!htbp]
\centering
\caption{Cell-level locked high-dimensional Merton results. ``Gap covered'' uses the prespecified cellwise $97.5\%$ UCB.}
\label{tab:ec-highdim-cells}
\scriptsize
\begin{tabular}{rrrrrrc}
\toprule
Market seed & $d$ & Residual UCB & Radius & External error & Radius/error & Gap covered\\
\midrule
20260831 & 5 & $4.754e-07$ & $0.0006$ & $0.0006$ & 1.005 & yes \\
20260831 & 20 & $5.820e-05$ & $0.0066$ & $0.0066$ & 1.000 & no \\
20260831 & 50 & $8.353e-05$ & $0.0080$ & $0.0079$ & 1.007 & yes \\
20260832 & 5 & $1.841e-07$ & $0.0004$ & $0.0004$ & 1.008 & yes \\
20260832 & 20 & $1.092e-05$ & $0.0029$ & $0.0029$ & 1.003 & yes \\
20260832 & 50 & $7.220e-05$ & $0.0074$ & $0.0074$ & 1.003 & yes \\
20260833 & 5 & $4.580e-07$ & $0.0006$ & $0.0006$ & 1.005 & yes \\
20260833 & 20 & $6.682e-05$ & $0.0071$ & $0.0069$ & 1.023 & yes \\
20260833 & 50 & $1.751e-04$ & $0.0116$ & $0.0115$ & 1.005 & yes \\
20260834 & 5 & $6.048e-08$ & $0.0002$ & $0.0002$ & 1.005 & yes \\
20260834 & 20 & $8.533e-05$ & $0.0080$ & $0.0080$ & 1.007 & yes \\
20260834 & 50 & $1.958e-04$ & $0.0123$ & $0.0122$ & 1.008 & yes \\
20260835 & 5 & $6.109e-07$ & $0.0007$ & $0.0007$ & 1.003 & yes \\
20260835 & 20 & $2.643e-05$ & $0.0044$ & $0.0044$ & 1.004 & yes \\
20260835 & 50 & $5.126e-05$ & $0.0063$ & $0.0062$ & 1.003 & yes \\
\bottomrule
\end{tabular}
\end{table}

The median certificate-audit times are $0.079$, $0.079$, and $0.091$ seconds at dimensions $5$, $20$, and $50$, respectively, after a separate GPU and QP warm-up. The full 15-cell run takes $4.41$ seconds. These timings concern the exact-wrapper Merton stress test; they do not imply that a generic learned state-dependent dual can be trained in comparable time.

\subsection{State-dependent high-dimensional dual diagnostic}
\label{ec:state-dependent-highdim-diagnostic}
The exact-wrapper experiment above separates action dimension from dual approximation. To diagnose the remaining bottleneck, we also retained the predictor-driven model and learned an exactly feasible state-dependent dual at $d=20$ and $50$. The certificate used $64$ restart states and $8{,}192$ conditional paths per state; the dynamic program was queried only after the residual and radius were fixed.

\begin{table}[!htbp]
\centering
\caption{State-dependent high-dimensional pilot. The certificates remain valid but become conservative as the learned feasible dual loosens.}
\label{tab:ec-state-dependent-highdim}
\small
\begin{tabular}{rccccccc}
\toprule
$d$ & Residual UCB & Weight LCB & Radius & External error & Radius/error & Weight ESS & Coverage\\
\midrule
20 & $5.32\times10^{-3}$ & 0.787 & 0.1163 & 0.00408 & 28.5 & 0.839 & yes\\
50 & $9.45\times10^{-3}$ & 0.831 & 0.1509 & 0.00337 & 44.8 & 0.910 & yes\\
\bottomrule
\end{tabular}
\end{table}

The pilot summary in \Cref{tab:ec-state-dependent-highdim} shows restart continuation ratios averaging $0.980$ and $0.974$, with conditional-dual CE coverage equal to one in both dimensions. Thus neither curvature recovery nor outer-weight concentration collapses. Instead, the residual rises through both components: the terminal Fenchel and dynamic complementarity means are $(3.46,1.83)\times10^{-3}$ at $d=20$ and $(6.55,2.85)\times10^{-3}$ at $d=50$. The selected dual remains at the final declared refinement checkpoint in both cells. The contrast with the exact-wrapper results isolates learned feasible-dual tightness---rather than action dimension, curvature recovery, outer-weight concentration, or residual-to-policy transfer---as the practical bottleneck. The experiment is therefore reported as a conservative pilot rather than as a high-dimensional learned-dual success.

\FloatBarrier

\section{Computational budget and wall-clock accounting}
\label{ec:computational-accounting}

The state-dependent studies are deliberately low dimensional because they admit external dynamic-program validation. In these benchmark cases the DP is much cheaper than the full neural-policy, feasible-dual, restart, and face-certification pipeline. \Cref{tab:ec-computational-budgets,tab:ec-wall-clock} report the declared budgets and stagewise timings; the intended comparison is inferential validity, not computational dominance over DP.

\begin{table}[!htbp]
\centering
\caption{Declared state-dependent certification budgets.}
\label{tab:ec-computational-budgets}
\small
\begin{tabular}{lccc}
\toprule
Study & Radius audit & Face family & Paired face paths\\
\midrule
One asset & $96\times16{,}384$ & $96\times7=672$ & $16{,}384$ calibration; $32{,}768$ test\\
Two assets & $96\times16{,}384$ & $96\times3\times5=1{,}440$ & $8{,}192$ calibration; $32{,}768$ test\\
\bottomrule
\end{tabular}
\end{table}

For \Cref{tab:ec-wall-clock}, ``Candidate'' denotes generation of the frozen primal--dual pair: primal-policy training, feasible-dual fitting and refinement, and held-out checkpoint and snapshot selection. The radius and face audits are timed separately.

\begin{table}[!htbp]
\centering
\caption{Approximate median wall-clock accounting for the reported state-dependent runs. Times are reconstructed from sequential stage-output timestamps and are hardware specific.}
\label{tab:ec-wall-clock}
\scriptsize
\begin{tabular}{lccccc}
\toprule
Study & Candidate & Radius & Face & Total & Fine DP\\
\midrule
One asset (7 new seeds) & $16.5$ min & $9$ s & $96$ s & $18.9$ min & $0.79$ s\\
Two assets (5 seeds) & $19.8$ min & $9$ s & $223$ s & $23.6$ min & $1.34$ s\\
\bottomrule
\end{tabular}
\end{table}

The face audits use a single Bonferroni family, prioritizing false-declaration control over recall. This multiplicity correction, together with the requirement that the relaxation gain exceed the full base residual, helps explain the certified fractions of $64/290$ and $111/374$. Step-down or resampling-based simultaneous procedures could improve power while retaining one-sided familywise control, but are not evaluated here.

\section{Implementation details}

The solver-neutral objects and their neural realizations are mapped in \Cref{tab:ec-computational-instantiation}.

\begin{table}[!htbp]
\centering
\caption{Map from solver-neutral certificate objects to the state-dependent neural realizations. Each study uses one vector-output dual network.}
\label{tab:ec-computational-instantiation}
\scriptsize
\renewcommand{\arraystretch}{1.08}
\begin{tabular}{p{0.27\textwidth}p{0.65\textwidth}}
\toprule
Theoretical role & Computational realization\\
\midrule
Primal policy $\pi_\vartheta(t,s)$
& Feasible neural feedback trained by direct policy optimization with BPTT and then frozen.\\
One-asset dual output
& One neural map $(t,y)\mapsto(\nu_\phi,\theta_{\perp,\phi},c_{\perp,\phi})$, with $\nu_\phi\ge0$ and $c_{\perp,\phi}<1/2$.\\
One-asset algebraic wrapper
& Set $d_\phi=\nu_\phi$ and $\theta_{\parallel,\phi}=(b(y)+\nu_\phi)/\sigma_S$; $(\theta_{\perp,\phi},c_{\perp,\phi})$ define the normalized orthogonal tilt.\\
Two-asset dual output
& One neural map $(t,y)\mapsto(d_{1,\phi},d_{2,\phi},\theta_{\perp,\phi},c_{\perp,\phi})$ with all four outputs bounded.\\
Two-asset algebraic wrapper
& Recover $\xi_\phi=\max\{0,-\min_i d_{i,\phi}\}$, $\nu_\phi=d_\phi+\xi_\phi\mathbf1$, and $\theta_{\parallel,\phi}=A^\top\Sigma^{-1}(b+d_\phi)$; the last two neural outputs define the normalized orthogonal tilt.\\
Certification stage
& Evaluate frozen primal and exactly feasible dual objects on held-out same-grid paths to obtain residual, radius, and face tests.\\
\bottomrule
\end{tabular}
\end{table}

All policies, dual fields, and residual calculations use double precision. The no-short covariance-metric projection is solved exactly by active-set enumeration in five dimensions, and the two-asset simplex decomposition uses the canonical map in \eqref{eq:two-asset-canonical-dual}. In the one-asset anchor fit, $\nu_\phi$ and $\theta_{\perp,\phi}$ are fitted to costate-based targets while $c_{\perp,\phi}$ is initially regularized toward zero. In the two-asset anchor fit, $d_\phi$ is fitted to the policy-transport target while both orthogonal outputs are initially regularized toward zero. All outputs are then refined jointly under the feasible dual objective. In both studies the traded market price of risk is reconstructed algebraically, the orthogonal density factor is normalized by \eqref{eq:quadratic-log-normalizer}, and snapshot selection, scalar-$y$ calibration, and final certification use disjoint Brownian banks.

The locked ten-seed radius audit uses $96$ independent on-policy restart states and $16{,}384$ conditional dual paths per state. Its active-face audit uses $16{,}384$ calibration paths, $32{,}768$ paired test paths, and one Bonferroni family over $96$ states and seven relaxation levels. The two-asset audit uses the same restart count and conditional budget, $8{,}192$ calibration paths, $32{,}768$ paired test paths, and one family over three faces and five relaxation levels. Its two-dimensional Bellman actions and relaxed candidates are solved by active-set quadratic-program/Newton routines. In both state-dependent studies, the residual and curvature tail budgets use normal one-sided bounds, and the external dynamic program is queried only after each radius or face certificate is computed.

The solved Merton checkpoint sweep uses a three-layer width-$128$ primal network trained for $2{,}500$ Adam iterations with batch size $8{,}192$, followed by exact feasibility restoration and feasible dual refinement. Unless otherwise stated, the paired residual audit uses $262{,}144$ paths. The replication study uses a $2^{21}$-path reference and 20 independent replications at each of four path budgets. These implementation choices make the audits reproducible; the certification theory itself does not depend on a particular primal or dual learning algorithm.

\section{Full proofs}
\label{app:proofs}

\subsection[Proof of the polyhedral budget theorem and CE identities]{Proof of \Cref{thm:polyhedral-budget} and \eqref{eq:y-general}--\eqref{eq:ce-general}}
\label{proof:polyhedral-budget}

Combining \eqref{eq:log-euler} and \eqref{eq:deflator-general}, conditional Gaussian integration gives
\begin{align*}
\E_k\left[\frac{H_{k+1}X_{k+1}}{H_kX_k}\right]
&=
\exp\left(
 h\pi_k^\top b_k
 -hg^\top\mu_k
 -h\pi_k^\top A_k\theta_k
\right)\\
&=
\exp\left(
 h\pi_k^\top b_k
 -hg^\top\mu_k
 -h\pi_k^\top(b_k-\Gamma^\top\mu_k)
\right)\\
&=
\exp\{-h\mu_k^\top(g-\Gamma\pi_k)\}.
\end{align*}
This proves \eqref{eq:general-conditional-budget}. Summing the expected decrements proves \eqref{eq:general-budget-telescope}. Fenchel's inequality
\[
U(X_K)\le U^\star(yH_K)+yH_KX_K
\]
and the budget inequality prove \eqref{eq:general-dual-bound}.

For $p=(\gamma-1)/\gamma$ and $m=\E[H_K^p]$,
\[
D(y,\mu)
=
\frac1{\gamma-1}
-\frac\gamma{\gamma-1}my^p
+yx_0.
\]
The first-order condition is $x_0=my^{-1/\gamma}$, yielding \eqref{eq:y-general}. Substitution gives \eqref{eq:ce-general}.

\subsection[Proof of the orthogonal-enrichment proposition]{Proof of \Cref{prop:orthogonal-enrichment}}
\label{proof:orthogonal-enrichment}

Conditionally on $\F_{t_k}$, the base deflator increment and the wealth return depend on the traded innovation $Z_{k+1}^{\parallel}$, while $\Lambda_{k+1}^{\perp}$ depends only on the conditionally independent orthogonal innovation. Hence
\[
\E_k[H_{k+1}X_{k+1}\Lambda_{k+1}^{\perp}]
=
\E_k[H_{k+1}X_{k+1}]\,
\E_k[\Lambda_{k+1}^{\perp}]
=
H_kX_ke^{-h\kappa_k}.
\]
The telescoping budget inequality and Fenchel bound follow exactly as in \Cref{thm:polyhedral-budget}. For \eqref{eq:quadratic-orthogonal-tilt}, Gaussian integration gives
\[
\log\E[e^{\alpha z+c_\perp(z^2-1)}]
=
-c_\perp-\frac12\log(1-2c_\perp)
+\frac{\alpha^2}{2(1-2c_\perp)},
\]
which is \eqref{eq:quadratic-log-normalizer}.

\subsection[Proof of the LP-attainment lemma]{Proof of \Cref{lem:lp-attainment}}
\label{proof:lp-attainment}

For fixed $d\in\cC_{\cK}$, consider the primal linear program
\[
\sup_{\pi}\{-d^\top\pi:\Gamma\pi\le g\}.
\]
Its Lagrangian is
\[
-d^\top\pi+\mu^\top(g-\Gamma\pi)
=
\mu^\top g-(d+\Gamma^\top\mu)^\top\pi,
\qquad \mu\ge0.
\]
The supremum over $\pi\in\R^n$ is finite exactly when $-\Gamma^\top\mu=d$, in which case the dual objective is $g^\top\mu$. Because $\cK$ is nonempty and $d\in- \Gamma^\top\R_+^m$, both primal and dual are feasible. Finite-dimensional linear-programming strong duality therefore gives
\[
\sigma_{\cK}(d)
=
\sup_{\pi\in\cK}(-d^\top\pi)
=
\min_{\mu\ge0:\,-\Gamma^\top\mu=d}g^\top\mu,
\]
and the dual minimum is attained. This also identifies the least-cost multiplier value with the support function used in constrained-portfolio duality.

\subsection[Proof of the polyhedral-recovery theorem]{Proof of \Cref{thm:polyhedral-recovery}}
\label{proof:polyhedral-recovery}

Because $\widehat\theta^{\mathrm{raw}}\in\operatorname{Range}(A^\top)$,
\[
\widehat\theta^{\mathrm{raw}}
=
A^\top\Sigma^{-1}(b+\widehat d^{\mathrm{raw}}).
\]
For every candidate $d$,
\begin{align*}
\norm{A^\top\Sigma^{-1}(b+d)-\widehat\theta^{\mathrm{raw}}}^2
&=
\norm{A^\top\Sigma^{-1}(d-\widehat d^{\mathrm{raw}})}^2\\
&=
(d-\widehat d^{\mathrm{raw}})^\top\Sigma^{-1}AA^\top\Sigma^{-1}(d-\widehat d^{\mathrm{raw}})\\
&=
(d-\widehat d^{\mathrm{raw}})^\top\Sigma^{-1}(d-\widehat d^{\mathrm{raw}}).
\end{align*}
This proves \eqref{eq:general-nearest-kernel}. The projected $\widehat d$ belongs to $\cC_{\cK}$, so \Cref{lem:lp-attainment} supplies a nonempty closed convex set $\mathcal M(\widehat d)$ of least-cost representations. The coercive strictly convex tie-break in \eqref{eq:canonical-multiplier} is therefore uniquely attained. Exact dual feasibility follows. Finally,
\[
\widehat\mu^\top(g-\Gamma\pi)
=g^\top\widehat\mu+\widehat d^\top\pi
=\sigma_{\cK}(\widehat d)+\widehat d^\top\pi,
\]
which proves \eqref{eq:support-defect}.

\subsection{Proof of the geometric corollaries}
\label{ec:geometric-corollary-proofs}
\label{proof:geometric-corollaries}

For no-short constraints, $\Gamma=-I_n$, so $-\Gamma^\top\R_+^n=\R_+^n$ and the representation is unique. The nearest-kernel equality follows from the calculation in the preceding proof.

For the simplex, the generators are $e_1,\ldots,e_n,-\mathbf1$. Since
\[
-e_i=-\mathbf1+\sum_{j\ne i}e_j,
\]
they positively span $\R^n$. Hence every $d\in\R^n$ has a feasible decomposition $d=\nu-\xi\mathbf1$. Feasibility requires $\xi\ge-\min_i d_i$ and $\xi\ge0$, so the smallest feasible scalar is \eqref{eq:canonical-xi}; then $\nu=d+\xi\mathbf1\ge0$. Since $g^\top\mu=\xi$, this is the least-cost representation.

\subsection[Proof of the positive-residual theorem]{Proof of \Cref{thm:positive-gap}}
\label{proof:positive-gap}

Add and subtract $yH_KX_K$:
\begin{align*}
D(y,\mu)-J(\pi)
&=
\E[U^\star(yH_K)+yx_0-U(X_K)]\\
&=
\E[U^\star(yH_K)+yH_KX_K-U(X_K)]
+y[x_0-\E(H_KX_K)].
\end{align*}
The first term is $\E[F_T(y)]$. The second equals $\E[B_T]$ by \eqref{eq:general-budget-telescope}. Fenchel equality holds exactly when $yH_K=U'(X_K)$, and the compensator vanishes pathwise when $\kappa_k=0$ for every $k$.

\subsection[Proof of the KKT-ledger proposition]{Proof of \Cref{prop:kkt-ledger}}
\label{proof:kkt-ledger}

Nonnegativity follows from $\mu_{j,k}\ge0$ and $s_{j,k}=g_j-\Gamma_j\pi_k\ge0$. If $\kappa_k>0$, the weights $\omega_{j,k}$ are nonnegative and sum to one, so
\[
\sum_jB_{j,k}
=yH_kX_k(1-e^{-h\kappa_k}).
\]
If $\kappa_k=0$, every $c_{j,k}=0$ and both sides vanish. Summing over $j$ and $k$ in \Cref{thm:positive-gap} yields \eqref{eq:kkt-ledger-sum}. The total defect depends only on $(d_k,\pi_k)$ by \eqref{eq:support-defect}; the minimum-norm tie-break makes the row-level allocation reproducible.

\subsection[Proof of the Doob-compensation theorem]{Proof of \Cref{thm:doob-comp}}
\label{proof:doob-comp}

From \eqref{eq:general-conditional-budget} and \eqref{eq:martingale-difference},
\[
H_{k+1}X_{k+1}
=
H_kX_k e^{-h\kappa_k}+\Delta M_{k+1}.
\]
Rearranging and summing over $k$ yields \eqref{eq:doob-budget}. Adding the terminal Fenchel residual proves \eqref{eq:crn-doob}.

\subsection[Proof of the scalar-calibration proposition]{Proof of \Cref{prop:y-floor}}
\label{proof:y-floor}

Let $I=(U')^{-1}$. Since $(U^\star)'(z)=-I(z)$,
\[
\Phi_{x,H}'(y)
=
H[x-I(yH)].
\]
The unique zero is $y^\star=U'(x)/H$. Also
\[
\Phi_{x,H}''(y)
=
-H^2I'(yH)
=
\frac{H^2}{-U''(I(yH))}>0.
\]
Evaluation at $y^\star$ proves the displayed curvature. A second-order Taylor expansion and the assumed calibration rate yield the $O(M_{\mathrm{cal}}^{-1})$ expectation under uniform integrability of the remainder.

\subsection[Proof of the multiplier and active-face results]{Proof of \Cref{prop:multiplier-secant,prop:paired-face-identity,cor:operational-active-face}}
\label{proof:operational-active-face}

Concavity of $V_s^j$ gives the supergradient inequality
\[
V_s^j(\varepsilon)-V_s^j(0)
\le
\varepsilon\lambda_{j,s}^\star,
\]
and monotonicity gives the lower inequality in \eqref{eq:multiplier-secant}.

For the paired identity, add and subtract the original-policy utility and use \Cref{thm:positive-gap} conditionally from state $s$:
\begin{align*}
\E_s[Z_{j,s,\varepsilon}^{\rm face}]
&=
J_s(\pi^{j,\varepsilon})-J_s(\pi)
-
\E_s[Z_s^{\rm gap}]\\
&=
J_s(\pi^{j,\varepsilon})-D_s(y,\mu).
\end{align*}
This proves \eqref{eq:paired-sensitivity-expectation}; adding and subtracting $J_s(\pi)$ gives \eqref{eq:resolving-power-identity}.

Since $D_s(y,\mu)\ge V_s^j(0)$, any lower confidence bound $L_{j,s}(\varepsilon)$ for the paired mean satisfies
\[
L_{j,s}(\varepsilon)
\le
V_s^j(\varepsilon)-V_s^j(0)
\le
\varepsilon\lambda_{j,s}^\star
\]
on the confidence event. Taking the positive part and dividing by $\varepsilon$ proves \eqref{eq:operational-multiplier-lb}. A strictly positive lower bound implies a positive optimal multiplier, and complementary slackness yields \eqref{eq:certified-active-face}.

\subsection[Proof of the finite-sample validity and resolution theorem]{Proof of \Cref{thm:resolution-complexity}}
\label{proof:resolution-complexity}

Let $c_{\mathcal R},c_W,U_{\mathcal R}$, and $L_W$ be as in \Cref{thm:resolution-complexity}. Cantelli's inequality and a union bound imply, with probability at least $1-\alpha_{\mathcal R}-\alpha_W$,
\begin{equation}
\label{eq:ec-policy-validity-event}
\mathcal R(y,\mu;\pi)\le U_{\mathcal R},
\qquad
\mathcal W_\rho\ge L_W.
\end{equation}
On this event, whenever $L_W>0$ and $U_{\mathcal R}\le a_{\tau,v}L_W$,
\[
\|\pi-\pi^\star\|_{\bar\rho,v}^2
\le
\frac{2\chi_v\mathcal R(y,\mu;\pi)}{\mathcal W_\rho}
\le
\frac{2\chi_vU_{\mathcal R}}{L_W}
\le
2\chi_v a_{\tau,v}
=
\tau^2.
\]
This proves the validity implication \eqref{eq:policy-valid-trigger}.

For resolution power, apply Cantelli also to the reverse tails. With probability at least $1-\alpha_{\mathcal R}-\alpha_W$,
\begin{equation}
\label{eq:ec-policy-power-event}
\widehat{\mathcal R}_{\rm comp}\le \mathcal R(y,\mu;\pi)+c_{\mathcal R},
\qquad
\widehat{\mathcal W}\ge \mathcal W_\rho-c_W.
\end{equation}
On the intersection of \eqref{eq:ec-policy-validity-event} and \eqref{eq:ec-policy-power-event},
\[
U_{\mathcal R}\le \mathcal R(y,\mu;\pi)+2c_{\mathcal R},
\qquad
L_W\ge \mathcal W_\rho-2c_W,
\]
and therefore
\begin{align*}
a_{\tau,v}L_W-U_{\mathcal R}
&\ge
 a_{\tau,v}\mathcal W_\rho-\mathcal R(y,\mu;\pi)
 -2a_{\tau,v}c_W-2c_{\mathcal R}\\
&=
\Delta_{\tau,v}-2a_{\tau,v}c_W-2c_{\mathcal R}>0.
\end{align*}
On the event \eqref{eq:ec-policy-validity-event}, $U_{\mathcal R}\ge \mathcal R(y,\mu;\pi)\ge0$, where the second inequality follows from weak duality. Since $a_{\tau,v}>0$, the strict inequality above therefore implies $L_W>0$; hence the trigger occurs. A union bound over the four one-sided events gives probability at least $1-2\alpha_{\mathcal R}-2\alpha_W$. When $N_{\mathcal R}=N_W=N$, solving \eqref{eq:policy-resolution-condition} for $N$ gives \eqref{eq:policy-resolution-sample-size}.

For the face family, Cantelli applied separately to both tails gives
\[
\Pr\bigl(|\widehat m_i-m_i|>c_i\bigr)\le2\beta.
\]
A union bound over $M$ triples yields simultaneous events $|\widehat m_i-m_i|\le c_i$ with probability at least $1-2M\beta=1-\alpha$. On this event,
\[
L_i=\widehat m_i-c_i\le m_i,
\qquad
L_i\ge m_i-2c_i>0.
\]
Thus every lower bound is valid and positive, so \Cref{cor:operational-active-face} certifies all target faces. Solving $m_i>2c_i$ for $N_i$ gives \eqref{eq:face-resolution-sample-size}.

\subsection[Proof of the Bellman curvature proposition]{Proof of \Cref{prop:bellman-curvature}}
\label{proof:bellman-curvature}

Because the objective contains terminal utility only, Bellman telescoping under the deployed policy gives
\begin{equation}
\label{eq:ec-bellman-telescope}
V_0^\star-J(\pi)
=
\E_\pi\sum_{k=0}^{K-1}
\left[
V_k^\star(S_k)-Q_k^\star(S_k,\pi_k)
\right].
\end{equation}
For a visited state $s$, strong concavity and \eqref{eq:bellman-hessian-lower} imply
\[
Q_k^\star(s,a)
\le
Q_k^\star(s,a^\star)
+\nabla_aQ_k^\star(s,a^\star)^\top(a-a^\star)
-
\frac12(a-a^\star)^\top\underline{\Curv}_k(s)(a-a^\star).
\]
Since $a^\star=\pi_k^\star(s)$ maximizes $Q_k^\star(s,\cdot)$ over the convex set $\cK$, the constrained first-order variational inequality gives
\[
\nabla_aQ_k^\star(s,a^\star)^\top(a-a^\star)\le0
\qquad\text{for every }a\in\cK.
\]
Taking $a=\pi_k(s)$ and substituting the resulting one-step quadratic lower bound into \eqref{eq:ec-bellman-telescope} proves \eqref{eq:bellman-primitive-growth}.

For the reduced Gaussian CRRA objective, let $\beta=1-\gamma$ and
$u_k=\nabla_a\ell_k$. Direct differentiation gives
\[
\nabla_{aa}^2\mathcal Q_k
=
\E\!\left[
q_{k+1}^\star e^{\beta\ell_k}
\left
\{(\gamma-1)h\Sigma_k
+(\gamma-1)^2u_ku_k^\top\right\}
\mid S_k=s
\right].
\]
The square term is positive semidefinite. Replacing $q_{k+1}^\star$ by its lower bound and the exponential moment by its infimum over $a\in\cK$ proves \eqref{eq:crra-primitive-curvature}. The standard CRRA Bellman normalization then converts this reduced-objective convexity into the lower curvature in \eqref{eq:bellman-primitive-growth}.

\subsection[Proof of the Merton multiplier and leakage bounds]{Proof of \eqref{eq:merton-multiplier-lb}--\eqref{eq:merton-leakage-bound}}
\label{proof:merton-multiplier-leakage}

Let $f_i^\star(\varepsilon)$ denote the optimal Merton growth objective when the $i$th lower bound is relaxed to $-\varepsilon$. Concavity gives
\[
f_i^\star(\varepsilon)-f_i^\star(0)
\le \nu_i^\star\varepsilon.
\]
Because $\CE_i^\star(\varepsilon)=x_0\exp([r+f_i^\star(\varepsilon)]T)$, a relaxed primal CE lower bound and original dual CE upper bound imply
\[
\frac1T\log\frac{\underline C_{i,\varepsilon}}{\overline C_0}
\le
f_i^\star(\varepsilon)-f_i^\star(0),
\]
which proves \eqref{eq:merton-multiplier-lb}. Under no shorting, the $i$th slack term in \eqref{eq:perf-lower} is $\nu_i^\star\pi_{k,i}$. Replacing $\nu_i^\star$ by its positive lower bound, the gap by its upper bound, and the total continuation weight by its lower bound gives \eqref{eq:merton-leakage-bound}.

\subsection[Proof of the performance theorem]{Proof of \Cref{thm:performance}}
\label{proof:performance}

The optimal continuation value is
\[
V_k^\star(x)
=
\frac{a_k^\star x^{1-\gamma}-1}{1-\gamma}.
\]
Conditional on $(X_k,\pi_k)$,
\[
\E_k[(X_{k+1})^{1-\gamma}]
=
X_k^{1-\gamma}
\exp((1-\gamma)[r+f(\pi_k)]h).
\]
Using the Bellman equality at $\pi^\star$, telescoping the one-step losses gives \eqref{eq:perf-exact}. From the KKT relation
\[
b=\gamma\Sigma\pi^\star+\Gamma^\top\mu^\star,
\]
we obtain
\begin{align*}
f(\pi^\star)-f(\pi)
&=
\frac\gamma2(\pi-\pi^\star)^\top\Sigma(\pi-\pi^\star)
+(\mu^\star)^\top(\Gamma\pi^\star-\Gamma\pi)\\
&=
\frac\gamma2(\pi-\pi^\star)^\top\Sigma(\pi-\pi^\star)
+(\mu^\star)^\top(g-\Gamma\pi),
\end{align*}
where the last equality uses complementary slackness. Finally, $e^z-1\ge z$ proves \eqref{eq:perf-lower}.

\subsection[Proof of the operational-radius corollary]{Proof of \Cref{cor:operational-radius}}
\label{proof:operational-radius}

Since $\overline C\ge\CE^\star=x_0e^{[r+f(\pi^\star)]T}$,
\[
\overline r_{\rm CE}\ge r+f(\pi^\star).
\]
Because $1-\gamma<0$, \eqref{eq:lower-coefficient} satisfies $\underline a_k\le a_k^\star$, so $\underline{\mathcal W}\le\mathcal W$. From \eqref{eq:perf-lower},
\[
\overline{\mathcal R}
\ge
\frac\gamma2\E\sum_khw_k
(\pi_k-\pi^\star)^\top\Sigma(\pi_k-\pi^\star).
\]
For every vector $z$,
\[
z_i^2\le(\Sigma^{-1})_{ii}z^\top\Sigma z.
\]
Therefore
\[
\E\sum_khw_k(\pi_{k,i}-\pi_i^\star)^2
\le
\frac{2\overline{\mathcal R}}{\gamma}(\Sigma^{-1})_{ii}.
\]
Divide by $\mathcal W\ge\underline{\mathcal W}$ to obtain \eqref{eq:coord-radius}.

\subsection[Proof of predictable curvature and the restarted operational radius]{Proof of \Cref{prop:predictable-curvature} and \Cref{thm:restarted-operational-radius}}
\label{proof:predictable-curvature-radius}

Let $\beta=1-\gamma$. For fixed $(k,y)$, write
\[
\Lambda_{k+1}(\pi)
=q_{k+1}^\star(Y_{k+1})e^{\beta\ell_k(\pi,y,Z_{1,k})}.
\]
The first and second derivatives of the exponent are
\[
\beta u_k(\pi,y,Z_{1,k}),
\qquad
-\beta\sigma_S^2h=(\gamma-1)\sigma_S^2h.
\]
Differentiating under the conditional expectation proves \eqref{eq:predictable-curvature-exact}. Positivity of the square term and $\mathcal Q_k(y,\pi;q_{k+1}^\star)\ge q_k^\star(y)$ give \eqref{eq:basic-curvature-lower}.

For the retained square term, the state-uniform next-stage bound gives
\[
\E[\mathcal L_{k+1}u_k^2\mid y]
\ge
\underline q_{k+1}^{\rm unif}\E[e^{\beta\ell_k}u_k^2\mid y].
\]
Let $m_0=(b(y)-\pi\sigma_S^2)h$, $s_0=\sigma_S\sqrt h$, and $c_0=\beta\pi\sigma_S\sqrt h$. The Gaussian identity
\[
\E[e^{c_0Z}(m_0+s_0Z)^2]
=e^{c_0^2/2}\{(m_0+s_0c_0)^2+s_0^2\}
\]
yields \eqref{eq:gaussian-square-factor}. Direct differentiation shows that $\partial_\pi\Psi_k(y,\pi)$ has the opposite sign of $b(y)-\gamma\pi\sigma_S^2$, proving \eqref{eq:psi-min}. Strong convexity of the one-step objective on the interval and the variational inequality at the constrained minimizer give a one-step residual equal to one half of the lower Hessian times $(\pi-\pi^\star)^2$. The Bellman performance-difference telescoping divides the reduced residual by $\gamma-1$, producing \eqref{eq:predictable-performance-hybrid} and the coefficient \eqref{eq:hybrid-q-lower}.

For the conditional dual, weak duality at wealth one gives
\[
\CE_k^\star(y)\le m_k(y)^{-\gamma/(\gamma-1)}.
\]
Since $q_k^\star(y)=(\CE_k^\star(y))^{1-\gamma}$ and $1-\gamma<0$, this is \eqref{eq:restart-q-lower}. Substituting the one-sided restart and uniform bounds into \eqref{eq:predictable-performance-hybrid}, bounding the resulting total weight below by $\underline{\mathcal W}_{\rm hyb}$, and the performance loss above by the residual UCB $\overline{\mathcal R}$ proves \eqref{eq:hybrid-restart-radius}. The stated confidence level follows from the joint one-sided event.

\subsection[Proof of the Cantelli policy-radius proposition]{Proof of \Cref{prop:cantelli-policy-radius}}
\label{proof:cantelli-policy-radius}

The sample-mean variances are at most $v_{\rm gap}/N_{\mathcal R}$ and $v_W/N_W$. Applying Cantelli's inequality to $-(\widehat{\mathcal R}_{\rm comp}-\mathcal R)$ and to $\widehat{\mathcal W}-\mathcal W$ gives
\[
\Pr\!\left(
\mathcal R(y,\mu;\pi)>
\widehat{\mathcal R}_{\rm comp}
+
\sqrt{\frac{v_{\rm gap}(1-\alpha_{\mathcal R})}{N_{\mathcal R}\alpha_{\mathcal R}}}
\right)
\le\alpha_{\mathcal R},
\]
and
\[
\Pr\!\left(
\mathcal W<
\widehat{\mathcal W}
-
\sqrt{\frac{v_W(1-\alpha_W)}{N_W\alpha_W}}
\right)
\le\alpha_W.
\]
Since $G(\pi)\le\mathcal R(y,\mu;\pi)$, the union bound gives the joint event. The deterministic residual-to-policy transfer then completes the proof.

\subsection[Proof of policy-distance sharpness and support asymmetry]{Proof of \Cref{lem:raw-policy-certificate}, \Cref{prop:sqrt-sharp}, and \Cref{cor:support-asymmetry}}
\label{proof:raw-policy-sharpness}

By \eqref{eq:abstract-quadratic-growth} and the chain $G(\pi)\le\mathcal R(y,\mu;\pi)\le\overline{\mathcal R}$,
\[
\|\pi-\pi^\star\|_{\rho,\Curv}^2\le2\overline{\mathcal R},
\]
which proves \eqref{eq:raw-ellipsoid}. For each state $s$, Cauchy--Schwarz in the $\Curv(s)$ metric gives
\[
\bigl(v^\top z(s)\bigr)^2
\le
\bigl(v^\top \Curv(s)^{-1}v\bigr)
\bigl(z(s)^\top \Curv(s)z(s)\bigr).
\]
Integrating with respect to $\bar\rho^\pi$ and using \eqref{eq:directional-curvature} yields
\[
\|z\|_{\bar\rho,v}^2
\le
\frac{\chi_v}{\mathcal W_\rho}\|z\|_{\rho,\Curv}^2.
\]
Combining this inequality with $\|\pi-\pi^\star\|_{\rho,\Curv}^2\le2\overline{\mathcal R}$ proves \eqref{eq:directional-radius}.

For the sharpness result, maximize $v^\top z$ subject to $z^\top \Curv z\le2\delta$. The Lagrange first-order condition gives $z=\alpha \Curv^{-1}v$, and the active constraint determines
\[
\alpha=\sqrt{\frac{2\delta}{v^\top \Curv^{-1}v}}.
\]
Substitution yields \eqref{eq:sharp-radius}; changing the sign of $z$ gives the absolute-value maximum.

For a constant policy and $v=e_i$, the projected-policy seminorm reduces to $|\pi_i-\pi_i^\star|$, so the directional bound is the stated coordinate interval. In the state-dependent statement, define
\[
\mathcal B:=\mathcal A\cap\{\pi_i^\star=0\}.
\]
On $\mathcal B$, we have $|\pi_i-\pi_i^\star|\ge\eta$. Therefore
\[
r_{i,\mathcal A}^2
\ge
\int_{\mathcal B}(\pi_i-\pi_i^\star)^2d\rho^\pi
\ge
\eta^2\rho^\pi(\mathcal B),
\]
which proves \eqref{eq:support-failure-mass}.

\section*{Acknowledgments}
Jeonggyu Huh received financial support from the National Research Foundation of Korea (No.~RS-2025-00562904).

\printbibliography
\end{document}